\documentclass[11pt]{article}%
\usepackage[table,dvipsnames]{xcolor}
\usepackage{pgfplots}
\pgfplotsset{compat=newest}
\usepackage{amsfonts}%
\usepackage{amsmath}%
\usepackage{amssymb}%
\usepackage{amsthm}%
\usepackage[english]{babel}%
\usepackage{csquotes}
\usepackage{bm}%
\usepackage{booktabs}%
\usepackage{caption}%
\usepackage{color}%
\usepackage{comment}%
\usepackage{enumerate}%
\usepackage{float}%
\usepackage[T1]{fontenc}%
\usepackage[margin=1in]{geometry}%
\usepackage{graphicx}%
\usepackage{placeins}
\usepackage[breaklinks=true,bookmarksopen=true,colorlinks=true,citecolor=blue,urlcolor=blue,linkcolor = blue]{hyperref}%
\usepackage{listings}%
\usepackage{mathtools}%
\usepackage{paralist}%
\usepackage{pgfplotstable}%
\pgfplotsset{compat=1.12}%
\usepackage[nodisplayskipstretch]{setspace}%
\usepackage[subrefformat=parens]{subcaption}%
\usepackage{tabularx}%
\usepackage{verbatim}%
\usepackage{tikz} 
\usetikzlibrary{positioning, calc} 
\usepackage{arydshln} 
\usepackage{cancel} 
\usepackage{dirtytalk} 
\usepackage[user,titleref]{zref}
\usepackage{ragged2e}
\usepackage{booktabs}  
\usepackage{threeparttable}  
\usepackage{pgf}        
\usepackage{colortbl}
\usepackage{dirtytalk} 
\usepackage{textcomp} 

\newtheorem{theorem}{Theorem}[section]%
\newtheorem{assumption}{Assumption}%
\newtheorem{corollary}{Corollary}[section]%
\newtheorem{definition}{Definition}[section]%
\newtheorem{example}{Example}[section]%

\newtheorem{lemma}{Lemma}[section]%

\newtheorem{proposition}{Proposition}[section]
\newtheorem{remark}{Remark}[section]%

\definecolor{darkgreen}{rgb}{0.0, 0.5, 0.0}

\newcolumntype{C}{>{\centering\arraybackslash}X}%

\newtheorem{innerexample}{Example}
\newenvironment{namedexample}[1]
  {\renewcommand\theinnerexample{#1}\innerexample}
  {\endinnerexample}

\renewcommand{\P}{\mathrm{P}}
\newcommand{\independent}{\mathrel{\mbox{\(\perp\!\!\!\perp\)}}}
\usetikzlibrary{patterns}
\usepackage{pdflscape}
\usepackage[bibstyle=numeric, doi=false, url=false, backend=biber, maxbibnames=100, maxcitenames=5, uniquelist=false, uniquename=false, sorting=nyt, natbib=true, style=authoryear, eprint=false]{biblatex}
\renewbibmacro{in:}{}
\AtEveryBibitem{\clearfield{issn}}
\AtEveryBibitem{\clearfield{month}}
\makeatletter
\newcommand{\neutralize}[1]{\expandafter\let\csname c@#1\endcsname\count@}
\makeatother

\newtheorem{thm}{Assumption}

\newcommand{\E}{\mathbb{E}} 
\newcommand{\Var}{\mathrm{Var}}
\newcommand{\D}{\mathcal{D}} 
\newcommand{\ST}{\mathcal{S}} 
\newcommand{\cst}{ \succ^{\hspace{-0.2em} \scalebox{0.7}{\pmb{+}}} } 
\newcommand{\tightminus}{\mspace{0mu}\scalebox{1.5}[1.0]{-}\mspace{0mu}}
\newcommand{\tightequals}{\mspace{0mu}\scalebox{1}[1.0]{=}\mspace{0mu}}

\newcommand{\mapD}[1]{%
  \begingroup
    \pgfmathparse{#1*0.5}%
    \pgfmathprintnumber[fixed,precision=1,zerofill]{\pgfmathresult}%
  \endgroup
}
\allowdisplaybreaks[3]  

\definecolor{derekcolor}{rgb}{0.05, 0.35, 0.05} 

\usepackage{xr}
\begin{document}
\begin{singlespace}
\title{\textbf{Causal Inference for Aggregated Treatment}\thanks{We thank D\'ebora Mazetto for the data used in the application. We also thank participants in the UGA Econometrics Reading Group, the Georgia Econometrics Workshop, Midwestern Econometrics Group Conference as well as Francesco Agostinelli, James Berry, Josh Kinsler, Marcelo Moreira, V\'itor Possebom, Eli Sellinger-Liebman, Meghan Skira, Alex Torgovitsky, and Emmanuel Tsyawo for helpful comments.}}
\author{
  Carolina Caetano\textsuperscript{$\dagger$}\quad
  Gregorio Caetano\textsuperscript{$\ddagger$}\quad
  Brantly Callaway\textsuperscript{$\S$}\quad
  Derek Dyal\textsuperscript{\textparagraph}
}

\begingroup
  \renewcommand\thefootnote{}
  \footnotetext{
    $\dagger$: \texttt{carol.caetano@uga.edu},
    $\ddagger$: \texttt{gregorio.caetano@uga.edu},
    $\S$: \texttt{brantly.callaway@uga.edu},
    \textparagraph: \texttt{ddyal@uga.edu}. All authors are affiliated with the John Munro Godfrey, Sr.~Department of Economics, University of Georgia.
  }
\endgroup

\maketitle

\end{singlespace}
\maketitle
\begin{abstract}
\begin{singlespace}
We study causal inference when the treatment variable is an aggregation of multiple sub-treatments. Researchers often report marginal effects for the aggregated treatment, implicitly assuming the target parameter corresponds to a well-defined average of sub-treatment effects. We show that, even under ideal conditions such as random assignment, the weights underlying this average have key undesirable properties: they are not unique, can be negative, and these issues become exponentially more likely as the number of sub-treatments increases and their support grows. We propose diagnostics to detect these problems and introduce alternative approaches to circumvent them, depending on whether sub-treatments are observed.
\end{singlespace}
\end{abstract}

\bigskip



\onehalfspacing


\section{Introduction} \label{sec:intro}

Causal inference requires asking questions that are precise enough to correspond to well-defined interventions (\citealt{Rubin2005}). Yet in many applied settings, data limitations and the need to streamline the narrative force researchers to pose causal questions at a level that is too vague to satisfy this requirement. When the treatment is only loosely defined---when it collapses multiple distinct versions into a single category---the causal question itself becomes ambiguous. In such cases, the concern goes beyond identification, pointing to a more fundamental conceptual limitation: the treatment variable may not correspond to any single coherent intervention at all. Yet, for lack of better options, researchers do estimate such causal effects in the hope that they capture something informative. While many implicitly recognize the limitations of this practice, there is no shared language or theoretical framework to precisely articulate the challenges inherent to this approach, or to guide interpretation of the resulting effects. In this paper, we study the consequences of this widespread practice, aiming to provide an early step in this direction.

Specifically, we study the setting where the researcher estimates causal effects using a treatment variable that aggregates multiple versions of the treatment, or \textit{sub-treatments}. In these settings, the true causal drivers---these different sub-treatments that represent different underlying interventions---may be numerous, non-mutually exclusive, and heterogeneous in their effects, while the aggregated treatment variable serves as a summary measure that simplifies analysis but does not itself cause the outcome. In these scenarios, researchers often interpret coefficients on aggregated treatments as marginal causal effects, implicitly assuming that these parameters represent well-defined averages of sub-treatment effects---i.e., specific combinations of precisely defined interventions.

This assumption, however, warrants scrutiny. In this paper, we show that the marginal effects of an aggregated treatment can be difficult to interpret. Even under ideal conditions for causal inference, such as random assignment, the reported treatment effects generally correspond to weighted averages of sub-treatment effects with weights that are not unique and may be negative. Moreover, the negative weights issue becomes more prevalent as the number of sub-treatments or the support of each sub-treatment grows. Both issues emerge because of heterogeneous effects across different sub-treatments. We show that negative weights arise due to \textit{incongruent} comparisons---those where a marginal increase in the level of the aggregated treatment involves a decrease in the level of at least one sub-treatment. The presence of negative weights can, in principle, lead to an aggregated treatment effect that is negative, even when all sub-treatment effects are positive.

The concern with non-unique weights is that different combinations of sub-treatment effects can yield the same estimated aggregate effect. In practice, this means the estimated aggregate effect reflects the impact of one particular mix of interventions, but the specific mix that produced that estimate is not generally identified. Without diagnostic tools to uncover which mix of interventions the data actually reflect, there is a risk that the estimate will be misinterpreted to justify interventions that are not supported by the data. This raises an external validity concern, but of a different kind: it arises from extrapolating the effect from one intervention to another, rather than from the standard problem of extrapolating the effect of a given intervention from one population to another.

If using an aggregated treatment yields an average treatment effect with non-unique and possibly negative weights, what should one do to improve causal inference in such a situation? We offer diagnostics that the researcher can use to convey the extent of incongruency issues in their application. We also propose alternative estimands that avoid these issues altogether. When sub-treatment data is available---that is, when the data is sufficiently rich to plausibly assume that the actual intervention components are observed---researchers can construct weighting schemes that restrict attention to congruent comparisons alone. In contrast, when sub-treatment data is unavailable, we show how researchers in this case can still circumvent all issues above by focusing on certain alternative non-marginal effect parameters.

These aggregation issues arise frequently in practice. Researchers may aggregate treatments due to data limitations, or to increase statistical precision, or to simplify their empirical strategy, or even to streamline narrative exposition. In labor economics, studies of the effects of minority or immigrant share frequently collapse distinct ethnic or national-origin groups into a single category (\citet{CardEtAl2008}, \citet{BohlmarkWillen2020}, \citet{Lowe2021}). Similarly, peer effects research typically defines treatment using a single aggregated statistic of the classroom, such as the average GPA or SAT score (\citet{Sacerdote2001}, \citet{CarrellEtAl2013}). In urban economics, crime measures often aggregate qualitatively different sub-categories such as thefts, robberies, and homicides (\citet{GreenbaumTita2004}, \citet{TitaEtAl2006}, \citet{IhlanfeldtMayock2010}, \citet{MejiaRestrepo2016}), while research in environmental economics regularly treats exposure to natural disasters as a single treatment, despite substantial differences between earthquakes, floods, and wildfires (\citet{BoustanEtAl2020}, \citet{ZhaoEtAl2022}). In health economics, treatment variables often include composite indices of health, behavior, or genetic risk---each bundling multiple sub-components that may interact in complex ways (\citet{FinkelsteinTaubmanEtAl2012}, \citet{AllcottEtAl2019}, \citet{BarthEtAl2020}, \citet{HoumarkEtAl2024}). Time use studies similarly group distinct activities under broad labels like ``exercise,'' ``leisure,'' or ``enrichment'' (\citet{FioriniKeane2014}, \citet{GCaetanoEtAl2019}, \citet{JurgesKhanam2021}, \citet{CCaetanoEtAl2024}). 

In fact, the practice of using an aggregated measure of the treatment variable is so common that we often take it for granted in most empirical work. For instance, consider one of the most widely used treatment variables in applied econometrics: years of schooling. This variable aggregates fundamentally different educational experiences---spanning institutions, curricula, teaching methods, peer groups, and levels of engagement---across the life of the individual. The same value on this variable may correspond to radically different combinations of sub-treatments across individuals. Moreover, increasing the number of years of schooling is not a well-defined intervention, because it never occurs in isolation: it always takes place within a particular context, where specific features of the schooling experience are altered.


A common initial reaction to our theoretical results is that they seem too damning for empirical work in causal inference. If even seemingly straightforward interventions become ill-defined once aggregation issues are taken seriously, it may appear that the very enterprise of causal inference is at risk. A central contribution of this paper, however, is to show that negative-weight concerns can be resolved by adopting a non-marginal interpretation of effects, which remains well-defined under plausible assumptions. Moreover, when marginal effects are unavoidable, another key contribution of the paper is to provide empirical diagnostics that both assess the severity of the problem and clarify the weakest assumptions under which a marginal interpretation remains valid. Finally, by raising awareness of this underappreciated problem and providing tools to improve the interpretation of the results, this paper helps ensure that causal effects from one intervention are interpreted appropriately before being extrapolated to inform decisions about another intervention. This also underscores the value of collecting more detailed datasets that capture features of the aggregated treatment.

The problem we consider is related to the Stable Unit Treatment Value Assumption (SUTVA), a foundational concept that appears throughout the causal inference literature (\citet{Rubin1980}).  We maintain the first part of SUTVA, often referred to as \textit{no interference} or \textit{no contamination}: that potential outcomes for unit $i$ do not depend on the treatment assignments of other units.  The second part of SUTVA, often referred to as \textit{no hidden versions of treatments}, states that different units do not experience different versions of the same treatment, or different interventions (\citet{Rubin1980}, \citet{VanderWeeleHernan2013}, \citet{imbens-rubin-2015}). If one operates at the level of the aggregated treatment, then the issues that we highlight correspond to a violation of the second part of SUTVA, as different combinations of sub-treatments can generate the same aggregate treatment amount but different outcomes.
It is easy to see in the examples above that the potential outcome would not be a well-defined function of the aggregated treatment variable. For instance, the same individual with the same number of years of schooling would likely obtain very different outcomes under different sub-treatments (e.g., graduating from a top-ranked vs.~a lower-ranked university). The second component of SUTVA has not been studied nearly as much as the first part of SUTVA. To the best of our knowledge, the only studies discussing violations of the second condition of SUTVA mainly come from the Epidemiology literature and focus mostly on mediation analysis (\citet{cole2009consistency}, \citet{VanderWeele2009}, \citet{HernanVanderWeele2011}, \citet{LaffersMellace2020}).\footnote{In our setting, the outcome is influenced only by the sub-treatments themselves. The aggregated treatment is merely an aggregated summary of the underlying vector of sub-treatments, rather than a well-defined intervention. This contrasts with mediation settings, where the treatment (analogous to our aggregated treatment) is often a well-defined intervention that affects the outcome through distinct channels or mediators (analogous to our sub-treatments).} One exception is \citet{VanderWeeleHernan2013}, who mostly study mediation, but also consider the context of \textit{ex post} coarsening of the treatment variable (e.g., transforming a multivalued sub-treatment variable into a binary treatment variable), and show that such practice leads to a violation of SUTVA.\footnote{An example of coarsening is discussed in the context of the Tennessee STAR experiment on class size, where \citet{AdusumilliEtAl2025} show that the ``small class'' status corresponded to different actual class sizes across schools, reflecting local implementation constraints.} To the best of our knowledge, no paper has yet considered aggregations beyond coarsening, or provided estimands that are robust to violations of the second part of SUTVA. 

Our paper suggests that one can still identify meaningful causal effects in a scenario where the second component of SUTVA is violated for the treatment variable the researcher uses. However, we require that the underlying sub-treatments satisfy SUTVA---i.e., each sub-treatment is a well defined intervention. This requirement implies that targeting \textit{marginal} estimands relies on an additional assumption: that the sub-treatment observed in the data is sufficiently granular for SUTVA to plausibly hold. Still, researchers can continue to interpret (non-marginal) causal effects meaningfully under violations of SUTVA---provided that sub-treatments satisfying SUTVA are conceptually well-defined, even if they are unobserved.

Our work also relates to the vast literature on treatment effect heterogeneity, which studies how the causal effect of a given treatment can vary across units (e.g., \citet{rubin-1974}, \citet{Holland1986}, \citet{heckman-smith-clements-1997}, \citet{Rosenbaum2002}, \citet{imbens-rubin-2015}). Our findings suggest that when treatments are aggregated, what appears to be treatment effect heterogeneity may instead reflect sub-treatment heterogeneity---that is, heterogeneity arising from distinct underlying components of the treatment itself---which is a violation of the no hidden versions of treatments component of SUTVA. Returning to the years of schooling example, the effect of one additional year of schooling may obscure heterogeneity in educational experiences even for the same individual. For instance, the additional year could reflect a counterfactual enrollment in a five-year major such as engineering, rather than a four-year major like economics. In this case, what looks like a marginal ``year effect'' partially reflects differences in content, difficulty, and the credential ultimately earned. Crucially, this heterogeneity is not across individuals, but within the same individual under different sub-treatments. As a result, the relevant unit of potential outcome variation is not the individual alone, but the individual–sub-treatment pair. Although some empirical work has grappled with these issues---contrasting ``heterogeneous treatments'' and ``heterogeneous treatment effects'' (e.g., \citet{Lechner2002}, \citet{PlescaSmith2007}, \citet{mccall2016government}, \citet{CaetanoMaheshri2018}, \citet{Smith2022})---to our knowledge, there is no formal framework for identifying or interpreting causal effects in such settings. One exception is \citet{HeilerKnaus2025}, which addresses related concerns about heterogeneous treatments by showing how group-level heterogeneity analyses can be misleading in this context, and provides a decomposition to separate causal effect heterogeneity from spurious differences driven by differential assignment to versions of treatment.

We illustrate the issues due to aggregation and our approaches to circumvent them in an application concerning the effects of enrichment activities on children's skills, based on \citet{CCaetanoEtAl2024}. 
The treatment variable---time per week spent on enrichment activities---aggregates time spent on homework, music lessons, and sports, among other extracurricular activities.  This is an application where sub-treatments reflecting detailed activities of the children are observed, which allows us to diagnose how much of an issue incongruency is if the aggregated treatment is used as the treatment variable, as done in that paper. In this application, incongruency is important, as some aggregate marginal effects put at least 30-40\% weight on incongruent comparisons. Estimates for alternative parameters that we propose, which exclude incongruent comparisons, are roughly twice as large in magnitude. 

The paper is organized as follows. Section \ref{sec:SUTVA} describes the causal setting that we consider, establishes the necessary notation, and formally defines both congruency and the relevant target parameters. Section \ref{sec:matt} presents the main challenges for identification and for interpreting these types of parameters. Section \ref{sec:DATE} proposes alternative parameters that rectify the challenges of aggregated treatment, including those that do not require observation of sub-treatment data. 
Section \ref{sec:EmpiricalApplication} delivers an empirical illustration from the time-use literature, highlighting both the problems introduced by aggregated treatment and the solutions we propose. Finally, Section \ref{subsec:Conclusion} offers some concluding remarks. The Appendix provides additional results and proofs.

\section{Aggregated Treatment Setting} \label{sec:SUTVA}

This section (i) provides notation and formalizes the setting that we consider, (ii) introduces notions of congruent and incongruent sub-treatment vectors, and (iii) defines our main target parameters.

\subsection{Notation and Setup} \label{subsec:notation-setup}
We consider a setting where a researcher is interested in understanding the relationship between an outcome and a treatment.  In our application, the outcome is a standard measure of noncognitive skills, and the treatment variable is a measure of time the child spends on enrichment activities.  We denote the outcome variable by $Y$.  The treatment variable is comprised of different types of enrichment activities, which can vary in amount across different units.  Let $S_{ik}$ denote the amount of component $k$ of the treatment that unit $i$ experiences. We refer to $S_{k}$ as the $k$th sub-treatment, and $\mathfrak{S}_{k}$ denotes the support of the $k$th sub-treatment. For example, if ``doing homework'' is the $k$th version of enrichment activity, then $S_{ik}=2$ for children who do two hours of homework.  Next, define $S_i = (S_{i1}, S_{i2}, \ldots, S_{iK})$ where $K$ denotes the total number of versions of the treatment.  We refer to $S_i$ as a unit's sub-treatment vector. Let $\mathcal{S}$ denote the support of $S$. We consider the case where the sub-treatments are discrete and share a common support---we consider this case to focus the exposition of the paper and note that both of these conditions could be relaxed without substantively changing our results below.    

We also define the aggregated treatment variable $D_i = A(S_i)$ where $A(s)$ is an aggregation function that maps sub-treatments to a scalar value of the aggregated treatment variable.  Let $\mathcal{D}$ denote the support of $D$ and $\mathcal{D}_{>0} = \mathcal{D} \setminus \{0\}$ denote the support of $D$ excluding $D=0$.  To keep the discussion concrete, we often focus on the case where $A(s) = \sum_{k=1}^K s_k$---where the aggregated treatment variable adds up all of the underlying components of the sub-treatment vector. However, other types of aggregation are possible as well; some of our results hold immediately for any aggregation scheme, while others would require minor modifications. Two important, immediate properties of the aggregated treatment are that: (i) it is fully determined by the sub-treatments, and (ii) different sub-treatments can lead to the same value of the aggregated treatment.

The discussion above provides a definition of aggregated treatment. Based on the aforementioned properties of aggregation, we define the set
\begin{align*}
    \ST_d &:= \{ s \in \ST : A(s) = d \}
\end{align*}
where $A(\cdot)$ is the aggregation rule.  Thus, $\ST_d$ is the set of distinct sub-treatment vectors that lead to a particular value of the aggregated treatment ($D=d$).  We refer to $\mathcal{S}_d$ as the \textit{aggregation set} corresponding to aggregated treatment $d$. We use the terminology \textit{sub-treatment group $s_d$} to refer to the set of units that experience sub-treatment vector $s_d$.

\begin{example} \label{ex:enrichment} To fix ideas in the discussion below, consider a simplified version of our application where there are three different types of enrichment activities: (1) homework, (2) music lessons, and (3) sports.  Children can participate in any of these enrichment activities. To simplify the discussion in the example, here we binarize each sub-treatment by only keeping track of whether or not the child participates in each enrichment activity, although the theory for our paper additionally allows for sub-treatments to be multivalued. The aggregated treatment $D$ indicates the total number of enrichment activities that a child does.  Children with the same value of the aggregated treatment, however, can experience different combinations of sub-treatments.  For example, one student who does homework and music lessons, and another student who does homework and sports, both participate in two enrichment activities. We can define $\mathcal{S}_d$ for any possible value of $d$,
\begin{align*}
    \mathcal{S}_0 = \{(0,0,0)\}~~\mathcal{S}_1 = \{(1,0,0), (0,1,0), (0,0,1) \}~~\mathcal{S}_2 = \{(1,1,0), (1,0,1), (0,1,1)\}~~\mathcal{S}_3 = \{(1,1,1)\},
\end{align*}
where, for example, $(1,1,0) \in \mathcal{S}_2$ indicates the sub-treatment vector of participating in homework and music lessons but not playing sports.
\end{example}
Because we are interested in causal effects, we also define potential outcomes. Let $Y_i(s)$ denote the outcome that unit $i$ would experience under sub-treatment vector $s$.  The observed outcome is equal to the potential outcome corresponding to the observed sub-treatment vector; that is, $Y_i = Y_i(S_i)$.  Implicit in this expression is that we impose SUTVA (parts 1 and 2) at the level of the sub-treatment---we formalize this in Assumption \ref{ass:sutva2} below. Notice that potential outcomes are only defined at the sub-treatment level. We do not define potential outcomes in terms of the aggregated treatment variable, as the aggregated treatment variable itself is non-causal and contingent upon the aggregation scheme.

Our reading of the literature suggests that common empirical practice is to regress the outcome on the aggregated treatment $D$; i.e., 
\begin{align} \label{eqn:reg}
    Y_i = \alpha_0 + \alpha_1 D_i + e_i
\end{align} 
and to interpret $\alpha_1$, the coefficient on $D$, in terms of marginal effects.\footnote{A representative example comes from \citet{MejiaRestrepo2016}, which studies the effects of property crime (their treatment variable) on different types of household expenditure.  They construct an aggregate measure of property crime by averaging the rate of robberies and burglaries (their sub-treatments).  Robberies and burglaries are distinct crimes.  The main difference is roughly that robberies involve directly stealing from someone while burglaries involve entering a structure to steal (see \citet{fbi-burglary-2018,fbi-robbery-2018} for more details).  Their main results involve interpreting coefficients on this aggregated treatment variable in terms of marginal effects.  For example, they write: ``conditional on controls, the coefficient of crime on total visible and non-steal-able consumption is negative and significant at the 5\% confidence level. In particular, a 10\% increase in property crime is associated with a 1.45\% decline in the consumption of visible and non-stealable goods...''}  Writing $D_i$ in terms of sub-treatments in the above regression, we have that
\begin{align*}
    Y_i = \alpha_0 + \alpha_1 (S_{i1} + S_{i2} + \cdots + S_{iK}) + e_i 
\end{align*}
which suggests that homogenous effects across different sub-treatments is an important implicit assumption for this regression to be able to recover the marginal effect of the sub-treatments on the outcome.  One of our main goals below is to understand how to interpret marginal effects of $D$ in settings where there can be heterogeneous effects of the sub-treatments.

\subsection{Congruent and Incongruent Sub-treatment Vectors} 

The notion of a marginal effect is more complicated in applications with sub-treatments.  In this section, we distinguish between \textit{congruent} and \textit{incongruent} sub-treatment vectors, which is then useful for precisely defining marginal effect parameters in the next section.  Define the marginal set, i.e., the set of neighboring aggregation sets, indexed by $d$, as  
\begin{align*}
    \mathcal{M}(d) &:= \{ (s_d, s_{d-1}) \in \ST_d \times \ST_{d-1} \}
\end{align*}
for all $d \in \D_{>0}$.  That is, $\mathcal{M}(d)$ represents the marginal set of $K$-tuple sub-treatment vector pairs whose $L_1$ norm equals either $d$ or $d-1$. 

\begin{definition}[Congruent and Incongruent Sub-treatment Vectors] \label{def:congruent-subtreatments} For the sub-treatment vectors $(s_d,s_{d-1}) \in \mathcal{M}(d)$, define the binary congruence relation $\cst$ as: $s_d \cst s_{d-1} \; \mathrm{if} \; 
s_d=s_{d-1} + 1_k$ for some $k$, where $1_k$ is the unit vector with $k^{th}$ element equal to one and zero otherwise. If $s_d \cst s_{d-1}$, then we say that $s_d$ and $s_{d-1}$ are congruent; otherwise, we say that they are incongruent.  \end{definition}

Definition \ref{def:congruent-subtreatments} defines congruent and incongruent sub-treatment vectors.  In particular, two sub-treatment vectors $s_d$ and $s_{d-1}$ are congruent if they correspond to neighboring aggregation sets, and the value of each element of the sub-treatment vector $s_d$ is equal to the value of the corresponding element of vector $s_{d-1}$, except for one element. Sub-treatment vectors $s_d$ and $s_{d-1}$ are incongruent if they are from neighboring aggregation sets, but the value of more than one element is different.

\begin{namedexample}{\ref*{ex:enrichment} (continued)} Consider the sub-treatment $(1,0,0) \in \mathcal{S}_1$ (i.e., this is the sub-treatment that amounts to doing homework but not doing music lessons or sports).  $(1,1,0)$---doing homework and music lessons but not sports---is congruent with $(1,0,0)$.  $(1,0,1)$---doing homework and sports but not music lessons---is also congruent with $(1,0,0)$.  $(0,1,1)$---doing music lessons and sports but not homework---is incongruent with $(1,0,0)$.
\end{namedexample}

It is also helpful to define the sets of congruent and incongruent sub-treatment vectors.  In particular, define
\begin{align*}
    \mathcal{M}^{+}(d) &:= \{ (s_d,s_{d-1}) \in \mathcal{M}(d) \; : s_d \cst s_{d-1} \} \; 
\end{align*}
for all $d \in \D_{>0}$, where $\mathcal{M}^{+}(d)$ represents the set of congruent sub-treatments vectors, and 
$\mathcal{M}^{-}(d) := \mathcal{M}(d) \setminus \mathcal{M}^+(d)$, which  represents the set of incongruent sub-treatment vectors. 

\subsection{Target Parameters} \label{subsec:marginal-target-parms}
This section defines the main target parameters that we consider in the paper.  We primarily focus on different weighted averages of marginal changes across sub-treatment vectors, since it is common empirical practice in this setting to interpret or report the coefficient on the aggregated treatment variable in terms of marginal effects.
First, for $(s_d,s_{d-1}) \in \mathcal{M}(d)$, define the \textit{marginal average treatment effect on the treated} ($\textrm{MATT}$) as: 
\begin{align*}  
    \textrm{MATT}(s_d, s_{d-1}) := \E[Y(s_d) - Y(s_{d-1}) | S=s_d] 
\end{align*}
which is the causal effect of moving from sub-treatment vector $s_{d-1}$ to $s_d$ for sub-treatment group $s_d$.\footnote{We primarily focus on on-the-treated type parameters because it is simpler to provide natural weighting schemes for some of the aggregated parameters that we consider in this section relative to unconditional parameters. In addition, unconditional parameters also tend to require stronger identification assumptions in our setting.  That said, extending our arguments to target unconditional parameters seems straightforward.} $\textrm{MATT}(s_d,s_{d-1})$ is defined for all $(s_d,s_{d-1})$, regardless of whether or not $s_d$ and $s_{d-1}$ are congruent.  Sometimes we use the notation $\textrm{MATT}^+(s_d,s_{d-1})$ to indicate a disaggregated marginal average treatment effect on the treated of congruent sub-treatments $(s_d,s_{d-1}) \in \mathcal{M}^{+}(d)$.  Likewise, we sometimes use the notation $\textrm{MATT}^-(s_d,s_{d-1})$ to indicate a disaggregated marginal average treatment effect on the treated of incongruent sub-treatments $(s_d,s_{d-1}) \in \mathcal{M}^{-}(d)$.

Given our interest in aggregation, next we introduce a parameter that is a weighted average of congruent $\textrm{MATT}^+$'s:
\begin{align} \label{eqn:amatt-tilde}
    \textrm{AMATT}^{+}_{w^+}(d) := \sum_{(s_d,s_{d-1}) \; \in \; \mathcal{M}^+(d)} w^+(s_d,s_{d-1}) \cdot \textrm{MATT}^+(s_d,s_{d-1}) 
\end{align}
where $w^+$ is some weighting function that satisfies $w^+(s_d,s_{d-1}) \geq 0$ for any $(s_d,s_{d-1}) \in \mathcal{M}^+(d)$, and $\displaystyle \sum_{(s_d,s_{d-1}) \; \in \; \mathcal{M}^+(d)} w^+(s_d,s_{d-1}) = 1$.  Following the terminology of \citet{blandhol-bonney-mogstad-torgovitsky-2025}, we refer to parameters like $\textrm{AMATT}^{+}_{w^+}(d)$ that are weighted averages (with all non-negative weights) of $\textrm{MATT}^+(s_d,s_{d-1})$ as \textit{weakly causal}.\footnote{In Section \ref{subsec:identification-with-observed-subtreatments}
, we consider more specific aggregated parameters (i.e., with a specific weighting scheme rather than allowing for any weighting scheme satisfying the criteria above).  However, we note here that, with an aggregated treatment, this can introduce some additional complications.  Therefore, we defer this discussion to later in the paper.}

Finally, although our main interest is in the effects of congruent sub-treatments, in some cases, a researcher may be interested in the effects of different sub-treatment vectors for a fixed amount of aggregated treatment. We refer to these as \textit{substitution average treatment effects on the treated} ($\textrm{SATT}$).\footnote{$\textrm{SATT}$'s have a precise mathematical definition---a local tradeoff between two sub-treatments at a fixed amount of aggregated treatment. We call this a substitution effect as, in many applications (e.g., our application on enrichment activities), this parameter coincides with an intuitive notion of substituting between different sub-treatments.  However, this intuition may not apply in all applications (e.g., it is unnatural to think of ``substitution'' in the natural disaster applications mentioned in the introduction); still, regardless of the exact terminology for a particular application, these types of terms continue to be relevant for our results below.} For $s_d$ and $s_d'$ both in $\mathcal{S}_d$, define 
\begin{align*} 
    \textrm{SATT}(s_d,s_d') := \E[Y(s_d) - Y(s_d') | S=s_d] 
\end{align*}
such that $s_d = s_d' + 1_j - 1_l$, where $1_j$ and $1_l$ denote the unit vector for the $j^{th}$ and $l^{th}$ coordinates. In other words, there is a unit exchange between the $j^{th}$ and $l^{th}$ sub-treatments from $s_d$ to $s_d'$. Later, we show that there is often an interesting connection between incongruent $\textrm{MATT}$'s and $\textrm{SATT}$'s.  


\section{Challenges to Identification} \label{sec:matt}

This section outlines several important difficulties that arise in applications with aggregated treatments, even under otherwise ideal conditions for causal inference. The discussion in this section is geared towards interpreting the marginal effect of the aggregated treatment on the outcome in terms of the underlying sub-treatments and the complications that this can induce.  The results are most relevant for applications where the sub-treatments themselves are not observed, but would continue to apply in applications where the researcher observes the sub-treatments yet still decides to use an aggregated treatment.  Many of the expressions below include terms that condition on the sub-treatment group---if the sub-treatments are not observed, then these terms would not be identified, though they are still useful to consider as underlying building blocks of the aggregate marginal effect.

\subsection{Causal Framework} \label{subsec:causal-framework}

We begin by formalizing what we mean by sub-treatments. In our paper, the key difference between the sub-treatment vector and the aggregated treatment is that the sub-treatments satisfy the second part of SUTVA, often referred to as \textit{no hidden versions of treatment}, while the aggregated treatment does not (see \citet{Rubin1980}, \citet{RobinsGreenland2000}, \citet{VanderWeele2009}, \citet{HernanVanderWeele2011}, \citet{imbens-rubin-2015}, and \citet{Hernan2016} for more discussion about SUTVA).  In particular, we make the following assumption.
\begin{assumption}[No Hidden Versions of Sub-treatments relative to $S$] \label{ass:sutva2}
    If unit $i$ experiences sub-treatment vector $s$, then its observed outcome equals the potential outcome corresponding to that sub-treatment vector; i.e., for all $s \in \mathcal{S}$,
    \begin{align*}
        S_i = s \implies Y_i = Y_i(s), \quad \text{where $Y_i(s)$ is a well-defined function of $s$.}
    \end{align*}
\end{assumption}
Throughout the paper, we maintain that the sub-treatments satisfy Assumption \ref{ass:sutva2}.  $Y_i(s)$ being a well-defined function of $s$ rules out ``hidden sub-versions'' of the sub-treatments.  This condition implies that knowing $s$ pins down a unit's potential outcome from experiencing that sub-treatment.\footnote{Note that this assumption represents SUTVA entirely (for multivariate treatment $S$), since it also implicitly assumes the first part of SUTVA---it rules out the possibility that the treatment level of other units $j\neq i$ affects the potential outcome of unit $i$.}  In our running example, where the sub-treatments are homework, music, and sports, it says that further sub-dividing the sub-treatments would not change the potential outcomes; for example, the potential outcomes from doing 2 hours of homework with mom, playing 1 hour of piano, and playing 1 hour of soccer are the same for all units as the potential outcomes from doing 2 hours of homework with dad, playing 1 hour of guitar, and playing 1 hour of basketball.  On the other hand, in our paper, SUTVA for the aggregated treatment is generally violated in that $Y_i(d)$ is not a well-defined function of $d$---in our example, if $D_i=d$, a child's potential outcome still depends on which versions (homework, music, or sports) the child actually experienced.  

Throughout the paper, we maintain that the $K$ researcher-specified sub-treatments satisfy Assumption \ref{ass:sutva2}. In Appendix \ref{sec:define-subtreatment}, we provide some practical guidance for defining the sub-treatments in a particular application.

Next, we introduce the primary assumption for identifying various aggregated and disaggregated average treatment effect parameters. 

\begin{assumption}[No Selection] \label{ass:unconfoundedness} For any $s \in \mathcal{S}$, $Y(s) \independent S$.
\end{assumption}

Assumption \ref{ass:unconfoundedness} states that the distribution of potential outcomes of any sub-treatment vector is independent of the actual sub-treatments experienced.  It would hold by construction if the sub-treatments were randomly assigned.  The main implication of this assumption is that, for any two sub-treatments $s,s' \in \mathcal{S}$, $\E[Y(s) | S=s'] = \E[Y(s) | S=s] = \E[Y|S=s]$, which would be identified in applications where the sub-treatments are observed.  
Thus, it immediately follows that 
\begin{align*}
    \textrm{MATT}(s_d,s_{d-1}) = \E[Y|S=s_d] - \E[Y|S=s_{d-1}] \quad \text{and} \quad \textrm{SATT}(s_d,s_d') = \E[Y|S=s_d] - \E[Y|S=s_d']
\end{align*}
under Assumption \ref{ass:unconfoundedness} (see Propositions \ref{prop:comparison-of-subtreatments} and \ref{prop:comparison-of-neighborly-subtreatments} in Appendix \ref{app:proofs}). 

Although this assumption is likely to be strong in many applications, we take it as a natural baseline for our setting---it provides a starting point that is as favorable as possible for causal inference, and, hence, allows us to emphasize issues related to aggregation in the discussion below. Many of our results are expressed in terms of causal effect parameters which use Assumption \ref{ass:unconfoundedness}; however, absent Assumption \ref{ass:unconfoundedness}, versions of the issues related to aggregation that we highlight continue to apply, just without a causal interpretation. Moreover, extending our arguments to other identification strategies seems straightforward, at least in some leading cases.  For example, under selection on observables, all of our identification results would go through, conditional on covariates.  Similarly, our arguments can be extended to difference-in-differences identification strategies by replacing the level of the outcome with the change in outcomes over time in the assumptions and results in this section.  Extending our results to other settings (e.g., instrumental variables, regression discontinuity, or bunching) might introduce additional complications, but we conjecture that versions of the issues that we point out stemming from aggregated treatment would continue to apply.  Finally, we note that our results on marginal effects go through under a weaker, local version of Assumption \ref{ass:unconfoundedness}, which we discuss in more detail in Appendix \ref{sec:minimal-assumptions}.

\subsection{A Decomposition of Marginal Effects with Aggregated Treatment} \label{sec:identification-with-unobserved-subtreatments}

This section contains one of our main results, which is a decomposition of the marginal effect of the aggregated treatment in terms of $\textrm{MATT}$ parameters.

\begin{theorem}\label{thm:causal-aggregate-marginal-comparison}
    Under Assumptions \ref{ass:sutva2} and \ref{ass:unconfoundedness}, and for any weighting function $w(s_d,s_{d-1})$ such that
\begin{itemize}
    \item [(i)] $\displaystyle \sum_{s_d \in \mathcal{S}_d} w(s_d,s_{d-1}) = \P(S=s_{d-1}|D=d-1)$,
    \item [(ii)] $\displaystyle \sum_{s_{d-1} \in \mathcal{S}_{d-1}} w(s_d,s_{d-1}) = \P(S=s_d|D=d)$,
\end{itemize}
    \begin{align*}
        \E[Y|D=d] - \E[Y|D=d-1] &= \sum_{(s_{d-1}, s_{d}) \in \mathcal{M}^+(d) } w(s_d,s_{d-1}) \cdot \mathrm{MATT}^+(s_d,s_{d-1}) \\
        & \hspace{10mm} + \sum_{(s_{d-1}, s_{d}) \in \mathcal{M}^-(d) } w(s_d,s_{d-1}) \cdot \mathrm{MATT}^-(s_d,s_{d-1})
    \end{align*}
\end{theorem}

Theorem \ref{thm:causal-aggregate-marginal-comparison} shows that, under Assumptions \ref{ass:sutva2} and \ref{ass:unconfoundedness}, the change in the mean of $Y$ given a one unit increase in the aggregated treatment $D$ can be decomposed into a weighted average of congruent and incongruent $\textrm{MATT}$ parameters.  There are several notable features of this decomposition that warrant further examination in the sections that follow.  First, in Section \ref{sec:incongruent-matts}, we show that the incongruent $\textrm{MATT}^-$ parameters that appear in the proposition can be difficult to interpret.  Second, the weights that satisfy the criteria in Theorem \ref{thm:causal-aggregate-marginal-comparison} are non-unique, and, given the difficulty of interpreting incongruent comparisons, ideally, we would like there to be a compatible weighting scheme that does not put any weight on these incongruent comparisons.\footnote{The non-uniqueness of the weights in Theorem \ref{thm:causal-aggregate-marginal-comparison} is conceptually different from all decompositions that we are aware of in econometrics that show up in other contexts, such as continuous treatments (e.g., \citet{yitzhaki-1996} and \citet{callaway-goodman-santanna-2025}), regressions that include covariates (e.g, \citet{angrist-1998}, \citet{sloczynski-2022}, and \citet{hahn-2023}), two-stage least squares (e.g., \citet{ishimaru-2024} and \citet{blandhol-bonney-mogstad-torgovitsky-2025}), and fixed effects regressions (e.g., \citet{chaisemartin-dhaultfoeuille-2020}, \citet{goodman-bacon-2021}, \citet{sun-abraham-2021}, \citet{caetano-callaway-2024}).  Non-uniqueness arises because we decompose the marginal effect of a more aggregated variable in terms of the marginal effects of less aggregated variables, which is an important difference relative to all of the aforementioned papers.}  With this in mind, Section \ref{sec:non-unique-weights} (i) shows that the number of incongruent comparisons grows rapidly in the number of distinct sub-treatments relative to the number of congruent comparisons; (ii) provides conditions under which it is guaranteed that there exists a valid weighting scheme that puts no weight on incongruent comparisons; (iii) characterizes settings where putting weight on incongruent comparisons is unavoidable; and (iv) shows how to test whether any of the aggregation issues mentioned in this paper are empirically relevant.

\begin{remark}[Descriptive Decomposition] 
    In Proposition \ref{prop:marginal-decomposition} in Appendix \ref{app:proofs}, we provide a non-causal version of Theorem \ref{thm:causal-aggregate-marginal-comparison} that does not invoke Assumption \ref{ass:sutva2} or \ref{ass:unconfoundedness} (in fact, this is the key step in proving Theorem \ref{thm:causal-aggregate-marginal-comparison}).  Thus, even if the researcher views  $\E[Y|D=d]-\E[Y|D=d-1]$ descriptively, the issues that we highlight coming from $D$ being an aggregation of the sub-treatments continue to apply.
\end{remark}

\begin{remark}[Regression] \label{rem:regression-marginal}
    As discussed above, it is common in empirical work to estimate a regression like the one in Equation \eqref{eqn:reg} that includes an aggregated treatment variable.  In Proposition \ref{prop:regression-with-marginal-building-block} in Appendix \ref{app:additional-results}, we show that $\alpha_1$, the coefficient on $D$ in the regression in Equation \eqref{eqn:reg}, can be expressed as 
    \begin{align*}
        \alpha_1 &= \sum_{d=1}^{\Bar{N}} \omega^{reg}(d) \cdot \Big( \E[Y|D=d] - \E[Y|D=d-1] \Big) 
    \end{align*}
    where the regression weights are
    \begin{align*}
        \omega^{reg}(d) := \frac{\big(\E[D|D \geq d] - \E[D]\big) \cdot \P(D \geq d)}{\Var(D)}
    \end{align*}
    and satisfy the properties: (i) $\omega^{reg}(d) \geq 0$ for all values of $d \in \D_{>0}$, and (ii) $\displaystyle \sum_{d=1}^{\bar{N}} \omega^{reg}(d) = 1$, where $\bar{N}$ is the maximum on the support of $D$. The result above essentially holds using a discrete version of the argument in \citet{yitzhaki-1996}.  The regression weights, $\omega^{reg}(d),$ have reasonable though less than ideal properties (see Appendix \ref{sec:regressions-with-marginal-building-blocks} for more details); more importantly, however, $\alpha_1$ is a weighted average of $\E[Y|D=d]-\E[Y|D=d-1]$---the same aggregate marginal effect that we decomposed in Theorem \ref{thm:causal-aggregate-marginal-comparison}.  Thus, all of the issues about incongruent $\mathrm{MATT}^{-}$'s and non-unique weights discussed below continue to apply when using regressions that include an aggregated treatment variable.
\end{remark}

\subsection{Interpreting Incongruent Comparisons} \label{sec:incongruent-matts} 

The decomposition in Theorem \ref{thm:causal-aggregate-marginal-comparison} in the previous section showed that $\E[Y|D=d] - \E[Y|D=d-1]$, the marginal effect of the aggregated treatment, included the incongruent $\textrm{MATT}^{-}(s_d,s_{d-1})$ parameters.  How should incongruent marginal causal effect parameters be interpreted?  This section provides two answers to this question.  First, it shows that these incongruent parameters can be expressed in terms of a congruent $\textrm{MATT}^+$ parameter and a sequence of substitution effects, the $\textrm{SATT}$ parameters discussed above.  Second, it shows that incongruent parameters can be expressed as a sequence of congruent $\textrm{MATT}^+$ parameters, but that almost half of the $\textrm{MATT}^+$ parameters in this sequence enter with negative weights.  In either case, it implies that $\textrm{MATT}^-(s_d,s_{d-1})$ is hard to interpret.

\subsubsection*{Incongruent Comparisons and Substitution Effects}

The following proposition re-expresses incongruent $\textrm{MATT}^{-}$ parameters in terms of a congruent $\textrm{MATT}^+$ parameter and a path-dependent sum of substitution effects.

\begin{proposition} \label{prop:incongruent-and-substitution-effects} Under Assumptions \ref{ass:sutva2} and \ref{ass:unconfoundedness}, for all $(s_d, s_{d-1}) \in \mathcal{M}^-(d)$ and for any $s_{d-1}'$ that is congruent with $s_d$, it holds that 
\begin{align*}
    \mathrm{MATT}^-(s_d,s_{d-1}) &= \mathrm{MATT}^+(s_d,s_{d-1}') + \sum_{ b=0 }^{B-1} \mathrm{SATT}(s^{(b)}_{\phi, \; d-1},s^{(b+1)}_{\phi, \; d-1}) 
\end{align*}
where $\phi := (x^{(0)}, \ldots, x^{(B)}) \in \mathcal{C}(s_{d-1}',s_{d-1})$ represents a particular set of chained vectors from the set of sets of chained sub-treatment vectors that create a unit-exchange pathway between sub-treatment vectors $s_{d-1}'$ and $s_{d-1}$ within the same aggregation set $\ST_{d-1}$; and $s^{(b)}_{\phi}, s^{(b+1)}_{\phi}$ are linked sub-treatment vectors that belong to the chain $\phi$.\footnote{Generally, $\mathcal{C}(s'_d,s''_d) := \bigl\{(x^{(0)}, \dots, x^{(B)}) \big| x^{(0)} = s_d',\, x^{(B)} = s_d'',\, \text{for }\,b=0,\dots,B-1: x^{(b)} \in \ST_d, \; \|x^{(b+1)} - s_d'' \|_1 < \|x^{(b)} - s_d''\|_1 \bigr\}$, for any $d \in \D_{>0}$ and some $B \in \mathbb{N}$, is the set of chained sub-treatment vectors that create a unit-exchange pathway between sub-treatment vectors $s_{d}'$ and $s_{d}''$ which belong to the same aggregation set $\ST_{d}$.} 
\end{proposition}

Proposition \ref{prop:incongruent-and-substitution-effects} shows that incongruent $\textrm{MATT}^-$'s can be decomposed into alternative congruent causal effect parameters and substitution effects. This shows that the aggregate treatment effect is composed of both marginal sub-treatment effects and substitution effects. Although substitution effects could be of interest in their own right, they are a different type of parameter from $\textrm{MATT}^+$; they involve substituting across sub-treatments rather than a marginal increase in one of the sub-treatments.

\begin{namedexample}{\ref*{ex:enrichment} (continued)}
Suppose that $s_2 = (1,1,0)$, and $s_1 = (0,0,1)$, which are incongruent, and consider $s_1' = (1,0,0)$.  Then, using the argument from Proposition \ref{prop:incongruent-and-substitution-effects}, it holds that $\mathrm{MATT}^-(s_2,s_1) = \mathrm{MATT}^+(s_2,s_1') + \mathrm{SATT}(s_1',s_1)$.  Or, in other words, the incongruent causal effect of both doing homework and music lessons relative to playing sports can be decomposed into (i) the congruent causal effect of both doing homework and music lessons relative to only doing homework and (ii) the substitution effect of doing homework relative to playing sports.
\end{namedexample}

By plugging the result of Proposition \ref{prop:incongruent-and-substitution-effects} into the decomposition of $\E[Y|D=d]-\E[Y|D=d-1]$ in Theorem \ref{thm:causal-aggregate-marginal-comparison}, it follows that the aggregate marginal effect is hard to interpret because it includes a mix of congruent $\textrm{MATT}^{+}$'s and substitution effects---two different types of parameters.  And, for example, a positive value of $\E[Y|D=d]-\E[Y|D=d-1]$ could be mainly driven by substitution effects rather than effects of marginal increases in any of the sub-treatments.\footnote{The discussion in this paragraph is conceptually related to the well-known decomposition in \citet{goodman-bacon-2021} that, in the context of difference-in-differences identification strategies estimated using two-way fixed effects regressions, relates the coefficient of a binary treatment variable to two different types of causal effect parameters: (i) causal effects of the treatment itself and (ii) treatment effect dynamics.  This is typically taken as a negative result for the two-way fixed effects regression, not because treatment effect dynamics are inherently uninteresting to study, but rather because mixing together two different types of parameters is hard to interpret.  Similarly, in our context, the $\mathrm{SATT}$ parameters could be interesting to learn about, but they do not involve a marginal increase in any sub-treatment and, hence, make the aggregate marginal effect difficult to interpret.}

\subsubsection*{Incongruent Comparisons and Negative Weights}

The next proposition decomposes both substitution effects and incongruent $\textrm{MATT}^-$'s into a path-dependent sum of congruent causal effect parameters.  

\begin{proposition} \label{prop:substitution-effect-decomp} Under Assumptions \ref{ass:sutva2} and \ref{ass:unconfoundedness}, for any $s_{d-1}, s_{d-1}' \in \ST_{d-1}$ such that $s_{d-1} = s_{d-1}' + 1_j - 1_l$, where $1_j$ and $1_l$ are unit vectors for coordinates $j$ and $l$, and for any $s_{d}' \in \ST_{d}$ that is congruent with both $s_{d-1}$ and $s_{d-1}'$, it holds that 
\begin{align*}
    \mathrm{SATT}(s_{d-1},s_{d-1}') &= \mathrm{MATT}^{+}(s_d',s_{d-1}') - \mathrm{MATT}^{+}(s_{d}',s_{d-1}) 
\end{align*}
Moreover, for any incongruent parameter and chain $\phi \in \mathcal{C}(s_{d-1}',s_{d-1})$: 
\begin{align*}
    \mathrm{MATT}^-(s_d,s_{d-1}) &= \mathrm{MATT}^+(s_d,s_{d-1}') + \sum_{ b=0 }^{B-1} \mathrm{MATT}^{+}(s_{\phi, \; d}^{(b)},s_{\phi, \; d-1}^{(b)}) - \sum_{ b=0 }^{B-1} \mathrm{MATT}^{+}(s_{\phi, \; d}^{(b)},s_{\phi, \; d-1}^{(b+1)})  
\end{align*}
\end{proposition}
The first part of Proposition \ref{prop:substitution-effect-decomp} says that any substitution effect between two sub-treatment vectors that share a unit exchange in treatment is equivalent to the difference between two congruent $\textrm{MATT}^+$'s. The second part says that an incongruent $\textrm{MATT}^-$ parameter can be decomposed into congruent $\textrm{MATT}^+$ parameters but that almost half of the $\textrm{MATT}^+$ parameters in this decomposition are included with a negative sign; this part follows from plugging the expression for $\mathrm{SATT}(s_{d-1},s_{d-1}')$ in the first part into Proposition \ref{prop:incongruent-and-substitution-effects} above.

\begin{namedexample}{\ref*{ex:enrichment} (continued)}
Resuming the example from the previous section, let $s_2 = (1,1,0)$, $s_2' = (1,0,1)$, $s_1 = (0,0,1)$, and $s_1' = (1,0,0)$.  Using the argument in the first part of Proposition \ref{prop:substitution-effect-decomp}, it holds that $\; \mathrm{SATT}(s_1', s_1) = \mathrm{MATT}^+(s_2',s_1) - \mathrm{MATT}^+(s_2',s_1')$.  In words, the substitution effect of doing homework relative to playing sports is equal to the difference between (i) the congruent causal effect of doing homework and playing sports relative to only playing sports and (ii) the congruent causal effect of doing homework and playing sports relative to only doing homework.  

Using the argument from the second part of the proposition, it holds that $\mathrm{MATT}^-(s_2,s_1) = \mathrm{MATT}^+(s_2,s_1') + \mathrm{MATT}^+(s_2',s_1) - \mathrm{MATT}^+(s_2',s_1')$.  That is, the incongruent causal effect of both doing homework and music lessons relative to playing sports can be decomposed into (i) the congruent causal effect of both doing homework and music lessons relative to only doing homework, (ii) the congruent causal effect of doing homework and playing sports relative to only playing sports and (iii) the congruent causal effect of doing homework and playing sports relative to only doing homework; however, the congruent effect (iii) enters the decomposition negatively.
\end{namedexample}

By plugging the second part of Proposition \ref{prop:substitution-effect-decomp} into the decomposition in Theorem \ref{thm:causal-aggregate-marginal-comparison}, it follows that the aggregate marginal effect can be fully expressed as a weighted average of congruent $\textrm{MATT}^+$ parameters.  However, due to the negative signs on some $\textrm{MATT}^+$ parameters in Proposition \ref{prop:substitution-effect-decomp}, it is evident that weights on some $\textrm{MATT}^+$ parameters can be negative.\footnote{To be clear, even if $\textrm{MATT}^+(s_d,s_{d-1})$ shows up negatively in the decomposition from being part of the chain of congruent $\textrm{MATT}^+$'s corresponding to an incongruent $\textrm{MATT}^-$, recall that it also shows up positively in the first term in Theorem \ref{thm:causal-aggregate-marginal-comparison}, and whether or not it ultimately shows up with a positive or negative weight depends on the relative magnitude of the corresponding weights.  Thus, the results in this section do not indicate that negative weights on certain $\textrm{MATT}^+$ parameters \textit{necessarily occur}, but rather that negative weights \textit{can occur.}}  Negative weights on underlying ``building block'' parameters have been emphasized in recent work in econometrics as being indicative of an ``unreasonable'' weighting scheme (e.g., \citet{chaisemartin-dhaultfoeuille-2020,mogstad-torgovitsky-2024}, among others).  For example, given enough heterogeneity in the congruent $\textrm{MATT}^+$'s, negative weights introduce the possibility of sign reversal, where, e.g., the $\textrm{MATT}^+$'s could all be positive but $\E[Y|D=d]-\E[Y|D=d-1]$ could be negative.


\subsection{Non-unique Weights: When Is Aggregation a Problem?} \label{sec:non-unique-weights}

The previous section highlighted that the incongruent causal effect parameters $\textrm{MATT}^-(s_d,s_{d-1})$ are difficult to interpret.  In this section, we return to the other main issue in Theorem \ref{thm:causal-aggregate-marginal-comparison}: that the weights are non-unique.  Let $\mathcal{W}_d$ denote the set of weighting schemes that satisfy the conditions in Theorem \ref{thm:causal-aggregate-marginal-comparison}.  If there exists a weighting scheme in $\mathcal{W}_d$ that puts zero weight on all $\textrm{MATT}^{-}(s_d,s_{d-1})$, then there exists an interpretation of $\E[Y|D=d]-\E[Y|D=d-1]$ that only puts weight on $\textrm{MATT}^{+}(s_d,s_{d-1})$.  This would fully bypass the problems of interpreting $\textrm{MATT}^{-}(s_d,s_{d-1})$ that we discussed above.

Implicit in the discussion above is that there can be multiple weighting schemes that satisfy the conditions in Theorem \ref{thm:causal-aggregate-marginal-comparison}. Thus, we start this section by showing that the weights are indeed non-unique.\footnote{Two exceptions are worth mentioning.  First, the weights are unique when there is a single unique version of the treatment, $K=1$. In this case, for any level of the aggregated treatment, there is only one sub-treatment, so that $w(s_d,s_{d-1})=1$.  Second, the weights are also unique when $d=1$ or $d=\bar{N}$.  Take, for instance, the case where $d=1$, so that we are interested in $\E[Y|D=1] - \E[Y|D=0]$.  Regardless of how many distinct sub-treatments there are, the only sub-treatment vector such that $d=0$ is $s=0_K$.  This implies that $\P(S=0_K|D=0)=1$, and the only weights that satisfy the criteria for the weights in Proposition \ref{prop:marginal-decomposition} are $w(s_1,s_0) = \P(S=s_1|D=1)$, implying that the weights are unique.  An analogous argument holds for $d=\bar{N}$ on the basis that there is only one sub-treatment vector (the one where each sub-treatment is set at its maximum value).} A leading example of weights that are always in $\mathcal{W}_d$ are the product weights $\P(S=s_d|D=d) \times \P(S=s_{d-1}|D=d-1)$.  
This weighting scheme necessarily implies that the aggregate marginal effect includes positive weight on incongruent comparisons. However, they are not the only weights that satisfy the criteria mentioned in the proposition, which we demonstrate by returning to our example.

\begin{namedexample}{\ref*{ex:enrichment} (continued)} \label{example uniform} For $d \in \{1,2\}$, suppose that $\P(s_d|D=d) = 1/3$ for all $s_d \in \mathcal{S}_d$.  Consider the following weights
\begin{align*}
    w_{\hspace{-0.4mm}{A}}(s_d,s_{d-1}) = 
    \begin{cases}  
        \displaystyle \frac{1}{6} & (s_d,s_{d-1}) \in \mathcal{M}^+(d) \\
        0 & (s_d,s_{d-1}) \in \mathcal{M}^-(d)
    \end{cases}
\end{align*}
i.e., $w_{\hspace{-0.4mm}{A}}(s_d,s_{d-1})$ puts $1/6$ weight equally on all six congruent comparisons and $0$ weight on the three incongruent comparisons.  Alternatively, consider the following weights
\begin{align*}
    w_{\hspace{-0.1mm}{B}}(s_d,s_{d-1}) = 
    \begin{cases}  
        0 & (s_d,s_{d-1}) \in \mathcal{M}^+(d) \\
        \displaystyle \frac{1}{3} & (s_d,s_{d-1}) \in \mathcal{M}^-(d)
    \end{cases}
\end{align*}
i.e., these are weights that involve putting $1/3$ weight equally on all three incongruent comparisons and $0$ weight on the six congruent comparisons. Both $w_{\hspace{-0.4mm}{A}}(s_d,s_{d-1})$ and $w_{\hspace{-0.1mm}{B}}(s_d,s_{d-1})$ meet the requirements for the weights that are discussed in Theorem \ref{thm:causal-aggregate-marginal-comparison}.  Besides $w_{\hspace{-0.4mm}{A}}(s_d,s_{d-1})$ and $w_{\hspace{-0.1mm}{B}}(s_d,s_{d-1})$, many other weighting schemes also satisfy the same requirements.
\end{namedexample}

The previous example demonstrates that the weights in Theorem \ref{thm:causal-aggregate-marginal-comparison} are non-unique.  There are different weighting schemes for $\textrm{MATT}$'s that can rationalize the aggregate marginal effect.  Different weighting schemes can lead to very different interpretations of the aggregate marginal effect.  In the example, one valid weighting scheme leads to an interpretation of $\E[Y|D=d]-\E[Y|D=d-1]$ as a weighted average that only includes congruent $\textrm{MATT}^+$'s. This weighting scheme is in line with our ideal scenario above---it provides an interpretation of the aggregate marginal effect $\E[Y|D=d]-\E[Y|D=d-1]$ that is fully congruent.  

In contrast, the following example shows that there exist cases where the aggregate marginal comparison is incompatible with a fully congruent comparison of means of sub-treatments.

\begin{namedexample}{\ref*{ex:enrichment} (continued)}  Suppose that 
\begin{align*}
    &\P\Big((1,0,0)\Big|D=1\Big) = 0.8 
    &&\P\Big((1,1,0)\Big|D=2\Big) = 0.1 \\[5pt]
    &\P\Big((0,1,0)\Big|D=1\Big) = 0.1 
    &&\P\Big((1,0,1)\Big|D=2\Big) = 0.1 \\[5pt]
    &\P\Big((0,0,1)\Big|D=1\Big) = 0.1 
    &&\P\Big((0,1,1)\Big|D=2\Big) = 0.8
\end{align*}
In this case, there do not exist fully congruent weights that can rationalize the $\E[Y|D=2]-\E[Y|D=1]$.  The explanation is that the incongruent sub-treatment vectors $(1,0,0)$ and $(0,1,1)$ occur too commonly for $\E[Y|D=2] - \E[Y|D=1]$ to be rationalized with only fully congruent comparisons across sub-treatment vectors.
\end{namedexample}

In the remainder of this section, we provide four arguments aiming to characterize empirical settings where the aggregate marginal effect, $\E[Y|D=d]-\E[Y|D=d-1]$, would be less likely to put weight on incongruent causal effects. These can be used by practitioners to diagnose, both conceptually and practically, how much of a problem an aggregated treatment may cause in a given application. First, in Section \ref{subsec:number-of-ST-and-incongruent-comparisons}, we show that the number of possible incongruent comparisons grows much more rapidly than the number of possible congruent comparisons as the complexity of the sub-treatments increases.  Second, in Section \ref{subsec:auxiliary-assumptions-to-rule-out-incongruency}, we provide auxiliary assumptions that guarantee a weighting scheme that does not put any weight on incongruent comparisons.  Third, in Section \ref{subsec:settings-with-guaranteed-incongruency}, we discuss the characteristics of applications that \textit{must} put weight on incongruent comparisons. Lastly, in Section \ref{sec: testing SUTVA} we provide a test for whether $D$ is appropriately aggregated---i.e., whether a version of Assumption \ref{ass:sutva2} for the aggregated treatment is valid, in which case there would be no aggregation issues.


\subsubsection{The Link between the Number of Sub-treatments and Incongruent Comparisons} \label{subsec:number-of-ST-and-incongruent-comparisons}

In this section, we establish a link between the complexity of the sub-treatments (i.e., the number of distinct sub-treatments and the number of values that the sub-treatments can take) and the number of incongruent $\textrm{MATT}^-$ parameters that show up in the decomposition in Theorem \ref{thm:causal-aggregate-marginal-comparison} relative to the number of congruent $\textrm{MATT}^+$ parameters.  We show that the number of incongruent terms grows much more rapidly than the number of congruent terms.  The implication for empirical work is that, all else equal, applications with more sub-treatments or complicated sub-treatments are more susceptible to the issues related to incongruent comparisons showing up in the aggregate marginal effects that we discussed above.

Suppose that all sub-treatments are binary (i.e., individuals can either participate or not in any of $K$ binary versions of the treatment). In Proposition \ref{prop:congruent-incongruent-counts-binary-subtreatments} of Appendix \ref{app:proofs}, we establish that $|\mathcal{M}|$, the total number of contrasts for all permutations of sub-treatment vectors at adjacent amounts of aggregated treatment, is equal to $\binom{2K}{K-1}$.\footnote{Formally define the set of all marginal pairs of sub-treatment vectors as $\mathcal{M} := \cup_{d=1}^{\bar{N}} \mathcal{M}(d)$. Likewise, define the marginally congruent set of pairs $\mathcal{M}^+ := \cup_{d=1}^{\bar{N}} \mathcal{M}^+(d)$ and marginally incongruent set of pairs $\mathcal{M}^- := \cup_{d=1}^{\bar{N}} \mathcal{M}^-(d)$.} This number dramatically increases in $K$. For example, if $K=3$, then there are $\binom{6}{2} = 15$ possible contrasts. If $K=4$, there are $\binom{8}{3} = 56$ contrasts and so on. In addition, Proposition \ref{prop:congruent-incongruent-counts-binary-subtreatments} reveals that the amount of congruent contrasts $|\mathcal{M}^+| = K \cdot 2^{K-1}$ and incongruent contrasts $|\mathcal{M}^-| = \binom{2K}{K} - K \cdot 2^{K-1}$, allowing us to formally show that the total number of incongruent pairs of sub-treatment vectors grows much more rapidly with $K$ than the total number of congruent pairs of sub-treatment vectors (see Corollary \ref{cor:asymptotic-congruent-proportion-result} in Appendix \ref{app:proofs}). The relatively rapid growth of the number of incongruent comparisons is illustrated in Figure \ref{fig:proportion-of-incongruent-subtreatments}. 

\begin{figure}[t!] 
\caption{Proportions of Congruent and Incongruent Sub-treatment Vectors}
\vspace{-4mm}
\label{fig:proportion-of-incongruent-subtreatments}
\centering
\begin{subfigure}{.49 \textwidth}
  \centering
  \includegraphics[width=1.\linewidth]{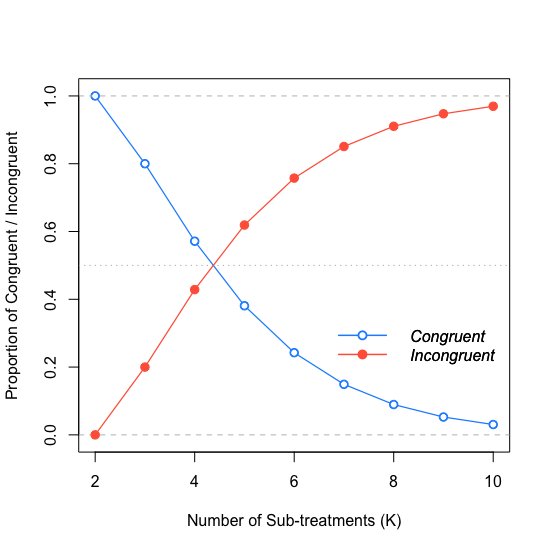}
  \caption{Binary Sub-treatments}
  \label{fig:sub1}
\end{subfigure} \hfill
\begin{subfigure}{.49 \textwidth}
  \centering
  \includegraphics[width=1.\linewidth]{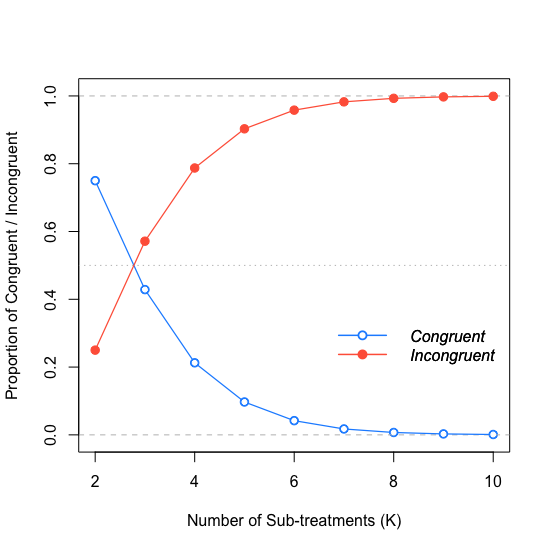}
  \caption{Multivalued Sub-treatments}
  \label{fig:sub2}
\end{subfigure}
\begin{justify}
{\small \textit{Notes:} The figure provides the relative proportion of congruent and incongruent sub-treatment vectors as a function of the number of sub-treatments, $K$.  Panel (a) considers the case where all of the sub-treatments are binary.  Panel (b) considers the case where the sub-treatments are discrete with three possible values $\{0,1,2\}$. }
\end{justify}
\end{figure}

Figure \ref{fig:proportion-of-incongruent-subtreatments} shows the fraction of congruent and incongruent comparisons of sub-treatment vectors for a given number of sub-treatments. Panel (a) presents this in the setting with binary sub-treatments. More than half of the sub-treatment vectors are incongruent for any $K > 4$, and the fraction of incongruent sub-treatment vectors grows rapidly with $K$. Panel (b) considers the case where the sub-treatments can be multivalued---in this panel, all of the sub-treatments can take values among $\{0,1,2\}$. In this case, the number of incongruent contrasts dominates the number of congruent contrasts for any value of $K>2$, and the relative fraction of incongruent contrasts grows even faster with $K$ compared to the case when sub-treatments are binary. See Proposition \ref{prop:congruent-incongruent-counts-trinary-subtreatments} in Appendix \ref{app:proofs} for exact expressions of the number of congruent and incongruent marginal contrasts in this case. Figure \ref{fig:proportion-of-incongruent-subtreatments} and Corollary \ref{cor:asymptotic-congruent-proportion-result} suggest that the relative fraction of incongruent comparisons grows rapidly with the support of sub-treatments.

Having more sub-treatments does not necessarily mean that $\E[Y|D=d]-\E[Y|D=d-1]$ includes comparisons between incongruent sub-treatment vectors. What we have established in this section is that the \textit{scope} for incongruency increases with the number of sub-treatments and with larger support size among the sub-treatments, suggesting that a researcher should pay especially close attention to issues related to incongruency in settings with a large number of sub-treatments, or few sub-treatments possessing multivalued supports. 



\subsubsection{Auxiliary Assumptions that Rule Out Incongruent Comparisons} \label{subsec:auxiliary-assumptions-to-rule-out-incongruency}
This section lists auxiliary assumptions that rule out incongruency affecting the aggregate marginal effect, $\E[Y|D=d]-\E[Y|D=d-1]$.  The first assumption rules out treatment effect heterogeneity across sub-treatments.  The second set of assumptions rules out units sorting across different values of the aggregated treatment and introduces a restriction on latent unit-types.

\subsubsection*{Approach 1: Restrictions On Treatment Effect Heterogeneity}
The first assumption we consider rules out treatment effect heterogeneity with respect to a marginal increase in any of the sub-treatments.
\begin{assumption}[No Heterogeneous Sub-treatment Effects] \label{ass:homogeniety} For all $d \in \D_{>0}$ and any $(s_d,s_{d-1}) \in \mathcal{M}^+(d)$, 
\begin{align*}
    \mathrm{MATT}^+(s_d,s_{d-1}) = \beta_d
\end{align*}
\end{assumption}

Assumption \ref{ass:homogeniety} says that the average marginal causal effect of any sub-treatment is constant across sub-treatments between adjacent levels of aggregated treatment.\footnote{A similar assumption has been referred to in the Epidemiology literature as ``treatment variation irrelevance''. See, for instance, \citet{VanderWeeleHernan2013}.} In many applications, this may be a strong auxiliary assumption.  For example, in our running example, it would say that the causal effect of a one-unit increase in any of the sub-treatments (whether it be homework, music, or sports) is the same for all sub-treatments.  Most likely, this is a strong assumption in this context.

In Proposition \ref{prop:homogeneity-of-effects} in Appendix \ref{app:proofs}, we show that, under Assumptions \ref{ass:sutva2}-\ref{ass:homogeniety},
\begin{align*}
    \E[Y|D=d] - \E[Y|D=d-1] = \beta_d
\end{align*}
In other words, the aggregate marginal effect recovers the average marginal causal effect of the sub-treatments, which is $\beta_d$. The intuition for this result comes from the second part of Proposition \ref{prop:substitution-effect-decomp}: under Assumption \ref{ass:homogeniety}, all the $\textrm{MATT}^+$'s are equal to each other, and Proposition \ref{prop:substitution-effect-decomp} therefore implies that all of the $\textrm{MATT}^-$'s are also equal to $\beta_d$.  Replacing all of the $\textrm{MATT}^+$'s and $\textrm{MATT}^-$'s in Theorem \ref{thm:causal-aggregate-marginal-comparison} then implies the result.  Thus, in some sense, Assumption \ref{ass:homogeniety} does not remove the weights on the $\textrm{MATT}^-(s_d,s_{d-1})$ terms in Theorem \ref{thm:causal-aggregate-marginal-comparison}, but it does make them irrelevant as all of them are equal to $\beta_d$. Equivalently, we can view Assumption \ref{ass:homogeniety} as setting all $\text{SATT}'s$ to zero (which can be seen from Proposition \ref{prop:substitution-effect-decomp}): intuitively, when one sub-treatment is substituted by another, there is no effect on the outcome.


\subsubsection*{Approach 2: Structural Assumptions}
An alternative approach to ruling out incongruity in the decomposition in Theorem \ref{thm:causal-aggregate-marginal-comparison} comes from introducing structural assumptions. 
Let $S_i(d)$ denote the sub-treatment that unit $i$ would experience under aggregated treatment $d$.\footnote{The object $S_i(d)$ should be understood as a latent sub-treatment type across $D$, which describes how the components of treatment are arranged at each total treatment level in observational settings. Recall that aggregated treatment variables are deterministic functions of the sub-treatments, devised \textit{ex post}, and are not causal on the outcome. This indicates that types are inseparable from the aggregation function that underlies them. Hence, the concept of type is an artifact of the aggregation scheme, not itself an intervention with its own causal effect.
} Thus, $S_i(\mathbf{d}) := (S_i(1), S_i(2), \ldots, S_i(\bar{N}))$ defines a unit-level latent aggregated treatment path---the particular sub-treatment vector that a unit would experience for all possible values of the aggregated treatment.  The set of possible values of $S(\mathbf{d})$ is finite, and we can define a notion of a unit's latent type on the basis of $S(\mathbf{d})$. 

\begin{assumption}[No Sorting on $D$] \label{ass:no-sorting-on-D} Latent types are independent of the aggregated treatment; that is,
    $$S(\mathbf{d}) \independent D$$
\end{assumption}

\begin{assumption}[No Incongruent Latent Types] \label{ass:no-incongruity}
    For all $d \in \D_{>0}$, treatment paths are locally congruent; that is,
    \begin{align*}
    \P\big(S(d)=s_d, \; S(d-1)=s_{d-1} \big| D \in \{d, d-1\}\big) &= 0, ~\textrm{if } (s_{d}, s_{d-1}) \in \mathcal{M}^{-}(d) \\
    \P\big(S(d)=s_d, \; S(d-1)=s_{d-1} \big| D \in \{d, d-1\}\big) &\geq 0, ~\textrm{if } (s_{d}, s_{d-1}) \in \mathcal{M}^{+}(d)
    \end{align*}
\end{assumption}
Assumption \ref{ass:no-sorting-on-D} says that latent types are balanced across aggregate amounts of treatment. This ensures that there is no selection at the disaggregate level based on the total amount of treatment.  No sorting holds under random assignment of the sub-treatments.

Assumption \ref{ass:no-incongruity} imposes an explicit restriction on the latent types in the population---that there are no units in a latent type that ``behaves'' incongruently. In our running example, it rules out types of units that would spend one hour of enrichment doing homework, but had they done two hours of enrichment, they would have done music lessons and sports. 

We show in Proposition \ref{prop:no-incongruent-behavior} in Appendix \ref{app:proofs} that these two conditions are sufficient to guarantee that there exists a weighting scheme satisfying the conditions in Theorem \ref{thm:causal-aggregate-marginal-comparison} that puts no weight on any $\textrm{MATT}^-(s_d,s_{d-1})$. 

\begin{namedexample}{\ref*{ex:enrichment} (continued)} \label{ex: types}
It is worth pointing out why Assumption \ref{ass:no-incongruity} alone is not sufficient to guarantee that the aggregate marginal effect can be decomposed entirely in terms of congruent $\textrm{MATT}^+$'s.  Consider an extreme version of our running example, where there are two latent types: type 1 would do homework if they did one hour of enrichment and would do homework and music lessons if they did two hours of enrichment; type 2 would play sports if they did one hour of enrichment and would do music lessons and play sports if they did two hours of enrichment.  Both latent types are congruent.  However, suppose there is sorting so that all type 1 units do one hour of enrichment, while all type 2 units do two hours of enrichment.  In this case, the aggregate marginal effect $\E[Y|D=2]-\E[Y|D=1]$ is \textit{fully incongruent} due to sorting, despite all units themselves belonging to a congruent latent type.
\end{namedexample}

\subsubsection{Settings where Incongruent Comparisons Are Unavoidable} \label{subsec:settings-with-guaranteed-incongruency}
In the previous section, we discussed additional assumptions that side-step incongruent $\textrm{MATT}^-$'s complicating the interpretation of the aggregate marginal effect.  This section pivots to characterizing the features of applications that necessarily include incongruent $\textrm{MATT}^-$'s.  

\subsubsection*{Sub-treatment Decreases in Aggregated Treatment Guarantees Incongruency}
The following result provides a straightforward characteristic of an application that indicates that incongruency is unavoidable in the aggregate marginal effect. 

\begin{proposition} \label{prop:decreasing-means-implies-incongruency}
    Provided there exists some sub-treatment indexed by $k \in \{1, \ldots, K\}$, and some $d \in \D_{>0}$, such that 
    \begin{align*}
        \E[S_k | D=d] \; &< \; \E[S_k | D=d-1]  
    \end{align*}
    then any weights that satisfy the properties in Theorem \ref{thm:causal-aggregate-marginal-comparison} must assign positive weight to at least one incongruent pair of sub-treatments, $(s_d,s_{d-1}) \in \mathcal{M}^{-}(d)$. 
\end{proposition}

Proposition \ref{prop:decreasing-means-implies-incongruency} states that if the conditional means of any sub-treatment declines between values of aggregated treatment $D=d-1$ and $D=d$, then weighting schemes that avoid incongruency are impossible. The condition in the proposition is easy to consider in applications as it concerns the mean of a particular sub-treatment across different values of the aggregated treatment.  In the context of our application, the proposition says that incongruent comparisons cannot be avoided if, for example, the mean number of hours spent on homework was 0.75 among children that did one hour of enrichment while the mean number of hours spent on homework was 0.5 among children that did two hours of enrichment.  This is an intuitive condition for guaranteeing incongruency: if the mean of some sub-treatment decreases in $D$, then there is simply not enough available mass on congruent sub-treatment vectors at the higher value of the aggregated treatment to satisfy the requirements on the weights in Theorem \ref{thm:causal-aggregate-marginal-comparison}.

\subsubsection*{Minimally Incongruent Weights} \label{sec:congruent-weight-search}
The condition in Proposition \ref{prop:decreasing-means-implies-incongruency} is a sufficient, but not necessary, condition for incongruency.  Moreover, if it holds, it implies that incongruency is a problem, but it does not necessarily provide much information about \textit{how much} of a problem it is.  With this in mind, in this section we define \textit{minimally incongruent weights} as a solution to the following linear programming problem:
\begin{align}
    w^\star \in \underset{w}{\textrm{arg\,min}} \sum_{(s_d,s_{d-1)} \in \mathcal{M}^-(d)} w(s_d,s_{d-1}) \label{eqn:maximally-congruent-weights}
\end{align}
subject to
\vspace{-2em}
{ \small 
\begin{align*}
    w(s_d,s_{d-1}) &\geq 0 \text{ for all } (s_d,s_{d-1}) \in \mathcal{M}(d)\\
    \sum_{(s_d,s_{d-1}) \in \mathcal{M}(d)} w(s_d,s_{d-1}) &= 1 \\
    \sum_{s_d} w(s_d,s_{d-1}) &= \P(S=s_{d-1}|D=d-1) \\
    \sum_{s_{d-1}} w(s_d, s_{d-1}) &= \P(S=s_d | D=d)
\end{align*}
}This defines the weights $w^\star$ to be a weighting scheme that minimizes the weight on incongruent comparisons between marginal sub-treatment vectors subject to satisfying the criteria for the weights discussed in Theorem  \ref{thm:causal-aggregate-marginal-comparison}.  There are several additional clarifications worth mentioning.  First, if $w^\star(s_d,s_{d-1}) > 0$ for any $(s_d,s_{d-1}) \in \mathcal{M}^-(d)$, it necessarily implies that the marginal comparison of aggregate means includes incongruent comparisons across sub-treatment vectors.  Second, if $w^\star(s_d,s_{d-1}) = 0$ for all $(s_d,s_{d-1}) \in \mathcal{M}^-(d)$, then the marginal comparison of aggregate means has a representation that only includes congruent comparisons across sub-treatment vectors. However, in general, there can be many weighting schemes that meet these criteria and involve congruent comparisons across sub-treatments; for instance, in the earlier example with uniform probabilities of each sub-treatment vector on page \pageref*{example uniform}, there are many weighting schemes that only involve congruent comparisons across sub-treatment vectors.  

\subsubsection{Testing Whether \texorpdfstring{$D$}{D} is Too Aggregated}\label{sec: testing SUTVA}
Next, we show that one can test whether the version of Assumption \ref{ass:sutva2} \textit{relative to $D$} holds. If it does, then the aggregation issues discussed in this paper are not relevant for the empirical application, and the researcher may use $D$ as their treatment variable without having to use the methods developed in this paper.  The second part of SUTVA is often considered to be untestable (see, for example, the discussion in \citet{Hernan2016}); in this section, we highlight that it is jointly testable with Assumptions \ref{ass:sutva2} and \ref{ass:unconfoundedness} in settings where the researcher observes sub-treatment $S$. We note that an analogous argument to the one below could be used to test Assumption  \ref{ass:sutva2} \textit{per se} (i.e., the version of that assumption relative to $S$) provided the researcher also observes a more disaggregated version of sub-treatment $\tilde{S}$. See Appendix \ref{sec:define-subtreatment} for further details. 

The version of Assumption \ref{ass:sutva2} relative to $D$ holds if, for all $d \in \mathcal{D}$, $Y_i(s_d) = Y_i(s_d')$ for all $i$ and $s_d, s_d' \in \mathcal{S}_d$.  Our test will be based on the comparison of means across different sub-treatment vectors corresponding to the same aggregate value of the treatment. If SUTVA holds for the aggregated treatment $D$, we have that, for any $s_d,s_d' \in \ST_d$,
\begin{align*}
    \E[Y|S=s_d] - \E[Y|S=s_d'] &= \E[Y(s_d)|S=s_d] - \E[Y(s_d')|S=s_d'] \\ 
    &= \underbrace{\E[Y(s_d) - Y(s_d')|S=s_d]}_{\text{SUTVA}} + \underbrace{\E[Y(s_d')|S=s_d] - \E[Y(s_d')|S=s_d']}_{\text{selection bias}},
\end{align*}
where the first equality holds by Assumption \ref{ass:sutva2} relative to $S$, and the second equality holds by adding and subtracting $\E[Y(s_d')|S=s_d]$. The first underlined term in the second line is equal to 0 when there are no hidden versions of the aggregated treatment (i.e., under Assumption \ref{ass:sutva2} relative to $D$), but the second term could still be non-zero---the mean of the potential outcomes of sub-treatment vector $s_d'$ could be different for sub-treatment group $s_d$ relative to sub-treatment group $s_d'$, even if these are not distinct versions of the treatment. For example, ``homework'' and ``music'' could be equivalent versions of the treatment, and yet the latter term could be non-zero if, for some reason related to selection, children who do homework tend to have higher or lower outcomes than children who do music.

However, the underlined selection bias term is equal to zero under Assumption \ref{ass:unconfoundedness}. This implies that the version of Assumption \ref{ass:sutva2} relative to $D$ is testable under the maintained Assumptions \ref{ass:sutva2} and \ref{ass:unconfoundedness}. One can carry out the test proposed here by simple tests for differences in means between all pairs of sub-treatments corresponding to the same aggregate level of the treatment, adjusting for multiple testing error.\footnote{See \citet{HasegawaEtAl2020} for simultaneous inference with two versions of treatment in the binary treatment case, which does not require correcting for multiple testing.}  We also note that related ideas could be used under alternative identification strategies that rely on different assumptions than Assumption \ref{ass:unconfoundedness}.

\subsubsection{Discussion}
This section has aimed to highlight the features of applications where the incongruent comparisons that show up in the decomposition in Theorem \ref{thm:causal-aggregate-marginal-comparison} arise.  First, we showed that the relative number of incongruent $\textrm{MATT}^-$'s grows rapidly in the complexity of the sub-treatments.  Second, we discussed additional assumptions (limitations on treatment effect heterogeneity and restrictions on sorting and latent types) that rule out incongruent $\textrm{MATT}^-$'s in the aggregate marginal effect.  Third, we provided conditions (a sub-treatment that decreases in the aggregated treatment) that guaranteed that incongruent $\textrm{MATT}^-$'s would show up in the aggregate marginal effect. Fourth, we showed how to test whether there should be any aggregation issues by testing the version of Assumption \ref{ass:sutva2} relative to $D$. 

To conclude this section, it is worth emphasizing that these four arguments provide complementary ways for a researcher to informally assess ``how much'' aggregation matters in a particular application.  For example, in an application with a small number of uncomplicated sub-treatments, where the sub-treatments are similar to each other and likely to have close to homogeneous effects, and where the version of Assumption \ref{ass:sutva2} relative to $D$ is not rejected, one should expect the negative implications of working with an aggregated treatment to be small.  In contrast, an application with a large number of more-distinct sub-treatments, heavy sorting across different values of the aggregated treatment, and where the version of Assumption \ref{ass:sutva2} relative to $D$ is rejected, is one in which we should expect major distortions to arise due to the aggregation of the treatment.


\section{Alternative Approaches} \label{sec:DATE}
In the preceding section, we saw that interpreting aggregate marginal effects encountered several complications, arising from two key issues: (i) comparisons across values of the aggregated treatment could mix marginal effects of congruent sub-treatment vectors and marginal effects of incongruent sub-treatment vectors, and (ii) marginal effects of incongruent sub-treatment vectors are difficult to interpret.  We then outlined some ways that a researcher could diagnose (or at least think about) the implications of incongruency with an aggregated treatment in a given application.  In this section, we consider two alternative approaches that can completely side-step the issues related to incongruency that were emphasized above.  First, in Section \ref{subsec:identification-for-nonmarginal-parameters}, we consider alternative, non-marginal causal effect parameters.  Targeting these parameters fully circumvents issues related to incongruency.  They are straightforward to interpret and are estimable when only the aggregated treatment is observed (i.e., they do not require the sub-treatments themselves to be observed).  Because this approach does not require the sub-treatments to be observed, this approach offers a path forward to conduct causal inference even in applications that require a very disaggregated notion of the sub-treatments to satisfy Assumption \ref{ass:sutva2}.  Moreover, the non-marginal comparisons that we consider in that section immediately apply for any aggregation function, not just the sum of the sub-treatments. Changing the target parameter means that it is no longer interpretable as a marginal causal effect parameter, which is a drawback for applications where a researcher strongly prefers this type of parameter.  Second, in Section \ref{subsec:identification-with-observed-subtreatments}, we show how to identify fully congruent marginal causal effect parameters in applications where sub-treatment data is available.  This approach delivers a marginal causal effect parameter, but it requires Assumption \ref{ass:sutva2} to hold for the observed sub-treatments and is more sensitive to the specific aggregation function specified by the researcher.

\subsection{Approach 1: Target Non-marginal Causal Effect Parameters} \label{subsec:identification-for-nonmarginal-parameters}
Marginal effects have a strong claim on being the most natural target parameters in the setting that we are considering, where the aggregated treatment can take multiple values, and reflect the most common ways that empirical work interprets results in these settings.  However, the previous section documented several challenges with interpreting aggregate marginal effects in the presence of sub-treatments.  Instead of considering marginal changes in the aggregated treatment, in this section, we focus on interpreting $\E[Y|D=d] - \E[Y|D=0]$ for any $d \in \D_{>0}$, which is the difference between the means of outcomes for the group that experiences aggregated treatment $d$ relative to the untreated group. We show that this non-marginal, aggregate comparison does not include incongruent comparisons across sub-treatments.  We further show that, under Assumption \ref{ass:unconfoundedness}, this comparison has a causal interpretation as the average of the causal effects of each sub-treatment $s_d \in \mathcal{S}_d$ relative to being untreated. Importantly, sub-treatments in this case do not need to be observed by the researcher. We refer to this as a \textit{baseline-to-$d$} comparison in the text below. 

In terms of causal effect parameters, the main building block parameter in this section is the   \textit{average treatment effect on the treated} (ATT)
\begin{align*}
    \textrm{ATT}(s_d) := \E[Y(s_d) - Y(0) | S=s_d] 
\end{align*}
which is defined with respect to a given sub-treatment $s_d$ and where $Y(0)$ is shorthand notation for being untreated (i.e., where all sub-treatments are equal to zero).  $\textrm{ATT}(s_d)$ is the average effect of experiencing sub-treatment vector $s_d$ relative to being untreated among sub-treatment group $s_d$.  In the spirit of using the aggregated treatment to summarize the causal effects of the sub-treatments, our main target parameter in this section is 
\begin{align*}
    \textrm{AATT}(d) := \E\big[\textrm{ATT}(S)\big|D=d\big]
\end{align*}
which is the aggregate average treatment effect on the treated across sub-treatments corresponding to the aggregated treatment being equal to $d$.  From the law of iterated expectations, it follows that 
\begin{align} \label{eqn:aatt-weights}
    \textrm{AATT}(d) = \sum_{s_d \in \mathcal{S}_d} \P(S=s_d|D=d) \cdot \textrm{ATT}(s_d)
\end{align}
i.e., that $\textrm{AATT}(d)$ is a weighted average of the underlying $\textrm{ATT}$'s of specific sub-treatments, with weights given by the relative frequency of that sub-treatment among all sub-treatments that aggregate to $d$.  

In some applications, it is also useful to scale $\textrm{ATT}(s_d)$, or $\textrm{AATT}(d)$, by the amount of the aggregated treatment, i.e., to consider the parameters
\begin{align*}
    \frac{\mathrm{ATT}(s_d)}{d}~~~~\text{or}~~~~\frac{\mathrm{AATT}(d)}{d}
\end{align*}
which can be interpreted as average treatment effects per unit of the (sub)-treatment.  We refer to these as scaled $\textrm{ATT}$'s and scaled $\textrm{AATT}$'s, respectively.


\subsubsection*{Identification} \label{sec:datt-identification}
Next, we provide identification results for the average treatment effect parameters discussed above. 

\begin{theorem} \label{thm:att-identification} Under Assumptions \ref{ass:sutva2} and \ref{ass:unconfoundedness}, for $d \in \mathcal{D}_{>0}$ and $s_d \in \mathcal{S}_d$,
\begin{align*}
    \mathrm{ATT}(s_d) = \E[Y|S=s_d] - \E[Y|S=0_K] \quad \text{and} \quad \mathrm{AATT}(d) = \E[Y|D=d] - \E[Y|D=0]
\end{align*}
where $0_K$ denotes the zero vector of length $K$. $\mathrm{ATT}(s_d)$ is identified if the sub-treatments are observed. $\mathrm{AATT}(d)$ is identified whether or not the sub-treatments are observed.
\end{theorem}

Theorem \ref{thm:att-identification} shows that $\textrm{AATT}(d)$ is identified under Assumptions \ref{ass:sutva2} and \ref{ass:unconfoundedness}, even if the researcher only observes the aggregated treatment (and not the sub-treatments).  It is interesting to compare this result with the one in Theorem \ref{thm:causal-aggregate-marginal-comparison} above concerning the comparison of means of outcomes for marginal increases in the aggregated treatment (i.e., $\E[Y|D=d]-\E[Y|D=d-1]$).  A major issue for the marginal comparison emphasized in Section \ref{sec:matt} was the non-uniqueness of the weights and the possibility of incongruency.  Neither of those issues apply for the baseline-to-$d$ comparisons, $\E[Y|D=d]-\E[Y|D=0]$, considered here.  The ``weights'' on the underlying sub-treatment-specific $\textrm{ATT}(s_d)$ parameters are given in Equation \eqref{eqn:aatt-weights}. 
These are unique, positive for all relevant sub-treatments, and intuitive---they correspond to the relative frequency of each relevant sub-treatment. Mechanically, the same sort of double-sum arguments can be used here as in the previous case, but, by construction, $\mathcal{S}_0$ only has one element, which results in the implicit weighting scheme being unique in this case.  The benefit is that, unlike for the marginal case, $\E[Y|D=d]-\E[Y|D=0]$ is straightforward to interpret, and all of the issues related to incongruency emphasized above can be avoided.\footnote{Both Theorem \ref{thm:causal-aggregate-marginal-comparison} and Theorem \ref{thm:att-identification} invoked Assumption \ref{ass:unconfoundedness}.  In both cases, this assumption is stronger than necessary, though the minimal assumptions to provide a causal interpretation in each result are non-nested.  Causal interpretations of marginal effects of sub-treatments can hold under a local version of no selection, while causal interpretations of $\textrm{AATT}(d)$ can be rationalized under a version of the no-selection assumption that involves untreated potential outcomes only.  This could be a meaningful difference in some applications (though it is not relevant for any of our discussions about aggregation specifically).  We discuss these differences in more detail in Appendix \ref{sec:minimal-assumptions}.}


\subsubsection*{Interpreting Regressions with Scaled \texorpdfstring{Baseline-to-$d$}{Baseline-to-d} Building Blocks} \label{sec:datt-regression}

Next, we return to interpreting the coefficient on the aggregated treatment variable in the regression from Equation \eqref{eqn:reg}, but we relate it to the scaled baseline-to-$d$ building blocks: $(\E[Y|D=d]-\E[Y|D=0])/d$.  In Proposition \ref{prop:regression-with-baseline-to-d-building-blocks} in the Supplementary Appendix (\cite{CCCDSupp2025}), we show that
\begin{align*}
    \alpha_1 &= \sum_{d=1}^{\Bar{N}} \tilde{\omega}^{reg}(d) \cdot \frac{ \E[Y|D=d] - \E[Y|D=0]}{d}
\end{align*}
where \vspace{-1em}
\begin{align*}
    \tilde{\omega}^{reg}(d) = \frac{d \cdot (d-\E[D])}{\Var(D)} \cdot \P(D=d)
\end{align*}
and satisfies the following properties: (i) $\displaystyle \sum_{d=1}^{\bar{N}} \tilde{\omega}^{reg}(d) = 1$ 
and (ii)  $\tilde{\omega}^{reg}(d) \lessgtr 0$ for $d \lessgtr \E[D]$.\footnote{The proof uses the same mechanics as Theorem S3 in the Supplementary Appendix of \citet{chaisemartin-dhaultfoeuille-2020}, though the context of that result (interpreting two-way fixed effects regressions) is very different from ours.}

The expression for $\alpha_1$ above can be combined with the result in Theorem \ref{thm:att-identification} to say that the regression coefficient on the aggregated treatment can be interpreted as a weighted average of scaled $\textrm{AATT}(d)$ parameters.  Relative to the regression weights discussed above for the marginal case, the regression weights with scaled baseline-to-$d$ primitives have worse properties. The weights are negative for values of the aggregated treatment below $\E[D]$, implying that $\alpha_1$ is not weakly causal when scaled $\textrm{AATT}(d)$'s are the underlying building block parameters.  In addition, the weights are systematically increasing in magnitude in their distance from $\E[D]$, meaning that effects for sub-treatment groups with more extreme values of the aggregated treatment $D$ ``count more'' than effects for other sub-treatment groups. 

The discussion above highlights a certain tension with interpreting $\alpha_1$ in terms of baseline-to-$d$ causal effects.  While the building blocks $(\E[Y|D=d] - \E[Y|D=0]) / d$ are more interpretable than the marginal effects discussed previously, the weights inherited from the regression have poor properties.  Thus, in settings where a researcher would like to report a single, scalar summary of the causal effects of the sub-treatments, a natural alternative to reporting $\alpha_1$ is to report either
\begin{align*}
    \E\big[\textrm{AATT}(D)\big|D>0\big]~~~~\text{or}~~~~\E\left[ \frac{\textrm{AATT}(D)}{D} \middle| D > 0\right]
\end{align*}
which are all identified under the same conditions that would give $\alpha_1$ a causal interpretation, but do not suffer from the poor weighting scheme stemming from the regression.

\begin{remark} [Linearity and Homogeneity]
    Consider the case where, for all $d \in \D_{>0}$ and $s_d \in \mathcal{S}_d$, $\mathrm{ATT}(s_d) = \theta \times d$, so that the average treatment effects of all sub-treatments are (i) constant across sub-treatments corresponding to the same aggregate amount of the treatment and (ii) linear in $d$. (i) and (ii) restrict treatment effect heterogeneity across sub-treatments and impose a linearity condition.  In this case, $\E\left[ \frac{\mathrm{AATT}(D)}{D} \middle| D > 0\right] = \theta$, and, in addition, it also follows that $\alpha_1 = \theta$.  In other words, in this case, the regression would deliver the unique scaled treatment effect parameter.  In practice, both (i) and (ii) are likely to be strong auxiliary assumptions for most applications.  This suggests it is a better strategy to directly target parameters such as $\E\left[ \frac{\mathrm{AATT}(D)}{D} \middle| D > 0\right]$ rather than hoping that the regression will deliver them. 
\end{remark}


\subsection{Approach 2: Target Marginal Effect Parameters with Sub-treatment Data} \label{subsec:identification-with-observed-subtreatments}

In this section, we target summary marginal causal effect parameters, exploiting that the sub-treatments are observed, which only applies for some applications.  To start with, recall that, under Assumptions \ref{ass:sutva2} and \ref{ass:unconfoundedness},
\begin{align*}
    \mathrm{MATT}^+(s_d,s_{d-1}) = \E[Y|S=s_d] - \E[Y|S=s_{d-1}],\quad (s_d,s_{d-1}) \in \mathcal{M}^+(d).
\end{align*}
Therefore, when the sub-treatments are observed, $\textrm{MATT}^+(s_d,s_{d-1})$ is identified (see Proposition \ref{prop:comparison-of-subtreatments} in Appendix \ref{app:proofs}).  Given that this marginal effect is defined at the sub-treatment level, none of the issues related to incongruency that we emphasized above in the context of aggregation apply.  In practice, a researcher could estimate and report $\textrm{MATT}^+(s_d,s_{d-1})$ for any (or all) combinations of congruent sub-treatment vectors.  Leaving the discussion here, however, would not fully address some relevant empirical challenges---presumably, in the majority of applications where the sub-treatments are observed, the entire reason to introduce an aggregated treatment variable is that there tend to be few observations that experienced each specific combination of sub-treatments.  This implies that the sort of non-parametric analysis mentioned above would suffer from a form of curse of dimensionality, resulting in each $\textrm{MATT}^+(s_d,s_{d-1})$ being estimated imprecisely and in poor performance of inference procedures for the $\textrm{MATT}^+$'s.

In contrast, however, even when the number of observations per combination of sub-treatments is small, one may still be able to estimate averages of $\textrm{MATT}^+$'s well.  A natural option is 
\begin{align} \label{eqn:amatt-plus}
    \widetilde{\mathrm{AMATT}^{+}}(d) &:=  \sum_{(s_d,s_{d-1}) \in \mathcal{M}^+(d)} \tilde{w}^+(s_d,s_{d-1}) \cdot \mathrm{MATT}^+(s_d, s_{d-1})
\end{align}
where, for $(s_d,s_{d-1}) \in \mathcal{M}^+(d)$, 
\begin{align} \label{eqn:matt-congruent-weight}
    \tilde{w}^+(s_d,s_{d-1}) &:=  \frac{\P\big(S(d)=s_d,S(d-1)=s_{d-1} \big| D \in \{d,d-1\} \big)}{\displaystyle \sum_{(s'_d,s'_{d-1}) \in \mathcal{M}^+(d)} \P\big(S(d)=s'_d,S(d-1)=s'_{d-1} \big| D \in \{d,d-1\} \big)}.
\end{align}
$\widetilde{\textrm{AMATT}^+}(d)$ is a special case of $\textrm{AMATT}^{+}_{w^+}(d)$ in Equation \eqref{eqn:amatt-tilde} above, as it is a specific weighted average of congruent $\textrm{MATT}^+$'s.  The weights come from the joint distribution of latent sub-treatment types local to the aggregated treatment either being $d$ or $d-1$. These weights give $\widetilde{\textrm{AMATT}^+}(d)$ a clear interpretation as a marginal ``on-the-treated'' type of parameter as it depends on the distribution of the sub-treatments that are (or would be) experienced.  
The term in the denominator of the expression for $\tilde{w}^+(s_d,s_{d-1})$ can be thought of as normalizing the weights on congruent sub-treatments so that they sum to one.\footnote{In settings where units would make congruent sub-treatment choices at different amounts of the aggregated treatment (i.e., if Assumption \ref{ass:no-incongruity} holds), then the expression in the denominator is equal to one, and the weights do not need to be normalized.  Alternatively, one can view the normalization as arising due to dropping incongruent $\textrm{MATT}$'s.} Recovering $\widetilde{\textrm{AMATT}^+}(d)$, however, introduces additional challenges relative to identifying $\textrm{MATT}$'s:  even if the sub-treatments are observed, the weights depend on the joint distribution of latent sub-treatment types for aggregated treatment $D=d$ or $D=d-1$ and, therefore, require additional assumptions to identify. 
We discuss these issues in more detail in Appendix \ref{app:id-matt-plus}; however, in order to avoid introducing additional assumptions, we instead focus on identifying a version of  $\textrm{AMATT}^{+}_{w^+}(d)$ with researcher-chosen weights, sacrificing some interpretability but increasing tractability.  A leading option for researcher-chosen weights is to use the normalized product weights, i.e., 
\begin{align*}
\nu^+(s_d, s_{d-1}) \coloneqq 
\frac{\P(S = s_d \mid D = d) \cdot \P(S = s_{d-1} \mid D = d - 1)}
     {\displaystyle\sum_{(s_d', s_{d-1}') \in \mathcal{M}^+(d)} \P(S = s_d' \mid D = d) \cdot \P(S = s_{d-1}' \mid D = d - 1)},
\end{align*}
for $(s_d,s_{d-1}) \in \mathcal{M}^+(d)$, and then to consider
\begin{align}
   \mathrm{AMATT}^+(d) := \sum_{(s_d,s_{d-1}) \; \in \; \mathcal{M}^+(d)} \nu^+(s_d, s_{d-1}) \cdot \mathrm{MATT}^+(s_d,s_{d-1}). \label{eqn:congruent-product-weight}
\end{align}
Using this weighting scheme results in putting more weight on common sub-treatments. An immediate implication of Assumptions \ref{ass:sutva2} and \ref{ass:unconfoundedness} and observing the sub-treatments (see Proposition \ref{prop:comparison-of-subtreatments}) is that $\textrm{AMATT}^{+}_{w^+}(d)$ in Equation \eqref{eqn:amatt-tilde} is identified; this also implies that $\textrm{AMATT}^+(d)$ is identified under the same conditions.  

In settings where a researcher would like a scalar summary of the marginal causal effects of the treatment, a natural target parameter is
\begin{align*}
    \textrm{AMATT}^+ &:= \E\big[ \textrm{AMATT}^+(D) \big| D>0 \big]
\end{align*}
which is an average marginal (weakly) causal effect parameter that comes from averaging $\textrm{AMATT}^+(d)$ over the distribution of the aggregated treatment variable $D$.  It is identified under Assumptions \ref{ass:sutva2} and \ref{ass:unconfoundedness} when the sub-treatments are observed. 


\section{Empirical Application} \label{sec:EmpiricalApplication}
In this section, guided by our results above, we illustrate the aggregation issues and the methods proposed above with data from \citet{ CCaetanoEtAl2024}, which studies the effects of enrichment activities on noncognitive skills among children in the U.S. Like that paper, we use data from the Childhood Development Supplement (CDS) provided by the Panel Study of Income Dynamics (PSID) that contains time-use diaries and measures of cognitive and noncognitive skills. 
For clarity, we make some simplifications. 
Our estimates come from simple comparisons of means; we ignore control variables and fixed effects in our analysis, and we abstract from the bunching identification strategy that is often used in this literature.  Instead of trying to make a causal empirical claim, in this section, we only aim to highlight the issues that stem from aggregating sub-treatment variables; and, as discussed above, these issues continue to apply whether or not Assumption \ref{ass:unconfoundedness} holds.

The aggregated treatment variable in our application is \textit{total hours of enrichment activity}, which is an aggregation of four sub-treatments: Lessons, Structured Sports, Volunteering, and Before \& After School Programs.  Each sub-treatment represents the average number of hours per week that the child spent on that specific activity, and the sub-treatment vector represents the bundle of sub-treatments each child was exposed to. In our data, we observe each child's participation in each sub-treatment.  We round the sub-treatments to their nearest half-hour. We then sum across all four sub-treatments for each individual so that the aggregate variable $D=(S_1+S_2+S_3+S_4)$ is the total amount of enrichment activities.  The outcome of interest is noncognitive skill, a normalized index of socio-emotional and behavioral ratings with mean zero and standard deviation of one. Larger values indicate better noncognitive scores. We follow \citet{CCaetanoEtAl2024} in constructing this variable; see that paper for more details. See Supplementary Appendix \ref{subsubsec:ApplicationDataDescription} for summary statistics and more details on how we constructed the data used in our application (\cite{CCCDSupp2025}).  In the main text, to simplify the discussion, we focus on a subsample of low socio-economic status children in 2019.\footnote{We opted for illustrating the analysis in this smaller subsample because it has sufficiently few different sub-treatment values, allowing us to report Table \ref{tab:algorithm_results} in the paper.}  In Supplementary Appendix \ref{supapp:full-sample-analysis} (\cite{CCCDSupp2025}), we provide the results for the full sample used in \citet{CCaetanoEtAl2024}.

Observing the sub-treatments is important for our analysis.  Below, we often temporarily ignore that we observe the sub-treatments and act as if we only had access to the aggregated treatment.  Then, exploiting the fact that we actually do observe the sub-treatments, we are able to diagnose how much aggregation itself affects the results.  It is also worth mentioning that, although the decompositions that we discussed in previous sections were written in terms of population quantities, the same arguments can be applied to their sample analogues.  We also report standard errors below, but we mainly emphasize the point estimates---because the sample, sub-treatments, and outcomes are the same across estimators, differences in results reflect real differences in the estimands rather than noise. As such, statistical significance is not the primary lens through which to assess these differences; instead, differences in the point estimates themselves reflect how much varying the estimands matters in practice.

\subsection{Sub-treatment Diagnostics} \label{subsec:application-diagnostics}

We begin by examining evidence of incongruency across different values of the aggregated treatment. Specifically, we investigate how the composition of sub-treatments varies with the level of total enrichment, recalling from Proposition \ref{prop:decreasing-means-implies-incongruency} that a sub-treatment whose mean declines in the aggregated treatment implies the presence of incongruency.

\begin{figure}[t]
\caption{Average Amount of Sub-treatments across Each Level of Aggregated Treatment}
\vspace{-8mm}
\label{fig:AvgSTxD_plot}
\begin{center}
\includegraphics[scale=0.49]{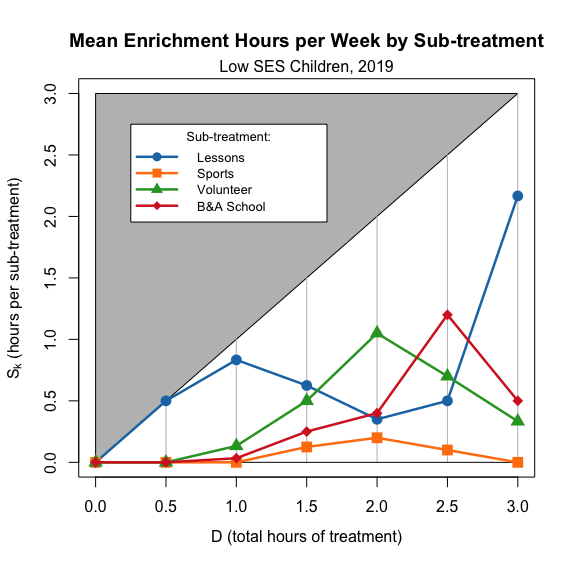}
\end{center}
\begin{justify}
{\small \textit{Notes:} The figure displays the average amount of each sub-treatment as a function of the total amount of treatment, $D$.}
\end{justify}
\end{figure}

Figure \ref{fig:AvgSTxD_plot} displays the mean of each sub-treatment across every level of the aggregated treatment $D$. The plot shows the mean number of hours for each type of enrichment activity (each sub-treatment) as the total hours of enrichment activity increases. The 45-degree line represents, for each value $D$ in the horizontal axis, the vertical sum of hours across all sub-treatments, which naturally equals the total number of hours spent on enrichment, $D$. For example, at $D=0.5$, the average amount of hours spent on lessons is 0.5, which is 100\% of the total enrichment for that value of $D$. At $D=1$, we observe that the average number of hours for lessons increases relative to $D=0.5$, but this sub-treatment is no longer the only sub-treatment that is experienced at $D=1$. 

More interestingly, the mean for the lessons sub-treatment declines from $D=1$ to $D=1.5$. From Proposition \ref{prop:decreasing-means-implies-incongruency}, this is evidence of incongruency, implying that the aggregate marginal effect at $D=1.5$ cannot avoid putting weight on incongruent marginal sub-treatment effects. Notice that these violations continue for marginal increases from $D=1.5$ to $D=2.0$ for lessons; from $D=2.0$ to $D=2.5$ for both sports and volunteering sub-treatments; and from $D=2.5$ to $D=3.0$ for all sub-treatments except lessons. This suggests that aggregate marginal effects (including regressions interpreted as marginal effects) are hard to interpret---as discussed above, the incongruent $\textrm{MATT}^-$ terms that will show up here include hard-to-interpret substitution effects or, equivalently, congruent $\textrm{MATT}^+$'s with negative weights.

Although the sub-treatment plot can tell us where weight on incongruent comparisons is inescapable, it does not explain \textit{how much} incongruity there is.  To answer this question, we next find the minimally incongruent weights by solving the linear program in Section \ref{sec:congruent-weight-search}. 
These weights provide an interpretation of the aggregate marginal effect with minimal incongruity.  The results from solving the problem are displayed in Table \ref{tab:algorithm_results}, which lists all marginal pairs of sub-treatment vectors at each $D=d$ from the data.  In line with the results from Figure \ref{fig:AvgSTxD_plot}, the minimally incongruent weights put weight on incongruent marginal sub-treatment effects for the aggregate marginal effects at $D=1.5, 2.0, 2.5,$ and $3.0$. Moreover, for $D=1.5,2.0,$ and $2.5$, the weight on incongruent marginal sub-treatment effects is substantial, ranging from 30-40\% of the total weight. Even more strikingly, all of the weight falls on incongruent comparisons between $D=2.5$ and $D=3.0$ because there are no available congruent sub-treatment vectors. 

\begin{table}[ht]
    \centering
    \caption{Minimally Incongruent Weights}
    \label{tab:algorithm_results}
    \scalebox{0.84}{%
    \begin{threeparttable}
        \begin{tabular}{@{}l l c c c c r@{}}
            \toprule
            & \multicolumn{1}{c}{\large $s_d$} & 
              \multicolumn{1}{c}{\large $s_{d-1}$} & 
              \multicolumn{1}{c}{\large $\in\mathcal M^-(d)$} & 
              \multicolumn{1}{c}{\large $\delta(s_d,s_{d-1})$} & 
              \multicolumn{1}{c}{\large $w^\star(s_d,s_{d-1})$} & 
              \multicolumn{1}{r}{\large \textit{Incongruent Wt.\,(\%)} } \\
            \bottomrule
            \midrule

            \textit{D=0.5} & & & & & & \\ 
                  & \mapD{1},\;\mapD{0},\;\mapD{0},\;\mapD{0} 
                  & \mapD{0},\;\mapD{0},\;\mapD{0},\;\mapD{0} 
                  &  & $-0.219$ & $1.000$ & \cellcolor{white}\textit{0.00\%} \\ 
            \midrule

            \textit{D=1.0} & & & & & & \\ 
                  & \mapD{1},\;\mapD{0},\;\mapD{0},\;\mapD{1} 
                  & \mapD{1},\;\mapD{0},\;\mapD{0},\;\mapD{0} 
                  &  & $0.388$  & $0.067$ & \cellcolor{white}  \\ 
                  & \mapD{1},\;\mapD{0},\;\mapD{1},\;\mapD{0} 
                  & \mapD{1},\;\mapD{0},\;\mapD{0},\;\mapD{0} 
                  &  & $-0.227$ & $0.267$ & \cellcolor{white}  \\ 
                  & \mapD{2},\;\mapD{0},\;\mapD{0},\;\mapD{0} 
                  & \mapD{1},\;\mapD{0},\;\mapD{0},\;\mapD{0} 
                  &  & $-0.018$ & $0.667$ & \cellcolor{white}\textit{0.00\%} \\
            \midrule

            \textit{D=1.5} & & & & & & \\ 
            \rowcolor{gray!10}
                  & \mapD{0},\;\mapD{0},\;\mapD{0},\;\mapD{3} 
                  & \mapD{1},\;\mapD{0},\;\mapD{0},\;\mapD{1} 
                  & $\times$ & $0.085$  & $0.004$ & \cellcolor{white}  \\ 
                  & \mapD{1},\;\mapD{0},\;\mapD{0},\;\mapD{2} 
                  & \mapD{1},\;\mapD{0},\;\mapD{0},\;\mapD{1} 
                  &  & $0.340$  & $0.063$ & \cellcolor{white}  \\ 
            \rowcolor{gray!10}
                  & \mapD{0},\;\mapD{0},\;\mapD{0},\;\mapD{3} 
                  & \mapD{1},\;\mapD{0},\;\mapD{1},\;\mapD{0} 
                  & $\times$ & $0.701$  & $0.017$ & \cellcolor{white}  \\ 
                  & \mapD{1},\;\mapD{0},\;\mapD{2},\;\mapD{0} 
                  & \mapD{1},\;\mapD{0},\;\mapD{1},\;\mapD{0} 
                  &  & $-0.000$ & $0.188$ & \cellcolor{white}  \\ 
                  & \mapD{1},\;\mapD{1},\;\mapD{1},\;\mapD{0} 
                  & \mapD{1},\;\mapD{0},\;\mapD{1},\;\mapD{0} 
                  &  & $0.752$  & $0.063$ & \cellcolor{white}  \\ 
            \rowcolor{gray!10}
                  & \mapD{0},\;\mapD{0},\;\mapD{0},\;\mapD{3} 
                  & \mapD{2},\;\mapD{0},\;\mapD{0},\;\mapD{0} 
                  & $\times$ & $0.491$  & $0.104$ & \cellcolor{white}  \\ 
            \rowcolor{gray!10}
                  & \mapD{0},\;\mapD{0},\;\mapD{3},\;\mapD{0} 
                  & \mapD{2},\;\mapD{0},\;\mapD{0},\;\mapD{0} 
                  & $\times$ & $0.236$  & $0.188$ & \cellcolor{white}  \\ 
            \rowcolor{gray!10}
                  & \mapD{0},\;\mapD{3},\;\mapD{0},\;\mapD{0} 
                  & \mapD{2},\;\mapD{0},\;\mapD{0},\;\mapD{0} 
                  & $\times$ & $0.030$  & $0.063$ & \cellcolor{white}  \\ 
                  & \mapD{3},\;\mapD{0},\;\mapD{0},\;\mapD{0} 
                  & \mapD{2},\;\mapD{0},\;\mapD{0},\;\mapD{0} 
                  &  & $0.061$  & $0.313$ & \cellcolor{white}\textit{37.51\%} \\
            \midrule

            \textit{D=2.0} & & & & & & \\ 
                  & \mapD{0},\;\mapD{0},\;\mapD{0},\;\mapD{4} 
                  & \mapD{0},\;\mapD{0},\;\mapD{0},\;\mapD{3} 
                  &  & $-0.170$ & $0.125$ & \cellcolor{white}  \\ 
                  & \mapD{0},\;\mapD{0},\;\mapD{4},\;\mapD{0} 
                  & \mapD{0},\;\mapD{0},\;\mapD{3},\;\mapD{0} 
                  &  & $-0.022$ & $0.075$ & \cellcolor{white}  \\ 
                  & \mapD{1},\;\mapD{0},\;\mapD{3},\;\mapD{0} 
                  & \mapD{0},\;\mapD{0},\;\mapD{3},\;\mapD{0} 
                  &  & $-0.141$ & $0.113$ & \cellcolor{white}  \\ 
                  & \mapD{0},\;\mapD{4},\;\mapD{0},\;\mapD{0} 
                  & \mapD{0},\;\mapD{3},\;\mapD{0},\;\mapD{0} 
                  &  & $-0.104$ & $0.063$ & \cellcolor{white}  \\ 
            \rowcolor{gray!10}
                  & \mapD{0},\;\mapD{0},\;\mapD{4},\;\mapD{0} 
                  & \mapD{1},\;\mapD{0},\;\mapD{0},\;\mapD{2} 
                  & $\times$ & $-0.532$ & $0.025$ & \cellcolor{white}  \\ 
            \rowcolor{gray!10}
                  & \mapD{0},\;\mapD{4},\;\mapD{0},\;\mapD{0} 
                  & \mapD{1},\;\mapD{0},\;\mapD{0},\;\mapD{2} 
                  & $\times$ & $-0.820$ & $0.038$ & \cellcolor{white}  \\ 
                  & \mapD{1},\;\mapD{0},\;\mapD{3},\;\mapD{0} 
                  & \mapD{1},\;\mapD{0},\;\mapD{2},\;\mapD{0} 
                  &  & $0.304$  & $0.188$ & \cellcolor{white}  \\ 
            \rowcolor{gray!10}
                  & \mapD{0},\;\mapD{0},\;\mapD{4},\;\mapD{0} 
                  & \mapD{1},\;\mapD{1},\;\mapD{1},\;\mapD{0} 
                  & $\times$ & $-0.329$ & $0.063$ & \cellcolor{white}  \\ 
            \rowcolor{gray!10}
                  & \mapD{0},\;\mapD{0},\;\mapD{0},\;\mapD{4} 
                  & \mapD{3},\;\mapD{0},\;\mapD{0},\;\mapD{0} 
                  & $\times$ & $0.260$  & $0.075$ & \cellcolor{white}  \\ 
            \rowcolor{gray!10}
                  & \mapD{0},\;\mapD{0},\;\mapD{4},\;\mapD{0} 
                  & \mapD{3},\;\mapD{0},\;\mapD{0},\;\mapD{0} 
                  & $\times$ & $0.154$  & $0.138$ & \cellcolor{white}  \\ 
                  & \mapD{4},\;\mapD{0},\;\mapD{0},\;\mapD{0} 
                  & \mapD{3},\;\mapD{0},\;\mapD{0},\;\mapD{0} 
                  &  & $-0.712$ & $0.100$ & \cellcolor{white}\textit{33.75\%} \\
            \midrule

            \textit{D=2.5} & & & & & & \\ 
                  & \mapD{1},\;\mapD{0},\;\mapD{0},\;\mapD{4} 
                  & \mapD{0},\;\mapD{0},\;\mapD{0},\;\mapD{4} 
                  &  & $-0.417$ & $0.200$ & \cellcolor{white}  \\ 
            \rowcolor{gray!10}
                  & \mapD{0},\;\mapD{0},\;\mapD{0},\;\mapD{5} 
                  & \mapD{0},\;\mapD{0},\;\mapD{4},\;\mapD{0} 
                  & $\times$ & $0.058$  & $0.200$ & \cellcolor{white}  \\ 
                  & \mapD{1},\;\mapD{0},\;\mapD{4},\;\mapD{0} 
                  & \mapD{0},\;\mapD{0},\;\mapD{4},\;\mapD{0} 
                  &  & $-0.762$ & $0.100$ & \cellcolor{white}  \\ 
            \rowcolor{gray!10}
                  & \mapD{1},\;\mapD{1},\;\mapD{0},\;\mapD{3} 
                  & \mapD{0},\;\mapD{4},\;\mapD{0},\;\mapD{0} 
                  & $\times$ & $0.601$  & $0.100$ & \cellcolor{white}  \\ 
                  & \mapD{1},\;\mapD{0},\;\mapD{4},\;\mapD{0} 
                  & \mapD{1},\;\mapD{0},\;\mapD{3},\;\mapD{0} 
                  &  & $-0.642$ & $0.100$ &  \\ 
                  & \mapD{2},\;\mapD{0},\;\mapD{3},\;\mapD{0} 
                  & \mapD{1},\;\mapD{0},\;\mapD{3},\;\mapD{0} 
                  &  & $0.687$  & $0.200$ & \cellcolor{white}  \\ 
            \rowcolor{gray!10}
                  & \mapD{1},\;\mapD{1},\;\mapD{0},\;\mapD{3} 
                  & \mapD{4},\;\mapD{0},\;\mapD{0},\;\mapD{0} 
                  & $\times$ & $0.601$  & $0.100$ & \cellcolor{white}\textit{40.00\%} \\ 
            \midrule

            \textit{D=3.0} & & & & & & \\ 
            \rowcolor{gray!10}
                  & \mapD{5},\;\mapD{0},\;\mapD{1},\;\mapD{0} 
                  & \mapD{0},\;\mapD{0},\;\mapD{0},\;\mapD{5} 
                  & $\times$ & $0.138$  & $0.200$ & \cellcolor{white} \\ 
            \rowcolor{gray!10}
                  & \mapD{5},\;\mapD{0},\;\mapD{1},\;\mapD{0} 
                  & \mapD{1},\;\mapD{0},\;\mapD{0},\;\mapD{4} 
                  & $\times$ & $0.507$  & $0.200$ & \cellcolor{white}  \\ 
            \rowcolor{gray!10}
                  & \mapD{5},\;\mapD{0},\;\mapD{1},\;\mapD{0} 
                  & \mapD{1},\;\mapD{0},\;\mapD{4},\;\mapD{0} 
                  & $\times$ & $0.959$  & $0.200$ & \cellcolor{white}  \\ 
            \rowcolor{gray!10}
                  & \mapD{3},\;\mapD{0},\;\mapD{0},\;\mapD{3} 
                  & \mapD{1},\;\mapD{1},\;\mapD{0},\;\mapD{3} 
                  & $\times$ & $-0.747$ & $0.133$ & \cellcolor{white}  \\ 
            \rowcolor{gray!10}
                  & \mapD{5},\;\mapD{0},\;\mapD{1},\;\mapD{0} 
                  & \mapD{1},\;\mapD{1},\;\mapD{0},\;\mapD{3} 
                  & $\times$ & $-0.117$ & $0.067$ & \cellcolor{white}  \\ 
            \rowcolor{gray!10}
                  & \mapD{3},\;\mapD{0},\;\mapD{0},\;\mapD{3} 
                  & \mapD{2},\;\mapD{0},\;\mapD{3},\;\mapD{0} 
                  & $\times$ & $-1.001$ & $0.200$ & \cellcolor{white}\textit{100.00\%} \\
            \midrule
            \bottomrule
        \end{tabular} 
    \begin{tablenotes}
        \footnotesize
        \item \textit{Notes:} The table presents the set of \textit{minimally incongruent weights} derived from the optimization problem presented in Section \ref{sec:congruent-weight-search}. The components of the sub-treatment vectors are in the following order: (Lessons, Sports, Volunteering, B\&A School). The $\in \mathcal{M}^{-}(d)$ column indicates if the pairs of sub-treatments are incongruent. The $\delta(s_d,s_{d-1})$ column reports the difference in mean outcomes at $s_d$ and $s_{d-1}$. The $w^{\star}(s_d,s_{d-1})$ column reports the minimally incongruent weight. The \textit{``Incongruent Wt.''} column reports the minimal percentage of possible weight on incongruent comparisons at that $D=d$. Any positive percentage indicates where incongruent comparisons are unavoidable. There are 95 total unique comparisons of sub-treatment vectors that can be made with this sample (19 congruent and 76 incongruent). The set of minimally incongruent weights has 38 of these comparisons (17 congruent and 21 incongruent).  
    \end{tablenotes}
    \end{threeparttable}%
    }%
\end{table}

\subsection{Regression and Target Parameter Estimation}
Next, we move to estimating the effects of enrichment activities, comparing the different approaches that we have discussed throughout the paper. Table \ref{tab:parameter_estimates} reports estimates of several scalar parameters summarizing the effects of enrichment activities on a child's noncognitive skills.  All estimates in Table \ref{tab:parameter_estimates} come from plug-in estimators of the target parameters considered in the paper. Each estimate is in terms of the number of standard deviations away from the mean noncognitive skill score (zero).

\FloatBarrier

Panel I contains an estimate of the coefficient on total enrichment activities from the regression in Equation \eqref{eqn:reg}, which represents the best linear approximation between noncognitive skill and total enrichment activity.  This regression represents the most common way to summarize the relationship between enrichment activities and noncognitive skills. The estimated coefficient is $\hat{\alpha}_1=-0.061$ and is not statistically different from zero at standard levels of significance.\footnote{This estimate and others reported in this section are qualitatively similar to the ones in \citet{CCaetanoEtAl2024}, which finds negative effects of enrichment on noncognitive skills. However, we note that their bunching/selection-on-unobservables identification strategy is very different from the approach in this section.} Given our discussion on the presence of incongruency in our application above, it follows that, if one aims to interpret this parameter in marginal terms, then it consists of incongruent comparisons as described in Section \ref{sec:identification-with-unobserved-subtreatments}.  The regression coefficient also inherits the implicit regression weights discussed in Remark \ref{rem:regression-marginal} above.

\begin{table}[t!]
    \centering
    \begin{threeparttable}
    \caption{Overall Aggregate Parameter Estimates}
    \label{tab:parameter_estimates}
    \begin{tabular}{@{}l l c c c@{}}
        \toprule
         & Parameter & Estimate & $S$ Data & Incongruity \\
        \bottomrule
        \midrule

        \textit{I. Regression} & & & & \\ 
              & $\alpha_1$                         
              & -0.061 (0.042) 
              &       
              & $\times$ \\ 
        \midrule

        \textit{II. Marginal} & & & & \\ 
              & $\E[ \; \delta (D) \;|D>0]$    
              & -0.040 (0.035) 
              &       
              & $\times$ \\ 
              & $\E[ \textrm{AMATT}^{+}(D) |D>0]$
              & -0.087 (0.063)
              & $\times$
              &          \\ 
        \midrule

        \textit{III. Non-marginal} & & & & \\ 
              & $\E[ \textrm{AATT}(D) |D>0]$                   
              & -0.167 (0.068)
              &          
              &          \\ 
              & $\E[ \textrm{AATT}(D)/D |D>0]$
              & -0.218 (0.095) 
              &          
              &          \\ 
        \midrule
        \toprule
        Observations 
            & \multicolumn{1}{l}{214} 
            &  
            &    
            &  \\ 
        \bottomrule
    \end{tabular}

    \begin{tablenotes}
        \footnotesize
        \item \textit{Notes:} Parameter estimates of different target parameters on children’s noncognitive skills. Standard errors in parentheses obtained by bootstrap (1000 iterations). The ``$S$ Data'' column indicates if sub-treatment data are required for estimation. The ``Incongruity'' column indicates when incongruent comparisons are present in the parameter. The estimates of $\E[\textrm{AMATT}^{+}(D) |D>0]$ use the scaled product weights from Equation \eqref{eqn:congruent-product-weight}.
    \end{tablenotes}
    \end{threeparttable}
\end{table}

Next, Panel II  reports estimates of two alternative overall marginal effects parameters.  The first is $\E[\delta(D)|D>0]$, where $\delta(d) :=\E[Y|D=d] - \E[Y|D=d-1]$ is the average aggregate marginal effect.  Like the regression coefficient from Panel I, this parameter includes incongruent marginal sub-treatment effects, but it does not inherit the implicit regression weights (i.e., it is a non-parametric summary of the aggregate marginal effects).  Our estimate of $\E[\delta(D)|D>0]$ is -0.040.  The estimated value of this parameter is closer to zero than the regression coefficient, which indicates that the regression weights put relatively more weight on values of the aggregated treatment with larger marginal effects in magnitude, in comparison to weighting by the distribution of $D$.
The second estimate in Panel II is for $\E[\textrm{AMATT}^{+}(D)|D>0]$; our estimate of this parameter is -0.087.   As discussed above, this parameter is a weighted average of all congruent marginal sub-treatment effects. Among the parameters reported in this section, this parameter has the strongest claim to being the best way to summarize the effects of marginal increases in the sub-treatments, though estimating it does require observing a sub-treatment vector satisfying Assumption \ref{ass:sutva2}. 

\begin{figure}[t!] 
\caption{$d$-Specific Aggregate Parameter Estimates}
\vspace{-4mm}
\label{fig:MATT_and_AATT_estimates}
\centering
\begin{subfigure}{.49 \textwidth}
  \centering
  \includegraphics[width=1.05\linewidth]{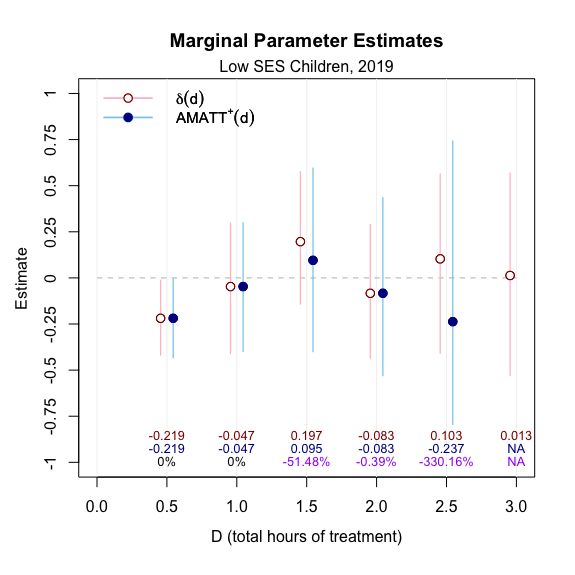}
  \caption{Aggregate Marginal ATT.}
  \label{fig:sub2_estimates}
\end{subfigure}
\begin{subfigure}{.49 \textwidth}
  \centering
  \includegraphics[width=1.05\linewidth]{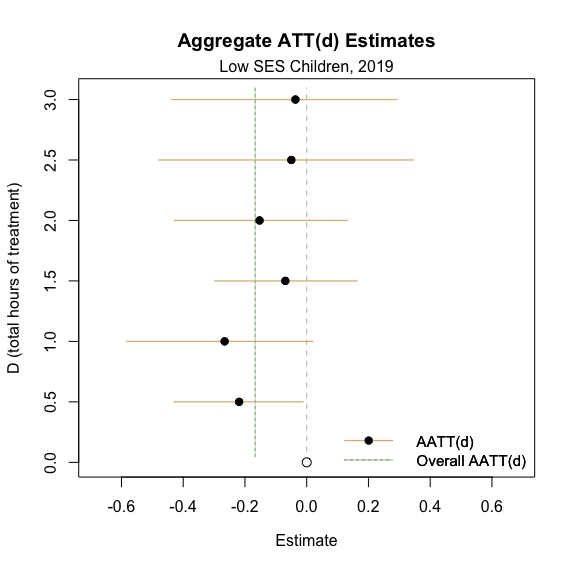}
  \caption{Aggregate ATT.}
  \label{fig:sub1_estimates}
\end{subfigure} \hfill
\begin{justify}
{\small \textit{Notes:} The figure provides estimates of the target parameters discussed in this paper for enrichment activity amounts on noncognitive skills in children of low socioeconomic status in 2019.  Panel (a) displays estimates of $\delta(d)$ and $\textrm{AMATT}^{+}(d)$ with 95\% confidence intervals across all $D=d$. The value of each estimate, and the percent change from $\delta(d)$ to $\textrm{AMATT}^{+}(d)$, are listed at the bottom of Panel (a): $\delta(d)$---top; $\textrm{AMATT}^{+}(d)$---middle; and, percent change---bottom.
At $D=3.0$ hours of total enrichment, no congruent comparisons were possible between $D=2.5$ and $D=3.0$ in the data; hence, the value of \texttt{NA} is reported for $\textrm{AMATT}^{+}(d=3)$. The estimates of $\textrm{AMATT}^{+}(d)$ use the scaled product weights from Equation \eqref{eqn:congruent-product-weight}.  Panel (b) displays estimates of $\textrm{AATT}(d)$ and 95\% confidence intervals across all $D=d$, and an overall $\textrm{AATT}$ which weights the $\textrm{AATT}(d)$'s by the distribution of the aggregated treatment $D$. }
\end{justify}
\end{figure}

Comparing the previous estimates to our estimate of $\E[\textrm{AMATT}^{+}(D)|D>0]$ provides a good way to assess how much incongruency matters in our application.  Relative to our estimate of $\E[\delta(D)|D>0]$, it is more than twice as large in magnitude.  The difference in these estimates is fully explained by whether or not they include incongruent comparisons (since they both weight across the aggregated treatment using the observed distribution of $D$), indicating that incongruency has a large impact on the aggregated estimates. To further decompose this difference, Panel (a) of Figure \ref{fig:MATT_and_AATT_estimates} reports estimates of $\delta(d)$ and $\textrm{AMATT}^{+}(d)$ across different values of $d$. Recall that, according to Figure \ref{fig:AvgSTxD_plot}, $\delta(d)$ must include those incongruent comparisons at the disaggregated level for values of $d \in \{ 1.5, 2.0, 2.5, 3.0 \}$, which could potentially lead to notable differences in the estimates of these parameters at those values of the aggregated treatment. This difference is most striking at $D=2.5$, where our estimate of $\delta(d=2.5)$ is 0.103, while our estimate of $\textrm{AMATT}^{+}(d=2.5)$ is -0.237.  In addition, because there are no available congruent marginal sub-treatment effects at $D=3$, $\textrm{AMATT}^{+}(d=3)$ cannot be estimated and does not contribute to the estimate of $\E[\textrm{AMATT}^{+}(D)|D>0]$. There are other smaller differences at other values of $d$.  Together, these differences explain the disparity between the two overall estimates in the middle panel of Table \ref{tab:parameter_estimates}. Interestingly, relative to $\E[\delta(D)|D>0]$, the estimate of $\alpha_1$ is closer to our estimate of $\E[\textrm{AMATT}^{+}(D)|D>0]$.  There are two sources of differences between these two estimates: $\alpha_1$ includes incongruent marginal sub-treatment effects and inherits a weighting scheme from the regression.  Here, these two issues work in opposite directions---incongruency pushes the estimates towards zero, while the regression weights push it farther from zero---resulting in the estimate of $\E[\textrm{AMATT}^{+}(D)\mid D > 0]$ being closer to that of $\alpha_1$ than to that of $\E[\delta(D)\mid D > 0]$.

To summarize, there are two main explanations for the variation in the estimated values of the various marginal effect parameters considered here.  First, for parameters that are averages of the aggregate marginal effects, it is not possible to avoid incongruent comparisons.  This is a byproduct of large changes in the composition of sub-treatments across different values of the aggregated treatment.  Second, there appears to be notable heterogeneity in sub-treatment-specific marginal effects.  Estimates of congruent sub-treatment marginal range from  -0.762 to 0.752 (these are reported in Table \ref{tab:algorithm_results}), implying that variations in weights across different summary marginal effect parameters are important.\footnote{To be clear, our estimates of sub-treatment-specific marginal effects are likely to be quite imprecise.  Nevertheless, mechanically, these estimates contribute non-trivially to the estimates of the summary marginal effect parameters that we have emphasized here.}

Finally, in Panel III of Table \ref{tab:parameter_estimates}, we report estimates of the non-marginal parameters that we discussed in Section \ref{sec:DATE}.  Unlike $\E[\textrm{AMATT}^{+}(D)|D>0]$, neither of these parameters requires the sub-treatments to be observed; moreover, they do not suffer from non-unique weights, nor do they include incongruent comparisons.  The first of these, $\E[\textrm{AATT}(D)|D>0]$, is the overall aggregate average treatment effect on the treated for those who participated in the treatment. Our estimate is -0.167.  This is notably larger in magnitude than any of the preceding estimates, though it is hard to directly compare them, as the underlying building blocks are different.  Here, the underlying building blocks are $\textrm{AATT}(d)$---among sub-treatments that aggregate to $d$, the average of the effect of each sub-treatment relative to being untreated.  Estimates of $\textrm{AATT}(d)$ are displayed in Panel (b) of Figure \ref{fig:MATT_and_AATT_estimates} for different values of $d$. Averaged across sub-treatments, we find our estimates of the effect of participating in different types of enrichment activities are largest in magnitude among those who participate in few enrichment activities (i.e., for $D=0.5$ and $1.0$). 
The last parameter considered in the table is $\E[\frac{\textrm{AATT}(D)}{D} \big| D>0]$; our estimate of this parameter is -0.218.  This parameter summarizes the average treatment effect on the treated per unit of treatment received. Recall that, as in Section \ref{sec:datt-regression}, $\alpha_1$ can be interpreted as a weighted average of the same scaled baseline-to-$d$ building blocks.  But our estimate of $\alpha_1$ is very different from our estimate of $\E[\frac{\textrm{AATT}(D)}{D} \big| D>0]$---our estimate of $\E[\frac{\textrm{AATT}(D)}{D} \big| D>0]$ is over three times as large in magnitude.  This difference is fully explained by differences in the weighting schemes for each case.  In our view, if a researcher did not insist on reporting a marginal effect parameter, the non-marginal effect parameters considered in Panel III would provide a good way to summarize the effects of enrichment activities on noncognitive skills. This would be especially true for an application where the types of enrichment activities that the child participated in were not observed.  However, using a regression to summarize these types of parameters seems to perform poorly, at least relative to calculating each $\textrm{AATT}(d)$ individually and then manually averaging them. 


\FloatBarrier

\section{Conclusion} \label{subsec:Conclusion}

Aggregated treatment variables are widely used in empirical work focused on causal inference. They often streamline the narrative of the paper, simplify empirical strategies, accommodate data limitations, and improve precision in estimation. But as we have shown in this paper, these conveniences can come at a cost: the marginal effects of an aggregated treatment can be difficult to interpret. The root of the problem lies in the fact that comparisons across different values of an aggregated treatment can include incongruent sub-treatment comparisons. We have shown that these comparisons lead to negative weight issues, so that even if all marginal sub-treatment effects were positive, one could still find a negative aggregate marginal effect. 

Our paper also provides a set of solutions to this underappreciated problem in the causal inference literature.  First, we proposed non-marginal estimands that avoid incongruent comparisons altogether and which can be implemented even when sub-treatment data are unavailable.  This provides a general path forward even when SUTVA fails for the treatment variable the researcher wishes to use.  Second, when sub-treatment data are available, researchers can also continue focusing on marginal estimands by applying desirable weighting schemes that restrict attention solely to congruent comparisons. With this in mind, we note that reporting non-marginal estimands has an underappreciated advantage: it offers greater robustness to violations of SUTVA, which may lead cautious researchers to consider non-marginal effects as a conservative complement to marginal estimands.

In order to emphasize issues related to aggregation, many of our results were in the context of no selection (Assumption \ref{ass:unconfoundedness}), an ideal setting for causal inference. 
The insights developed in our paper, however,  extend beyond this setting. They are also relevant for other identification strategies---including selection on observables, instrumental variables, regression discontinuity designs, difference-in-differences, and bunching. Relative to our setting---where all units are exchangeable and thus comparable (in expectation)---these strategies impose restrictions on the set of permissible comparisons used to estimate causal effects. As a result, they may either attenuate or exacerbate the extent to which incongruent sub-treatment comparisons enter the estimand, relative to the baseline case studied in this paper. A fuller understanding of how aggregation interacts with these identification strategies remains an important direction for future work.



\printbibliography



\appendix


\section{Defining Sub-treatments} \label{sec:define-subtreatment}

In the main text, we invoked Assumption \ref{ass:sutva2}, which rules out hidden versions of the sub-treatments. We maintained this assumption for the sub-treatment vector $S$ specified by the researcher.
Still, in any given application, it may not be obvious how to define this sub-treatment vector, and this section provides some clarifying discussion on this point.  

To start with, we introduce some extra notation.  Given a researcher-specified definition of sub-treatments, let $\mathcal{V}$ denote the set of versions of the treatment.  In our running example from the main text,
\begin{align*}
    \mathcal{V} = \{\text{homework}, \text{music}, \text{sports} \}.
\end{align*}
Note that, in the main text, we mainly worked in terms of sub-treatment vectors that were elements of $\mathcal{S}$.  Relative to elements of $\mathcal{V}$, an element of $\mathcal{S}$ also contains information about the \textit{amount} of each sub-treatment.  

Now, consider a more disaggregated notion of the versions of the treatment, $\tilde{\mathcal{V}}$.  What we mean by ``more disaggregated'' is that there exists a \textit{function} $\tilde{\mathcal{V}} \mapsto \mathcal{V}$ such that the elements of $\tilde{\mathcal{V}}$ are all either the same as some element in $\mathcal{V}$ or represent versions of an element of $\mathcal{V}$.  For example,
\begin{align*}
    \tilde{\mathcal{V}} = \{ \text{homework with mom}, \text{homework with dad}, \text{music}, \text{sports} \},
\end{align*}
where the elements in $\tilde{\mathcal{V}}$ \textit{homework with mom} and \textit{homework with dad} both map to \textit{homework} in $\mathcal{V}$.  
Define $\tilde{\mathcal{S}}$ to be the set of sub-treatment vectors corresponding to the versions of treatment in $\tilde{\mathcal{V}}$.  It immediately follows that, if Assumption \ref{ass:sutva2} holds with respect to the sub-treatment vectors in $\mathcal{S}$, then a version of Assumption \ref{ass:sutva2} holds with respect to the sub-treatment vectors in $\tilde{\mathcal{S}}$; i.e., if no hidden versions of the sub-treatments (the second part of SUTVA) holds with respect to one notion of the sub-treatments, then it necessarily holds with respect to a more disaggregated notion of the sub-treatments.  This is an argument for favoring the more disaggregated notion of the sub-treatments, as it invokes a weaker version of Assumption \ref{ass:sutva2}.

Continuing with the same example, notice that it is possible to create even further sub-versions of the treatment.  For example, a child doing homework with their mother can be disaggregated into doing homework with their mother in the morning, doing homework with their mother at night, etc. The discussion above implies that a version of Assumption \ref{ass:sutva2} would hold with respect to this definition of sub-treatments as well.  This line of thinking can go on and on, becoming more specific and granular at each step, up to the point of all treatments being fully distinct from each other; e.g., a child doing homework with their mother, Jane Doe, at 12:39pm on June 20, 2025.\footnote{Our discussion echoes the Epidemiology literature on this point. For instance, \citet{Hernan2016} states that 
``The process of precisely specifying (versions of treatment) never ends. Version \scalebox{0.87}{$\mathrm{\#}$}1,000,000 would be very long but still imprecise. It is impossible to provide an absolutely precise definition of a version of treatment.''} 

The discussion above highlights how strong Assumption \ref{ass:sutva2} with respect to low-dimensional sub-treatments may be in many applications---let alone the version of this assumption for the aggregated treatment, which motivated the analysis of this paper in the first place. It does not seem farfetched that, in our running example, Assumption \ref{ass:sutva2} is only plausible for an extremely disaggregated notion of the sub-treatments.  For example, the potential outcome of a child doing an hour of homework with their mother could also depend on the time of day or the day of the week (or other things).  We conjecture that similar arguments would apply for many treatments in economics.

Observing this degree of granularity is very rare, and, given the implausibility of this assumption for most versions of treatment, it is possible that the researcher observes treatment at some level of detail, yet still judges those versions to be too coarse for Assumption \ref{ass:sutva2} to plausibly hold. The researcher in this case may want to pivot to report non-marginal estimands, which have robustness properties with respect to Assumption \ref{ass:sutva2}---it requires this assumption to be valid only for a \textit{conceptualizable} finite vector of sub-treatments that does not need to be observed.

However, what happens if the researcher wants to identify a marginal estimand?  Given the discussion above, is there any reason that this researcher should not use the most disaggregated notion of the sub-treatment vector available in the data? While it is true that working with very disaggregated notions of the sub-treatment vector does make the corresponding version of Assumption \ref{ass:sutva2} more plausible, there are important tradeoffs: when Assumption \ref{ass:sutva2} holds with respect to a more aggregated notion of the sub-treatments, using that more aggregated notion of the sub-treatments suffers less from issues related to incongruity that we emphasized in the main text.  This suggests that, all else equal, it is desirable to use the most aggregated notion of the sub-treatment vector that satisfies Assumption \ref{ass:sutva2}. In that sense, the test discussed in Section \ref{sec: testing SUTVA} can be implemented to detect whether there is enough evidence against Assumption \ref{ass:sutva2} for a given level of definition of $S$, provided the researcher also observes a more disaggregated notion of sub-treatment vector $\tilde{S}$. A failure to reject this test may justify using $S$ as the notion of sub-treatment for the analysis. In the extreme case, if the researcher fails to reject Assumption \ref{ass:sutva2} with respect to the aggregated treatment $D$, this may justify not conducting the analyses suggested in this paper. In that case, none of the issues related to incongruent comparisons that we emphasized in the main text would apply, and the researcher can work directly with the aggregated treatment rather than the sub-treatments (even if they are observed).
\section{Additional Results} \label{app:additional-results}

This section contains additional details for several results that were briefly discussed in the main text.

\subsection{Minimal Assumptions for Causal Interpretation} \label{sec:minimal-assumptions}
For some settings, weaker assumptions than Assumption \ref{ass:unconfoundedness} can be applied to achieve identification. Here, we introduce weaker assumptions that we may use to give $\E[Y|D=d]-\E[Y|D=d-1]$ and $\E[Y|D=d]-\E[Y|D=0]$ a causal interpretation. 
\begin{assumption}[No Local Selection] \label{ass:local-unconfoundedness} For any $d \in \D_{>0}$, and for any $s_d \in \ST_d$ and $s_{d-1} \in \ST_{d-1}$, $\Big(Y(s_d), Y(s_{d-1})\Big) \independent S|D \in \{d-1,d\}$.
\end{assumption}

Assumption \ref{ass:local-unconfoundedness} is the local version of Assumption \ref{ass:unconfoundedness}, where the distribution of potential outcomes is the same across sub-treatment groups among those sub-treatment groups that have aggregated treatment equal to $d$ or $d-1$ only.  Because this assumption is local, it is weaker than Assumption \ref{ass:unconfoundedness}, yet it still implies that $\E[Y(s)|S=s'] = \E[Y|S=s]$, but only among sub-treatments $s, s' \in \mathcal{M}(d)$, rather than any sub-treatments.

Next, we consider an assumption that is weaker than Assumption \ref{ass:unconfoundedness} that may still be used to identify the non-marginal causal effect parameters $\textrm{ATT}(s_d)$. 
\begin{assumption}[No Selection on Untreated Potential Outcomes] 
    \label{ass:untreated-potential-outcomes} $Y(0) \independent S$.
\end{assumption}
Assumption \ref{ass:untreated-potential-outcomes} says that untreated potential outcomes are independent of sub-treatments.  It is implied by Assumption \ref{ass:unconfoundedness} but not by Assumption \ref{ass:local-unconfoundedness}.  It is helpful to compare Assumption \ref{ass:untreated-potential-outcomes} to Assumption \ref{ass:local-unconfoundedness}.  One dimension in which Assumption \ref{ass:untreated-potential-outcomes} is weaker is that it only involves untreated potential outcomes rather than potential outcomes of different sub-treatments.  This allows for certain forms of selection into a particular sub-treatment that are ruled out by Assumption \ref{ass:local-unconfoundedness}.  On the other hand, Assumption \ref{ass:untreated-potential-outcomes} requires some degree of similarity (i.e., in terms of untreated potential outcomes) across units that experience very different sub-treatments, which is not required for the local version of unconfoundedness in Assumption \ref{ass:local-unconfoundedness}; see \citet{Lewis1973} for an argument that the nearest, or most similar, world/sub-treatment to the one observed should be primarily considered as the counterfactual.  



\subsection{Interpreting Regressions with Marginal Building Blocks} \label{sec:regressions-with-marginal-building-blocks}
Since regressions are frequently used in estimation in applications with aggregated treatments, this section provides results on interpreting regressions that include the aggregated treatment as a regressor, formalizing some of the claims made in the main text.\footnote{As discussed in the main text, all of our identification results above can hold conditional on observed covariates, which is likely to be an important extension for many applications.  There would be considerable differences for interpreting regressions that also include covariates as they would introduce extra issues even if there were no sub-treatments (see, for example, \citet{angrist-1998}, \citet{aronow-samii-2016}, \citet{sloczynski-2022}, and \citet{hahn-2023}).  We leave this extension for future work.} We consider interpreting $\alpha_1$ from the regression of the outcome on the aggregated treatment variable in Equation \eqref{eqn:reg}. The following result provides a decomposition of $\alpha_1$ in terms of aggregate marginal effects.

\begin{proposition} \label{prop:regression-with-marginal-building-block}
    Whether sub-treatments are observed or unobserved, $\alpha_1$ from the regression in Equation \eqref{eqn:reg} can be decomposed as
    \begin{align*}
        \alpha_1 &= \sum_{d=1}^{\Bar{N}} \omega^{reg}(d) \cdot \Big( \E[Y|D=d] - \E[Y|D=d-1] \Big) 
    \end{align*}
    where the regression weights are
    \begin{align*}
        \omega^{reg}(d) := \frac{\big(\E[D|D \geq d] - \E[D]\big) \cdot \P(D \geq d)}{\Var(D)}
    \end{align*}
    and satisfy the properties: (i) $\omega^{reg}(d) \geq 0$ for all values of $d \in \D_{>0}$, (ii) $\displaystyle \sum_{d=1}^{\bar{N}} \omega^{reg}(d) = 1$, and (iii) $\omega^{reg}(d)$ is decreasing in distance from $\E[D]$. 
\end{proposition}

The proof is provided in the Supplementary Appendix (\cite{CCCDSupp2025}).  Proposition \ref{prop:regression-with-marginal-building-block} amounts to the discrete analog of a well-known result in \citet{yitzhaki-1996} on interpreting regressions like the one in Equation \eqref{eqn:reg} with a continuous regressor (see also \citet{callaway-goodman-santanna-2025} for some related results on interpreting regressions with an ordered, discrete regressor in the context of difference-in-differences). Notice that all of the terms in the expression for $\alpha_1$ involve the aggregated treatment variable, not the sub-treatments.  In particular, $\alpha_1$ is equal to a weighted average of $\E[Y|D=d]-\E[Y|D=d-1]$.  The weights are all positive, which is a good property and one that does not always hold for interpreting regressions in various other contexts (e.g., regressions that include covariates (\citet{hahn-2023}), regressions with multiple treatments (\citet{goldsmith-hull-kolesar-2024}), and two-way fixed effects regressions (\citet{chaisemartin-dhaultfoeuille-2020})).  The other notable property of the regression weights is that they are largest for values of the aggregated treatment that are close to the mean of $D$.  This indicates that a byproduct of summarizing the effects of the treatment using a regression is that certain effects of marginal increases in the sub-treatments are systematically weighed more heavily than others, depending on their corresponding amount of the aggregated treatment.  This is likely to be an undesirable property of the regression weights in most applications.

The last thing to mention in this section is that all of our previous results on interpreting $\E[Y|D=d] - \E[Y|D=d-1]$ in the presence of sub-treatments continue to apply to all of the terms that show up in the expression for $\alpha_1$.  For example, combining the expression for $\alpha_1$ with the result in Theorem \ref{thm:causal-aggregate-marginal-comparison} implies that comparisons across incongruent sub-treatment vectors can contribute to $\alpha_1$.

Taken together, the discussion here implies that running a regression using the aggregated treatment as a regressor (i) does not alleviate issues related to incongruent sub-treatment vectors, and (ii) also introduces a somewhat strange implicit weighting scheme to combine information across different values of the aggregated treatment.  Issue (i) corresponds exactly to the issues that we discussed in Section \ref{sec:matt}.  Issue (ii) can be directly addressed, simply by averaging $\E[Y|D=d] - \E[Y|D=d-1]$ over the distribution of $D$, rather than inheriting the weighting scheme from the regression.\footnote{Recall from Section \ref{subsec:auxiliary-assumptions-to-rule-out-incongruency} that $\E[Y|D=d]-\E[Y|D=d-1]$ recovered the average marginal causal effect of the sub-treatments under Assumption \ref{ass:homogeniety}, which said that the marginal effect of all sub-treatments at a given level of the treatment was constant, $\beta_d$.  In order for the regression in Equation \eqref{eqn:reg} to recover the average marginal causal effect of the sub-treatments requires strengthening this assumption so that $\beta_d = \beta$, which is constant across $d$.  This argument holds by combining the arguments in Section \ref{subsec:auxiliary-assumptions-to-rule-out-incongruency} with the result in Proposition \ref{prop:regression-with-marginal-building-block}.}


\subsection{Identification of \texorpdfstring{$\mathbf{\widetilde{\textrm{AMATT}^{\boldsymbol{+}}}(d)}$}{AMATTilde+(d)}} \label{app:id-matt-plus}  
Finally, in this section, we discuss identification of $\widetilde{\textrm{AMATT}^+}(d)$, which was defined in Equation \eqref{eqn:amatt-plus}.  Relative to previous identification results, in order to recover $\widetilde{\textrm{AMATT}^+}(d)$, we need to additionally identify $\P\big(S(d)=s_d, S(d-1)=s_{d-1} | D\in \{d,d-1\}\big)$ for $(s_d,s_{d-1}) \in \mathcal{M}^+(d)$.

\begin{assumption}[Latent Sub-treatment Independence] \label{ass:latent-subtreatment-independence}
For all $d \in \D_{>0}$,
\begin{align*}
    S(d) \independent S(d-1) \big| D \in \{d-1,d\}.
\end{align*}
\end{assumption}

Assumption \ref{ass:latent-subtreatment-independence} says that the sub-treatment actually experienced at, say, the lower level of the aggregated treatment is not informative about which sub-treatment would have been experienced at the higher level of the aggregated treatment. It is difficult to think of examples where this assumption is plausible in the types of applications that we have emphasized in the paper, particularly where the aggregated treatment is a summary of the sub-treatments rather than itself causal---this is a main reason why we have relegated this assumption (and identifying $\widetilde{\textrm{AMATT}^+}(d)$) to the appendix.  Perhaps the leading case where this assumption would be plausible is one in which the sub-treatments are assigned in a sequential procedure where the aggregated treatment is randomly assigned in the first step and then the sub-treatments are randomly assigned conditional on the assignment of the aggregated treatment; though, at least to some extent, this assignment mechanism goes against the notion of aggregated treatment that we have considered, as it is more aligned to the mediation analysis case. Also, notice that a by-product of Assumption \ref{ass:latent-subtreatment-independence}, is that, by construction, there would be incongruent latent types that occur with positive probability, implying that this assumption is incompatible with Assumption \ref{ass:no-incongruity} from the main text.

\begin{proposition} \label{prop:matt-identification} Under Assumptions \ref{ass:sutva2}, \ref{ass:unconfoundedness}, \ref{ass:no-sorting-on-D} and \ref{ass:latent-subtreatment-independence}, if sub-treatments are observed, $\widetilde{\mathrm{AMATT}^+}(d)$ is identified and given by
{ \small 
\begin{align*}
    \widetilde{\mathrm{AMATT}^+}(d) = \sum_{(s_d,s_{d\tightminus 1}) \in \mathcal{M}^+(d)} \left(  \frac{\P\big(S\tightequals s_d\big|D\tightequals d\big) \times \P(S\tightequals s_{d-1} \big| D\tightequals d\tightminus 1\big)}{\displaystyle \sum_{(s_d',s_{d\tightminus 1}') \in \mathcal{M}^+(d)}  \P\big(S\tightequals s_d' \big|D\tightequals d\big) \times \P(S\tightequals s_{d\tightminus 1}' \big| D\tightequals d\tightminus 1\big) } \right) \cdot \Big( \E[Y|S\tightequals s_d] - \E[Y|S\tightequals s_{d\tightminus 1}] \Big).
\end{align*}
}
\end{proposition}

The proof is provided in \ref{app:proofs-marginal-regression} of the Supplementary Appendix (\cite{CCCDSupp2025}).  Proposition \ref{prop:matt-identification} shows that $\widetilde{\textrm{AMATT}^+}(d)$ is identified when the sub-treatments are observed and under Assumption \ref{ass:latent-subtreatment-independence}.  It is equal to a weighted average of the comparison of mean outcomes among congruent sub-treatment vectors corresponding to the aggregate amount of the aggregated treatment being equal to $d$ or $d-1$.  The weights are the same as the scaled product weights suggested in the discussion of $\textrm{AMATT}^{+}(d)$ in Equation \eqref{eqn:congruent-product-weight} above.

\section{Proofs} \label{app:proofs}

This section provides proofs for all the results in the paper.


\subsection{Proofs of Results from Section \ref{subsec:causal-framework}} \label{app:proofs-causal-framework}

The following proposition relates the observed difference of mean outcomes across sub-treatment vectors with adjacent levels of the aggregated treatment to causal quantities. 
\begin{proposition} \label{prop:comparison-of-subtreatments} 
Under Assumption \ref{ass:sutva2}, the difference in mean outcomes between sub-treatment vectors at adjacent levels of the aggregated treatment can be decomposed as
\begin{align*}
    \E[Y|S=s_d] - \E[Y|S=s_{d-1}] &= \mathrm{MATT}(s_d,s_{d-1}) + B(s_d, s_{d-1})
\end{align*}
where
\begin{align*}
    B(s_d,s_{d-1}) &:= \E[Y(s_{d-1})|S=s_d ] - \E[Y(s_{d-1})|S=s_{d-1}] 
\end{align*}
which is a selection bias term.  If, in addition, Assumption \ref{ass:unconfoundedness} (or Assumption \ref{ass:local-unconfoundedness}) holds, then $B(s_d,s_{d-1}) = 0$, so that 
\begin{align*}
    \mathrm{MATT}(s_d,s_{d-1}) = \E[Y|S = s_d] - \E[Y|S=s_{d-1}]
\end{align*}
\end{proposition}

The first part of Proposition \ref{prop:comparison-of-subtreatments} shows that the comparison of means between sub-treatment vectors with adjacent levels of the aggregated treatment can be decomposed into a marginal average treatment effect on the treated term and a selection bias term. The second part shows that, under Assumption \ref{ass:unconfoundedness}, the same comparison of means is equal to $\textrm{MATT}(s_d,s_{d-1})$. 
The result in Proposition \ref{prop:comparison-of-subtreatments} holds irrespective of whether $s_d$ and $s_{d-1}$ are congruent sub-treatments. 

\begin{proof}[\textbf{Proof of Proposition \ref{prop:comparison-of-subtreatments}}]
    To show the first part, notice that
    \begin{align*}
        \E[Y|S=s_d] - \E[Y|S=s_{d-1}] &= \E[Y(s_d) - Y(s_{d-1}) | S=s_d ] \\ 
        & \hspace{25pt} + \Big( \E[Y(s_{d-1})|S=s_d ] - \E[Y(s_{d-1})|S=s_{d-1}] \Big) \\
        &= \mathrm{MATT}(s_d, s_{d-1}) + B(s_d, s_{d-1})
    \end{align*}
    where the first equality holds by writing observed outcomes in terms of their corresponding potential outcomes, and then by adding and subtracting $\E[Y(s_{d-1}) | S=s_d ]$; and the second equality holds by the definitions of $\textrm{MATT}(s_d,s_{d-1})$ and $B(s_d, s_{d-1})$.

    For the second part, notice that
    \begin{align*}
        \mathrm{MATT}(s_d, s_{d-1}) &= \E[Y(s_d) | S=s_d] - \E[Y(s_{d-1}) | S=s_d] \\
        &= \E[Y(s_d) | S=s_d, D=d~\textrm{or}~D=d-1] \\
        &\hspace{50pt} - \E[Y(s_{d-1}) | S=s_d, D=d~\textrm{or}~D=d-1] \\
        &= \E[Y(s_d) | S=s_d, D=d~\textrm{or}~D=d-1] \\
        &\hspace{50pt} - \E[Y(s_{d-1}) | S=s_{d-1}, D=d~\textrm{or}~D=d-1] \\
        &= \E[Y(s_d) | S=s_d] - \E[Y(s_{d-1}) | S=s_{d-1}] \\
        &= \E[Y|S=s_d] - \E[Y|S=s_{d-1}]
    \end{align*}
    where the first equality hold by the definition of $\textrm{MATT}(s_d,s_{d-1})$; the second equality holds because $D$ is fully determined by $S$ (and $S$ being equal to $s_d$ or $s_{d-1}$ implies that $D$ is either equal to $d$ or $d-1$); the third equality holds by Assumption \ref{ass:unconfoundedness} (or Assumption \ref{ass:local-unconfoundedness}); the fourth equality holds again because $S$ fully determines $D$; and the last equality holds by writing potential outcomes in terms of their observed counterparts.
\end{proof}

\bigskip

\begin{proposition} \label{prop:comparison-of-neighborly-subtreatments} 
Under Assumption \ref{ass:sutva2}, the difference in mean outcomes between sub-treatment vectors at identical levels of the aggregated treatment can be decomposed as
\begin{align*}
    \E[Y|S=s_d] - \E[Y|S=s_{d}'] &= \mathrm{SATT}(s_d,s_{d}') + SB(s_d, s_{d}')
\end{align*}
where
\begin{align*} 
    SB(s_d,s_{d}') &:= \E[Y(s_d')|S=s_d] - \E[Y(s_d')|S=s_d'] 
\end{align*}
which is a selection bias term.  If, in addition, Assumption \ref{ass:unconfoundedness} (or Assumption \ref{ass:local-unconfoundedness}) holds, then $SB(s_d,s_{d}') = 0$, so that $\mathrm{SATT}$ is identified and is given by
\begin{align*}
    \mathrm{SATT}(s_d,s_{d'}) = \E[Y|S = s_d] - \E[Y|S=s_{d}']
\end{align*}
\end{proposition}
\noindent The proof is provided in \ref{supapp:proofs-identification-problems} of the Supplementary Appendix (\cite{CCCDSupp2025}).


\subsection{Proofs of Results from Section \ref{sec:identification-with-unobserved-subtreatments}} \label{app:proofs-identification-problems}

\begin{proposition} \label{prop:marginal-decomposition} For any weighting function $w(s_d,s_{d-1})$ such that
\begin{itemize}
    \item [(i)] $\sum_{s_d \in \mathcal{S}_d} w(s_d,s_{d-1}) = \P(S=s_{d-1}|D=d-1)$
    \item [(ii)] $\sum_{s_{d-1} \in \mathcal{S}_{d-1}} w(s_d,s_{d-1}) = \P(S=s_d|D=d)$
\end{itemize}
\begin{align*}
    \E[Y|D=d] - \E[Y|D=d-1] &= \sum_{(s_{d-1}, s_{d}) \in \mathcal{M}^+(d) } w(s_d,s_{d-1}) \cdot \Big( \E[Y|S=s_d] - \E[Y|S=s_{d-1}]\Big) \\
    &+\sum_{(s_{d-1}, s_{d}) \in \mathcal{M}^-(d) } w(s_d,s_{d-1}) \cdot \Big( \E[Y|S=s_d] - \E[Y|S=s_{d-1}]\Big)
\end{align*}
\end{proposition}

\begin{proof}[\textbf{Proof of Proposition \ref{prop:marginal-decomposition}}] 
Notice that, for any weighting function $w(s_d,s_{d-1})$ that satisfies properties (i) and (ii), 
\begin{align*} 
    &\E[Y|D=d] - \E[Y|D=d-1] \\
    &\hspace{50pt} = \sum_{s_d \in \ST_d} \P(S=s_d|D=d) \cdot \E[Y|S=s_d] - \sum_{s_{d-1} \in \ST_{d-1}} \P(S=s_{d-1}|D=d-1) \cdot \E[Y|S=s_{d-1}] \\
    &\hspace{50pt} = \sum_{s_d \in \ST_d} \sum_{s_{d-1} \in \ST_{d-1}} w(s_d,s_{d-1}) \cdot \E[Y|S=s_d] - \sum_{s_{d-1} \in \mathcal{S}_{d-1}} \sum_{s_d \in \ST_d} w(s_d, s_{d-1}) \cdot \E[Y|S=s_{d-1}] \\
    &\hspace{50pt} = \sum_{s_d \in \ST_d} \sum_{s_{d-1} \in \ST_{d-1}} w(s_d,s_{d-1}) \cdot \Big(\E[Y|S=s_d]-\E[Y|S=s_{d-1}]\Big) \\
    & \hspace{50pt} = \sum_{(s_{d-1}, s_{d}) \in \mathcal{M}^+(d) } w(s_d,s_{d-1}) \cdot \Big( \E[Y|S=s_d] - \E[Y|S=s_{d-1}]\Big) \\
    & \hspace{75pt} +\sum_{(s_{d-1}, s_{d}) \in \mathcal{M}^-(d) } w(s_d,s_{d-1}) \cdot \Big( \E[Y|S=s_d] - \E[Y|S=s_{d-1}]\Big)
\end{align*}
where the first equality holds by applying the law of iterated expectations to each term (and because the aggregated treatment is fully determined by the sub-treatment); the second equality holds from properties (i) and (ii) of the weights, and because the first term does not depend on $s_{d-1}$ and the second term does not depend on $s_d$; the third equality holds by re-arranging the summations and combining terms; and the last equality holds by separating the summation in the previous line among congruent and incongruent comparisons across sub-treatment vectors.
\end{proof}

\bigskip


\begin{proof}[\textbf{Proof of Theorem \ref{thm:causal-aggregate-marginal-comparison}}] The result holds immediately from Propositions \ref{prop:comparison-of-subtreatments} and \ref{prop:marginal-decomposition}. \end{proof}


\subsection{Proofs of Results from Section \ref{sec:incongruent-matts}}

\begin{lemma}[Congruent Sub-treatment Existence] \label{lem:congruent-st-existence}
    For every $s_{d-1} \in \ST_{d-1}$ such that $d>0$, there always exists a vector $s_d \in \ST_d$ such that $s_d \cst s_{d-1}$. That is, for non-empty $\ST_{d-1}$ there will exist an $s_d \in \ST_d$ that is congruent with $s_{d-1}$.  
\end{lemma}

\noindent The proof is provided in \ref{supapp:proofs-identification-problems} of the Supplementary Appendix (\cite{CCCDSupp2025}).


\bigskip

\begin{lemma}[Unit Exchange Property] \label{lem:unit-exchange}
    Let $s_d \in S_d \subset S$, which satisfies the $\ell_1$ norm $||s_d||_1=d$ and $S \subseteq \mathbb{Z}^{K}_{\geq 0}$, which is the vector space of non-negative integers with dimension $K$. If $K \geq 2$ and $d>0$, then there exists $s_d' \in S_d$ with coordinates $j$ and $l$ where $j \neq l$ such that 
    \[
    s_d = s_d' + 1_j - 1_l,
    \]
    where $1_j$ and $1_l$ are the unit vectors for coordinates (sub-treatments) $j$ and $l$. 
\end{lemma}

\noindent The proof is provided in \ref{supapp:proofs-identification-problems} of the Supplementary Appendix (\cite{CCCDSupp2025}). 

\bigskip

\begin{lemma}[Productive Chain Progression] \label{lem:productive-progression}    
    Take any $s_{d},s_{d}' \in \ST_{d}$ with $d>0$ such that $||s_{d} - s_{d}'||_1 > 2$,  where at least for some $j^{th}$ and $l^{th}$ coordinates $s_{d,j} < s_{d,j}'$ and $s_{d,l} > s_{d,l}'$. There exists a vector $s_{d}^{\star} \in S_{d}$ such that $s_{d}^{\star} = s_{d} + 1_{j} - 1_{l}$ for some coordinates $j$ and $l$, and $||s_{d} - s_{d}'||_1 > ||s_{d}^{\star} - s_{d}'||_1 \geq 2$. That is, there always exists a vector $s_{d}^{\star}$ which is a unit-exchange towards the terminating vector $s_{d}'$ which strictly decreases in distance $||s_{d}^{\star} - s_{d}'||_1$.
\end{lemma}

\noindent The proof is provided in \ref{supapp:proofs-identification-problems} of the Supplementary Appendix (\cite{CCCDSupp2025}). 

\bigskip

\begin{proposition}[Chain Algorithm] \label{prop:chain-algorithm}
    Let $s_{d}, s_{d}' \in \ST_d \subset \mathbb{Z}_{\geq 0}^{K}$ for $d>0$ and $K \geq 2$. For some $B \in \mathbb{Z}^{+}$, there is a finite sequence of vectors 
    \[
    x^{(0)}=s_d, x^{(1)}, \ldots, x^{(B-1)}, x^{(B)}=s_d'
    \]
    such that each vector in the sequence, indexed by $b$, is a unit exchange 
    \[
    x^{(b+1)} = x^{(b)} + 1_j - 1_l
    \]
    of neighboring vectors which makes strict progress towards $s_d'$, where $1_j$ and $1_l$ are unit vectors for the $j^{th}$ and $l^{th}$ coordinate.
\end{proposition}

\noindent The proof is provided in \ref{supapp:proofs-identification-problems} of the Supplementary Appendix (\cite{CCCDSupp2025}). 


\bigskip

\begin{proof}[\textbf{Proof of Proposition \ref{prop:incongruent-and-substitution-effects}}]
    Given no (local) selection (Assumption \ref{ass:unconfoundedness} or \ref{ass:local-unconfoundedness}), we know by Proposition \ref{prop:comparison-of-subtreatments} that disaggregate comparisons across average outcomes between sub-treatments at adjacent levels of total treatment are $\textrm{MATT}$'s, and by Proposition \ref{prop:comparison-of-neighborly-subtreatments} that disaggregate comparisons across average outcomes between sub-treatments at identical levels of total treatment are $\textrm{SATT}$'s. Under this condition, suppose that $d>0$ and $K\geq2$ so that $S_{d}$ is not empty. Hence, $\mathcal{M}^{+}(d)$ is not empty, but $\mathcal{M}^{-}(d)$ may or may not be empty. If $\mathcal{M}^{-}(d)$ is empty, then the claim holds vacuously. So, we proceed with non-empty $\mathcal{M}^{-}(d)$. 

    Take any pair $(s_d, s_{d-1}) \in \mathcal{M}^{-}(d)$ and, by Lemma \ref{lem:congruent-st-existence}, any congruent $s_{d-1}'$ of $s_d$ in $S_{d-1}$. Denote $m(s) := \E[Y|S=s]$ for any $s \in S$. Then from Proposition \ref{prop:comparison-of-subtreatments} the incongruent $\textrm{MATT}$ can be written as: 
    \begin{align*} 
        \mathrm{MATT}^{-}(s_d,s_{d-1}) &= m(s_d) - m(s_{d-1}) \\
        &= m(s_d) - m(s_{d-1}')  + m(s_{d-1}') - m(s_{d-1})  \\
        &= \mathrm{MATT}^{+}(s_d, s_{d-1}') + \big( m(s_{d-1}') - m(s_{d-1}) \big) 
    \end{align*}
    where the first equality is the definition of $\textrm{MATT}$; the second equality holds by addition and subtraction of $m(s_{d-1}')$; and the third equality holds by the definition of $\textrm{MATT}^{+}$. 

    From here, there are two cases: 
    \begin{enumerate}
        \item $s_{d-1}'$ is a unit exchange of $s_{d-1}$, $s_{d-1}' = s_{d-1} + 1_j - 1_l$, where $1_j$ and $1_l$ are unit vectors for some $j$ and $l$ coordinates;
        \item $s_{d-1}'$ is not a unit exchange of $s_{d-1}$, $s_{d-1}' \neq s_{d-1} + 1_j - 1_l$, for any $j$ and $l$ coordinates.
    \end{enumerate}

    \noindent
    \textit{Case 1:} If $s_{d-1}'$ is a unit exchange of $s_{d-1}$, then by Lemma \ref{lem:unit-exchange}, $s_{d-1}' = s_{d-1} + 1_j - 1_l$. Consequently, $m(s_{d-1}') - m(s_{d-1}) = \textrm{SATT}(s_{d-1}', s_{d-1})$ by the definition of $\textrm{SATT}$. Thus, the claim holds: 
    \begin{align*}
        \mathrm{MATT}^{+}(s_d, s_{d-1}') + m(s_{d-1}') - m(s_{d-1}) &= \mathrm{MATT}^{+}(s_d, s_{d-1}') + \mathrm{SATT}(s_{d-1}', s_{d-1})
    \end{align*}
    by definition of $\textrm{SATT}$ in Proposition \ref{prop:comparison-of-neighborly-subtreatments}, the substitution average treatment effect on the treated. 

    \bigskip
    \noindent
    \textit{Case 2:} Suppose $s_{d-1}'$ is not a unit exchange of $s_{d-1}$. That is, $s_{d-1}' \neq s_{d-1} + 1_j - 1_l$, for any $j$ and $l$ coordinates. 
    
    Next, our goal is to express the difference $m(s_{d-1}') - m(s_{d-1})$ as unit exchanges between $m(s_{d-1}')$ to $m(s_{d-1})$. To do this, we need a set of vectors, chained by the unit exchange property, that begins at $s_{d-1}'$ and eventually terminates at $s_{d-1}$. We approach this task algorithmically. To build such a chain, we know by Lemma \ref{lem:unit-exchange} that there exists a vector $s_{d-1}^{(1)} \in S_{d-1}$ which is a unit exchange of $s_{d-1}'$ such that $s_{d-1}' = s_{d-1}^{(1)} + 1_{j'} - 1_{l'}$ for some coordinates $j'$ and $l'$.   Next, by Lemma \ref{lem:productive-progression}, there always exists a vector that makes productive progress from $s_{d-1}'$ towards the terminating vector $s_{d-1}$. By Proposition \ref{prop:chain-algorithm}, these conditions guarantee the procedure terminates at $s_{d-1}$ in finitely many steps without any repeating patterns. 
    
    In the first step of this procedure, we may write: 
    \begin{align*}
        \mathrm{MATT}^{-}(s_d,s_{d-1}) &= \mathrm{MATT}^{+}(s_d, s_{d-1}') + m(s_{d-1}') - m(s_{d-1}) \\ 
        &= \mathrm{MATT}^{+}(s_d, s_{d-1}') + \big( m(s_{d-1}') - m(s_{d-1}^{(1)}) \big) + \big( m(s_{d-1}^{(1)}) - m(s_{d-1}) \big) \\ 
        &= \mathrm{MATT}^{+}(s_d, s_{d-1}') + \mathrm{SATT}(s_{d-1}', s_{d-1}^{(1)}) + \big( m(s_{d-1}^{(1)}) - m(s_{d-1}) \big) 
    \end{align*}
    where the first equality holds from above; the second equality holds by the addition and subtraction of $m(s_{d-1}^{(1)})$; and the third equality holds by the definition of $\textrm{SATT}$ in Proposition \ref{prop:comparison-of-neighborly-subtreatments}. If $s_{d-1}^{(1)}$ is a unit exchange of $s_{d-1}$ also, $s_{d-1}^{(1)} = s_{d-1} + 1_{j''} - 1_{l''}$ for coordinates $j''$ and $l''$, then the original claim holds: 
    \begin{align*}
        \mathrm{MATT}^{-}(s_d,s_{d-1}) &= \mathrm{MATT}^{+}(s_d, s_{d-1}') + \mathrm{SATT}(s_{d-1}',s_{d-1}^{(1)}) + \mathrm{SATT}(s_{d-1}^{(1)},s_{d-1}) 
    \end{align*}
    by the definition of $\textrm{SATT}$ in Proposition \ref{prop:comparison-of-neighborly-subtreatments}.
    Otherwise, we continue this process by iteratively applying the unit exchange property, Lemma \ref{lem:unit-exchange}, that satisfies productive progression towards the terminating vector, Lemma \ref{lem:productive-progression}, at each step $b$ until a chain $\phi := \{ s_{d-1}', s_{d-1}^{(1)}, \ldots, s_{d-1}^{(B-1)}, s_{d-1} \} $ of size $|\phi|=B \in \mathbb{Z}^{+}$ is formed. This implies that we may write the incongruent parameter as a telescoping sum of all $m(s_{d-1}^{(b)})$ in the chain: 
    \begin{align*}
        \mathrm{MATT}^{-}(s_d,s_{d-1}) &= \mathrm{MATT}^{+}(s_d, s_{d-1}') + \sum_{b=0}^{B-1} \mathrm{SATT}(s_{\phi, \; d-1}^{(b)},s_{\phi, \; d-1}^{(b+1)})
    \end{align*}
    where $s_{\phi}^{(b)}$ is linked in the chain $\phi$ to the vector $s_{\phi}^{(b+1)}$ by the unit exchange property in Lemma \ref{lem:unit-exchange}. Hence, the claim holds. Note that Proposition \ref{prop:chain-algorithm} permits a set of possible chains denoted by $\mathcal{C}(s_{d-1}', s_{d-1})$, where $\phi$ is possibly one element of many in this set.
\end{proof}

\bigskip

\begin{lemma}[Single Pairwise Congruent Sub-treatment Vector] \label{lem:pairwise-congruent-vectors}
    For any $s_{d-1}, s_{d-1}' \in \ST_{d-1}$ such that $s_{d-1} = s_{d-1}' + 1_j - 1_l$ for some coordinates $j$ and $l$, which is a unit exchange of $s_{d-1}'$, there must exist one $s_{d} \in \ST_{d}$ that is congruent with $s_{d-1}$ and $s_{d-1}'$. 
\end{lemma}

\noindent The proof is provided in \ref{supapp:proofs-identification-problems} of the Supplementary Appendix (\cite{CCCDSupp2025}).

\bigskip

\begin{proof}[\textbf{Proof of Proposition \ref{prop:substitution-effect-decomp}}]
    Given no (local) selection (Assumption \ref{ass:unconfoundedness} or \ref{ass:local-unconfoundedness}), we know by Proposition \ref{prop:comparison-of-subtreatments} that the disaggregate comparisons of average outcomes between sub-treatments at adjacent levels of aggregated treatment are $\textrm{MATT}$'s, and by Proposition \ref{prop:comparison-of-neighborly-subtreatments} that disaggregate comparisons of average outcomes between sub-treatments at identical levels of aggregated treatment are $\textrm{SATT}$'s. Under this condition, take any $s_{d-1}, s_{d-1}' \in \ST_{d-1}$ for $d>0$ such that $\|s_{d-1} - s_{d-1}'\|_1=2$ and $s_{d-1}$ is a unit exchange of $s_{d-1}'$. That is, by Lemma \ref{lem:unit-exchange}, $s_{d-1} = s_{d-1}' + 1_j - 1_l$, where $1_j$ and $1_l$ are unit vectors for some coordinates $j$ and $l$. 

    Next, by Lemma \ref{lem:pairwise-congruent-vectors}, there exists a unique sub-treatment vector $s_d \in \ST_d$ which is congruent to both $s_{d-1}$ and $s_{d-1}'$. Now define $m(s) := \E[Y|S=s]$ for any $s \in \ST$. See that the substitution average treatment effect on the treated may be expressed as: 
    \begin{align*}
        \textrm{SATT}(s_{d-1},s_{d-1}') &= m(s_{d-1}) - m(s_{d-1}') \\ 
        &= \big( m(s_{d}') - m(s_{d-1}') \big) + \big( m(s_{d-1}) - m(s_{d}') \big) \\ 
        &= \textrm{MATT}^{+}(s_d',s_{d-1}') - \textrm{MATT}^{+}(s_{d}',s_{d-1})
    \end{align*}
    where the first quality is the definition of $\textrm{SATT}$ in Proposition \ref{prop:comparison-of-neighborly-subtreatments} using the $m(\cdot)$ notation; the second equality holds by the addition and subtraction of $m(s_d')$; and the final equality holds by the definitions of the congruent parameters $\textrm{MATT}^{+}$. Thus, any $\textrm{SATT}$ between two sub-treatment vectors that share a unit exchange is equivalent to the difference between two congruent $\textrm{MATT}^{+}(s_d', \cdot)$.
\end{proof}


\subsection{Proofs of Results from Section \ref{sec:non-unique-weights}}

\subsubsection{Proofs of Results from Section \ref{subsec:number-of-ST-and-incongruent-comparisons}}

\begin{lemma}[Double Binomial Reduction] \label{lem:counting-equality}
    The sum $\sum_{d=1}^{K} \binom{K}{d} \cdot \binom{K}{d-1}$ is equivalent to $\binom{2K}{K-1}$, for all $K \in \mathbb{Z}$.  
\end{lemma}

\noindent The proof is provided in \ref{supapp:proofs-identification-problems} of the Supplementary Appendix (\cite{CCCDSupp2025}).

\bigskip

\begin{proposition}[Number of Congruent and Incongruent Contrasts with Binary Sub-treatments] \label{prop:congruent-incongruent-counts-binary-subtreatments}
    For a given number of binary sub-treatments $K>1$, there are $\sum_{d=1}^K \binom{K}{d} \binom{K}{d-1} = \binom{2K}{K-1}$ distinct contrasts, and the proportion of incongruent contrasts will increase as $K$ increases - illustrated in Figure \ref{fig:proportion-of-incongruent-subtreatments}. The number of congruent disaggregated contrasts present is $\sum_{d=1}^K K \cdot \binom{K-1}{d-1} = K \cdot 2^{K-1}$, and can be presented as a row-scaled Pascal's (Hui-Khayyam's) Triangle. The number of incongruent disaggregated contrasts present is the difference between total contrasts and congruent contrasts. 
\end{proposition}

\begin{proof}[\textbf{Proof of Proposition \ref{prop:congruent-incongruent-counts-binary-subtreatments}}]
    Given $K>1$ binary sub-treatments, we show that for all $d \in \D$: (i) the total number of distinct disaggregated contrasts is $\sum_{d=1}^K \binom{K}{d} \binom{K}{d-1} = \binom{2K}{K-1}$; and (ii) the number of congruent disaggregated contrasts is $\sum_{d=1}^K K \cdot \binom{K-1}{d-1}$.

    \begin{enumerate}
        \item \textit{Total distinct contrasts:} 
        First, we show that the number of $K$-tuples that sum to $D=d$ is $\binom{K}{d}$. The size of a tuple is $K$. Since each variable is binary, each element in the $K$-tuple is either zero or one. For each tuple to sum to $D=d$, there must be exactly $d$ ones, and $K-d$ zeros. This amounts to choosing how many ways we can put $d$ ones among $K$ empty positions. Hence, this is a problem of the form ``how to choose $d$ objects from $K$ total objects'', which is known to be determined by the binomial coefficient, $\binom{K}{d}$. This implies, for $K$ binary sub-treatments, that the number of distinct sub-treatment vectors that lie in the aggregate set $\ST_d$ is equal to $\binom{K}{d}$, or $|\ST_d| = \binom{K}{d}$.
        
        Similarly, if we wanted to know how many $K$-tuples would sum to $D=d-1$, we would write $\binom{K}{d-1}$. Recall that marginal contrasts are made between adjacent aggregate sets, $\ST_d$ and $\ST_{d-1}$. So, for every $d \in \D_{>0}$, there are $|\ST_d| \cdot |\ST_{d-1}| = \binom{K}{d} \cdot \binom{K}{d-1}$ possible comparisons. Hence, the total number of distinct contrasts that can be made is the sum across all $d \in \D$, $\sum_{d=1}^{K} \binom{K}{d} \cdot \binom{K}{d-1}$, as desired. And by Lemma \ref{lem:counting-equality}, this sum $\sum_{d=1}^{K} \binom{K}{d} \cdot \binom{K}{d-1} = \binom{2K}{K-1}$. 

        \item \textit{Congruent distinct contrasts:} 
        Next, we aim to find the number of congruent contrasts out of the total number of contrasts. We demonstrate that the number of congruent contrasts for $K$ binary sub-treatments is $\sum_{d=1}^{K} K \cdot \binom{K-1}{d-1} = K \cdot 2^{K-1}$. 

        Recall that congruent vectors of treatment must satisfy $s' = s + 1_j$, for some $j^{th}$ component. Since the sub-treatments are binary, and thus the elements in each position may not exceed one, this implies that one position must remain fixed while all other positions may vary because all positions for ones must be identical except for one position in the binary case. That is, there are $K-1$ available positions/elements that we are free to vary as long as the tuple sums to $d$. Hence, at each $d$, we are counting the number of ways we can put $d$ ones in $K-1$ positions, $\binom{K-1}{d}$. We can do this for each position in the tuple, which is of length $K$. Therefore, at each $d$, the number of congruent contrasts is $K \cdot \binom{K-1}{d}$. The total number of congruent contrasts can be summed across all $d \in \D$, $\sum_{d=0}^K K \cdot \binom{K-1}{d} = K \cdot \sum_{d=0}^K \binom{K-1}{d}$. 
        Notice that by the binomial theorem, this sum of binomial coefficients can be written as $\sum_{d=0}^K \binom{K-1}{d} = 2^{K-1}$. Hence, the number of congruent contrasts is the product: $K \cdot 2^{K-1}$. \qedhere
    \end{enumerate}
\end{proof}

\bigskip

\begin{corollary}[Asymptotic Congruity Is Proportionally Zero] \label{cor:asymptotic-congruent-proportion-result}
    As the number of binary sub-treatments increases, the proportion of congruent contrasts in the marginal set approaches zero. That is, 
    \begin{align*}
        \lim_{K \rightarrow \infty} \frac{\big| \mathcal{M}^{+} \big| }{ \big| \mathcal{M} \big|} &= \lim_{K \rightarrow \infty} \frac{\big| \mathcal{M}^{} \big| - \big| \mathcal{M}^{-} \big|}{ \big| \mathcal{M} \big| } = 0 
    \end{align*}
    In other words, as the number of sub-treatments grows, the proportion of incongruity dominates. 
\end{corollary}

\begin{proof}[\textbf{Proof of Corollary \ref{cor:asymptotic-congruent-proportion-result}}]
    Here we prove that as we let the number of binary sub-treatments grow to infinity, the proportion of congruent contrasts approaches zero. That is, we show that: $\lim_{K \rightarrow \infty} \frac{\big| \mathcal{M}^{+} \big| }{ \big| \mathcal{M} \big|} = 0$.
    If the results in Proposition \ref{prop:congruent-incongruent-counts-binary-subtreatments} holds, the we can write: 
    \begin{align*}
        \lim_{K \rightarrow \infty} \frac{\big| \mathcal{M}^{+} \big| }{ \big| \mathcal{M} \big|} &= \lim_{K \rightarrow \infty} \frac{K \cdot \sum_{d=1}^K \binom{K-1}{d-1}}{\sum_{d=1}^K \binom{K}{d} \binom{K}{d-1}} 
        = \lim_{K \rightarrow \infty} \frac{K \cdot 2^{K-1}}{ \binom{2K}{K-1} } 
    \end{align*}
    which are the results from Proposition \ref{prop:congruent-incongruent-counts-binary-subtreatments}.
    Next, see that the denominator can be rewritten as: 
    \begin{align*}
        \binom{2K}{K-1} &= \frac{(2K)!}{(K-1)! (2K-(K-1))!} = \left( \frac{K}{K+1} \right) \frac{(2K)!}{K! (2K-K)!} = \left( \frac{K}{K+1} \right) \cdot \binom{2K}{K}
    \end{align*}
    This implies that we are trying to find the following limit: 
    \begin{align}
        \lim_{K \rightarrow \infty} & \frac{K \cdot 2^{K-1}}{ \left( \frac{K}{K+1} \right) \cdot \binom{2K}{K} } \label{eq:corA1.limit} 
    \end{align}
    Before taking the limit, notice that we can restate the denominator by Stirling's approximation formula for factorials. First, note that: 
    \begin{align}
    (2K)! &\sim \sqrt{2 \pi (2K)} \cdot \left( \frac{2K}{e} \right)^{2K} 
    = \left( \frac{\sqrt{\pi} \cdot 2^{2K+1} \cdot K^{\frac{4K+1}{2}}}{e^{2K}} \right) \label{eq:corA1.A} 
    \end{align}
    and similarly, 
    \begin{align}
        K! &\sim \sqrt{2 \pi K} \cdot \left( \frac{K}{e} \right)^{K} 
        = \left( \frac{\sqrt{2 \pi} \cdot K^{\frac{2K+1}{2}}}{e^{K}} \right) 
        \label{eq:corA1.B} 
    \end{align}
    where Stirling's formula states, for some $n>0$, that $n! \sim \sqrt{2 \pi n} \cdot \left( \frac{n}{e} \right)^{n}$, and $e$ denotes Euler's constant. Both (\ref{eq:corA1.A}) and (\ref{eq:corA1.B}) imply that the denominator of the limit (\ref{eq:corA1.limit}) can be approximated well by the following expression:
    \begin{align}
        \binom{2K}{K} &= \frac{(2K)!}{K! K!} \sim \frac{ \frac{ \sqrt{\pi} \cdot 2^{2K+1} \cdot K^{(4K+1)/2} }{e^{2K}} }{ \left( \frac{ \sqrt{2 \pi} K^{(2K+1)/2} }{e^K} \right)^2 } = \frac{ \sqrt{\pi} \cdot 2^{2K+1} \cdot K^{(4K+1)/2} }{ 2 \pi \cdot K^{2K+1} } = \pi^{-1/2} \cdot 2^{2K} \cdot K^{-1/2} \label{eq:corA1.C} 
    \end{align}
    Therefore, by (\ref{eq:corA1.C}) above, the limit (\ref{eq:corA1.limit}) becomes: 
    \begin{align*}
        \lim_{K \rightarrow \infty} \frac{K \cdot 2^{K-1}}{ \left( \frac{K}{K+1} \right) \cdot \binom{2K}{K} } &\sim \lim_{K \rightarrow \infty} \frac{ K \cdot 2^{2K-1} }{ \pi^{-1/2} \cdot 2^{2K} \cdot K^{-1/2} } \\
        &= \lim_{K \rightarrow \infty} \sqrt{\pi} \cdot \frac{ K (K+1) }{K^{1/2}} \cdot 2^{(K-1) - 2K} \\ 
        &= \sqrt{\pi} \cdot \lim_{K \rightarrow \infty} \frac{ K^{3/2} + K^{1/2} }{ 2^{K+1} } = 0
    \end{align*}
    where the first line holds by applying Stirling's formula; the second and third lines hold by algebra; and the fourth equality holds since the numerator of the limit grows at the polynomial rate and the denominator grows at the exponential rate. This implies that as the number of binary sub-treatments increases, the proportion of congruent contrasts tends to zero, as claimed. 
\end{proof}

\bigskip

\begin{proposition}[Number of Congruent and Incongruent Contrasts with Trinary Sub-treatments] \label{prop:congruent-incongruent-counts-trinary-subtreatments}
    For a given number of trinary sub-treatments $K>1$ with common support $\{0,1,2\}$, there are $\sum_{d}^K \binom{K}{d}_2 = \sum_{d=1}^K \sum_{r=0}^{\lfloor d/2 \rfloor} \binom{K}{r} \cdot \binom{K-r}{d-2r}$ distinct contrasts, and the proportion of incongruent contrasts will increase as $K$ increases - illustrated in Figure \ref{fig:sub2}. The number of congruent disaggregated contrasts present is the row-sum of a row-scaled extended trinomial triangle, whose elements are the sum of the top-three adjacent elements. The number of incongruent disaggregated contrasts present is the difference between total contrasts and congruent contrasts.  
\end{proposition}

\noindent The proof is provided in the Supplementary Appendix (\cite{CCCDSupp2025}).


\subsubsection{Proofs of Results from Section \ref{subsec:auxiliary-assumptions-to-rule-out-incongruency}}

\begin{proposition}[Homogeneity of Effects] \label{prop:homogeneity-of-effects} 
    Under Assumptions \ref{ass:sutva2}-\ref{ass:homogeniety}, and for any $d \in \D_{>0}$, 
    \begin{align*}
        \E[Y|D=d] - \E[Y|D=d-1] = \beta_d.
    \end{align*} 
\end{proposition}


\begin{proof}[\textbf{Proof of Proposition \ref{prop:homogeneity-of-effects}}]
    Under Assumption \ref{ass:homogeniety}, we have that $\textrm{MATT}^+(s_d, s_{d-1}) = \beta_d$ for all $(s_d, s_{d-1}) \in \mathcal{M}^+(d)$.  Furthermore, by Proposition \ref{prop:substitution-effect-decomp}, for any $(s_d,s_{d-1}) \in \mathcal{M}^{-}(d)$ and any $\phi \in \mathcal{C}(s_d,s_{d-1})$,
    \begin{align*}
        \mathrm{MATT}^{-}(d) &= \mathrm{MATT}^+(s_d,s_{d-1}') + \sum_{ b=0 }^{B-1} \mathrm{MATT}^{+}(s_{\phi, \; d}^{(b)},s_{\phi, \; d-1}^{(b)}) - \sum_{ b=0 }^{B-1} \mathrm{MATT}^{+}(s_{\phi, \; d}^{(b)},s_{\phi, \; d-1}^{(b+1)}) \\ 
        &= \beta_d + |B| \cdot \beta_d - |B| \cdot \beta_d = \beta_d
    \end{align*}
    Hence, all marginal effects are $\beta_d$ at $D=d$. Thus, from Theorem \ref{thm:causal-aggregate-marginal-comparison},  $\E[Y|D=d] - \E[Y|D=d-1] = \beta_d$, as desired.
\end{proof}

\bigskip

\begin{proposition}[No Incongruent Behavior] \label{prop:no-incongruent-behavior}
    Under Assumptions \ref{ass:sutva2}, \ref{ass:unconfoundedness}, \ref{ass:no-sorting-on-D}, and \ref{ass:no-incongruity}, no incongruent comparisons are present in the aggregate contrast, $\E[Y|D=d] - \E[Y|D=d-1]$. 
\end{proposition}


\begin{proof}[\textbf{Proof of Proposition \ref{prop:no-incongruent-behavior}}]
    Under Assumptions \ref{ass:no-sorting-on-D} and \ref{ass:no-incongruity}, there exists a latent treatment joint distribution $F(S(d),S(d-1))$ of the sub-treatments where for all $d \in \D$: 
    \begin{align*}
        \P\big(S(d)=s_d,S(d-1)=s_{d-1} \big| D \in \{d, d-1\}\big) &= 0, ~\textrm{if } (s_{d}, s_{d-1}) \in \mathcal{M}^{-}(d) \\
        \P\big(S(d)=s_d,S(d-1)=s_{d-1} \big| D \in \{d, d-1\}\big) &\geq 0, ~\textrm{if } (s_{d}, s_{d-1}) \in \mathcal{M}^{+}(d)
    \end{align*}
    Under conditions (i) and (ii), $w(s_d, s_{d-1}) = \P(S(d) = s_d, S(d-1)=s_{d-1}|D \in \{d,d-1\})$. 
    
    Next, we show that the above probabilities satisfy the criteria of Proposition \ref{prop:marginal-decomposition} and put no weight on incongruent comparisons. To show this, first note that the sum: 
    \begin{align*}
        & \sum_{s_{d-1} \in \ST_{d-1}} \P\big(S(d)=s_d,S(d-1)=s_{d-1} \big| D \in \{d, d-1\}\big) \\ 
        & \hspace{10mm} = \P\big(S(d)=s_d \big| D \in \{d, d-1\}\big) = \P\big(S(d)=s_d \big| D=d \big) = \P\big(S=s_d \big| D=d \big),
    \end{align*}
    where the first equality holds by the law of total probability; the second equality holds by Assumption \ref{ass:no-sorting-on-D}; and the last equality holds by writing potential outcomes in terms of their observed counterparts.
    Similarly, we have: 
    \begin{align*}
        & \sum_{s_{d} \in \ST_{d}} \P\big(S(d)=s_d,S(d-1)=s_{d-1} \big| D \in \{d, d-1\}\big) = \P\big(S(d-1)=s_{d-1} \big| D \in \{d, d-1\}\big) \\ 
        & \hspace{10mm} = \P\big(S(d-1)=s_{d-1} \big| D=d-1 \big) = \P\big(S=s_{d-1} \big| D=d-1 \big),
    \end{align*}
    where the first equality holds by the law of total probability; the second equality holds by Assumption \ref{ass:no-sorting-on-D}; and the last equality holds by writing potential outcomes in terms of their observed counterparts. Both results fulfill the requirements of Proposition \ref{prop:marginal-decomposition}.
    
    Lastly, we demonstrate that the assumptions imply that there are no weights on incongruent comparisons. Denote $m(s_d,s_{d-1}) := \E[Y|S=s_d] - \E[Y|S=s_{d-1}]$. This implies that the marginal contrast decomposition becomes:
    \begin{align*}
        \E[Y|D=d] &- \E[Y|D=d-1] = \sum_{(s_{d-1}, s_{d}) \in \mathcal{M}^+(d) } w(s_d,s_{d-1}) \Big( \E[Y|S=s_d] - \E[Y|S=s_{d-1}]\Big) \\
        & \hspace{40mm} +\sum_{(s_{d-1}, s_{d}) \in \mathcal{M}^-(d) } w(s_d,s_{d-1}) \Big( \E[Y|S=s_d] - \E[Y|S=s_{d-1}]\Big) \\ 
        &= \sum_{(s_{d-1}, s_{d}) \in \mathcal{M}^+(d) } w(s_d,s_{d-1}) \cdot m(s_d, s_{d-1}) + \sum_{(s_{d-1}, s_{d}) \in \mathcal{M}^-(d) } w(s_d,s_{d-1}) \cdot m(s_d, s_{d-1}) \\
        &= \sum_{(s_{d-1}, s_{d}) \in \mathcal{M}^+(d) } \P\big(S(d)=s_d,S(d-1)=s_{d-1} \big| D \in \{d, d-1\}\big) \cdot m(s_d, s_{d-1})
    \end{align*}
    where the first equality holds by Proposition \ref{prop:marginal-decomposition}; the second equality holds by definition of $m(s_d,s_{d-1})$; and the third equality is implied by the condition on $F(S(d),S(d-1))$. Thus, no incongruent comparisons exist in the aggregate marginal contrast as required.
\end{proof}

\subsubsection{Proofs of Results from Section \ref{subsec:settings-with-guaranteed-incongruency}}
\begin{proof}[\textbf{Proof of Proposition \ref{prop:decreasing-means-implies-incongruency}}]
    We prove the result by contrapositive. Consequently, this proof shows that sets of weights which are non-zero for only congruent pairs of sub-treatment vectors must imply that for any sub-treatment $k$, and all $d \in \D_{>0}$, $\E[S_k | D=d] \; \geq \; \E[S_k | D=d-1]$.  

    For any sub-treatment indexed by $k \in \{1, \ldots, K\}$, and any $d \in \D$, see that we may write the difference in conditional means of sub-treatment $S_k$ for adjacent values of aggregated treatment $D$ as: 
    \begin{align}
        & \E[S_k | D=d] - \E[S_k | D=d-1] \notag \\ 
        &= \E[ \; \E[S_k |S_d=s_d, D=d] \;| D=d] - \E[ \; \E[S_k |S_{d-1}=s_{d-1}, D=d-1] \;| D=d-1] \notag \\
        &= \sum_{s_d \in \ST_d} \P(S=s_d | D=d) \cdot \E[S_k |S_d=s_d, D=d] \notag \\
        & \hspace{10mm} - \sum_{s_{d-1} \in \ST_{d-1}} \P(S=s_{d-1} | D=d-1) \cdot \E[S_k |S_d=s_{d-1}, D=d-1] \notag \\
        &= \sum_{s_d \in \ST_d} \P(S=s_d | D=d) \cdot g_k(s_d) \notag \\
        & \hspace{10mm} - \sum_{s_{d-1} \in \ST_{d-1}} \P(S=s_{d-1} | D=d-1) \cdot g_k(s_{d-1}) \notag \\
        &= \sum_{s_d \in \ST_d} \P(S=s_d | D=d) \cdot \Big( \sum_{s_{d-1} \in \ST_{d-1}} \P(S=s_{d-1}|D=d-1) \Big) \cdot g_k(s_d) \notag \\
        & \hspace{10mm} - \sum_{s_{d-1} \in \ST_{d-1}} \Big( \sum_{s_{d} \in \ST_{d}} \P(S=s_{d}|D=d) \Big) \cdot \P(S=s_{d-1} | D=d-1) \cdot g_k(s_{d-1}) \notag \\ 
        &= \sum_{s_{d} \in \ST_{d}} \sum_{s_{d} \in \ST_{d}} w(s_d,s_{d-1}) \cdot \Big[ g_k(s_d) - g_k(s_{d-1}) \Big] \label{eqn:3.3A} 
    \end{align}
    where the first equality holds by the law of iterated expectations on $S_d$; the second equality holds by the definition of conditional probability; the third equality holds by defining $g_k(s_d) := \E[S_k \big| S_d=s_d, D=d]$ 
    as the value of the $k^{th}$ sub-treatment given the vector of sub-treatment $s_d$ which aggregates to $D=d$; and the fourth equality holds by the second axiom of probability (unity); the fifth equality follows by algebra and by the definition of the class of weighting functions with the properties outlined in Proposition \ref{prop:marginal-decomposition}.

    In Equation \eqref{eqn:3.3A}, we defined $g_k(s_d)$ as the value of the $k^{th}$ sub-treatment when the vector of sub-treatments was $s_d$. Notice that this is a simple constant, the realized value of $S_k$, which lies at the $k^{th}$  component of the vector $s_d$: 
    \begin{align}
        g_k(s_d) &= \E[S_k|S=s_d, D=d] = \E[S_k|S=s_d] = \sum_{s_k \in \ST_k} \P(S_k=s_k | S=s_d) \cdot s_k \notag \\ 
        &= \P(S_k=s_k | S=s_d) \cdot s_k = s_k \label{eqn:3.3B1} 
    \end{align}
    where the first equality holds by definition of $g_k(s_d)$; the second equality holds since $d$ is fully determined by $s_d$; the third equality holds by the definition of expectation; the fourth equality holds because the only realization that $S_k$ can take is the value of $S_k$ from the sub-treatment vector $s_d$ found in the conditioning set; and the fifth equality holds since $\P(S_k=s_k | S=s_d)=1$ due to the previous fact.

    Now, given the result in Equation \eqref{eqn:3.3B1}, see that if a pair of sub-treatment vectors is congruent, $(s_d,s_{d-1}) \in \mathcal{M}^{+}(d)$, then
    \begin{equation}
        g_k(s_d) - g(s_{d-1}) = s_{d,k} - s_{d-1,k} = \begin{cases}
            1, \text {if } s_{d,k} > s_{d-1,k} \\ 
            0, \text{ if } s_{d,k} = s_{d-1,k}
        \end{cases}
        \label{eqn:3.3B2} 
    \end{equation}
    which holds by the definition of congruency. Any two sub-treatment vectors which can be written as $s_d = s_{d-1} + 1_k$, where $1_k$ is a unit vector for the $k^{th}$ sub-treatment, are considered congruent and lie in the set $\mathcal{M}^{+}(d)$ at any $d \in \D$. 
    
    Since the weights in Equation \eqref{eqn:3.3A} are the same as in Proposition \ref{prop:marginal-decomposition}, then they are positive and sum to one. This, together with Equation \eqref{eqn:3.3B2}, implies that if weights were only non-negative for congruent sub-treatment vectors, then \eqref{eqn:3.3A} would be non-negative. Thus, take any weighting function which satisfies the properties in Proposition \ref{prop:marginal-decomposition} and puts non-negative weight on congruent sub-treatment vectors only, and denote it as $w^{+}(\cdot,\cdot)$. From \eqref{eqn:3.3A}, for any sub-treatment $k$, we have: 
    \begin{align}
        \sum_{s_{d} \in \ST_{d}} \sum_{s_{d} \in \ST_{d}} w(s_d,s_{d-1}) \cdot \Big[ g_k(s_d) - g_k(s_{d-1}) \Big] &= \sum_{s_{d} \in \ST_{d}} \sum_{s_{d} \in \ST_{d}} w^{+}(s_d,s_{d-1}) \cdot \Big[ g_k(s_d) - g_k(s_{d-1}) \Big] \notag \\
        &= \sum_{(s_d,s_{d-1}) \in \mathcal{M}^{+}(d)} w^{+}(s_d,s_{d-1}) \cdot 1\Big\{ g_k(s_d) > g_k(s_{d-1}) \Big\} \notag \\ 
        &\geq 0 \label{eqn:3.3C} 
    \end{align}
    where the first line holds since only positive weights are placed on the congruent comparisons of sub-treatments, denoted by $w^{+}(\cdot,\cdot)$; the second equality holds by the result in \eqref{eqn:3.3B2}, since differences in sub-treatment $k$ must be either one or zero for congruent vectors of sub-treatments; and the inequality in the third line holds because the weights must be non-negative by Proposition \ref{prop:marginal-decomposition}.

    Lastly, the inequality in \eqref{eqn:3.3C} implies that, for any weighting scheme that applies positive weight to only congruent pairs of sub-treatments, the difference in conditional means for any sub-treatment $k$ in \eqref{eqn:3.3A} must satisfy 
    \begin{align*}
        \E[S_k | D=d] - \E[S_k | D=d-1] \geq 0 &\iff \E[S_k | D=d] \geq \E[S_k | D=d-1], 
    \end{align*}
    which states the claim, as desired. As $d$ increases, the means of sub-treatments must not decline in order to rationalize fully congruent weighting schemes. Violations of this inequality signal that weight on incongruent sub-treatments is inevitable. 
\end{proof}


\subsection{Proofs of Results from Section \ref{subsec:identification-for-nonmarginal-parameters}}

\begin{proposition} \label{prop:datt-identification} 
Under Assumption \ref{ass:sutva2}, the difference in mean outcomes for the group that experienced sub-treatment $s_d$ relative to the untreated group can be decomposed as 
\begin{align*}
    \E[Y|S=s_d] - \E[Y|S=0_K] &= \mathrm{ATT}(s_d) + UB(s_d)
\end{align*}
where
\begin{align*}
    UB(s_d) &:= \E[Y(0)|S=s_d] - \E[Y(0) | S=0_K]
\end{align*}
which is a selection bias term.  If, in addition, Assumption \ref{ass:unconfoundedness} holds, then $UB(s_d) = 0$, so that
\begin{align*}
    \mathrm{ATT}(s_d) = \E[Y|S = s_d] - \E[Y|S=0_K]
\end{align*}
where $0_K$ denotes the zero vector of length $K$ for the untreated group.
\end{proposition}
Proposition \ref{prop:datt-identification} shows that $\textrm{ATT}(s_d)$ is identified if the sub-treatments are observed and under Assumption \ref{ass:unconfoundedness} (or Assumption \ref{ass:untreated-potential-outcomes}).

\begin{proof}[\textbf{Proof of Proposition \ref{prop:datt-identification}}]
    For the first part, notice that
    \begin{align*}
        \E[Y|S=s_d] - \E[Y|S=0_K] &= \E[Y(s_d)|S=s_d] - \E[Y(0)|S=0_K] \\
        &= \E[Y(s_d) - Y(0)|S=s_d] + \E[Y(0)|S=s_d] - \E[Y(0)|S=0_K] \\
        &= \textrm{ATT}(s_d) + UB(s_d)
    \end{align*}
    where the first equality holds by writing observed outcomes as potential outcomes; the second equality holds by adding and subtracting $\E[Y(0)|S=s_d]$; and the last equality holds by the definitions of $\textrm{ATT}(s_d)$ and $UB(s_d)$. For the second part of the result, it follows immediately from Assumption \ref{ass:unconfoundedness} (or Assumption \ref{ass:untreated-potential-outcomes}) that $UB(s_d)=0$, which implies the result. 
\end{proof}

\bigskip

\begin{proof}[\textbf{Proof of Theorem \ref{thm:att-identification}}]
    We proved the first part of the result as Proposition \ref{prop:datt-identification}.  For the second part, starting from the definition of $\textrm{ATT}(d)$, we have that
    \begin{align*}
        \textrm{AATT}(d) &= \E\big[\textrm{ATT}(S)\big|D=d\big] \\
        &= \sum_{s_d \in \mathcal{S}_d} \textrm{ATT}(s_d) \cdot \P(S=s_d|D=d) \\
        &= \sum_{s_d \in \mathcal{S}_d} \Big(\E[Y|S=s_d] - \E[Y|S=0_K]\Big) \cdot \P(S=s_d|D=d) \\
        &= \E[Y|D=d] - \E[Y|D=0]
    \end{align*}
    where the first equality holds from the definition of $\textrm{AATT}(d)$; the second equality expands the expectation; the third equality holds by Proposition \ref{prop:datt-identification}; and the last equality holds by the law of iterated expectations.
\end{proof}



\end{document}